\documentclass[12pt, a4paper, onecolumn]{belarticle2}
\usepackage{amsmath}
\usepackage{amssymb}
\usepackage{epsfig}
\usepackage{textcomp}
\usepackage{amsthm, dsfont}
\usepackage{footnote}
\usepackage{fullpage, mathrsfs}
\usepackage{multicol}
\usepackage{geometry, fullpage, stackrel, subfigure}

\usepackage{tikz}
\usetikzlibrary{calc,decorations.pathreplacing}
\usetikzlibrary{arrows,shapes}
\usetikzlibrary{calc,decorations.pathreplacing}

\let\OLDthebibliography\thebibliography
\renewcommand\thebibliography[1]{
  \OLDthebibliography{#1}
  \setlength{\parskip}{0pt}
  \setlength{\itemsep}{0pt plus 0.3ex}
}

\newcommand{\tr}[1]{\mathrm{tr}\left\{ #1 \right\}}
\newcommand{\rmA}{\textrm{A}}

\newcommand{\Tr}{\mathrm{tr}}
\newcommand{\ket}[1]{\left| #1 \right>} 
\newcommand{\bra}[1]{\left< #1 \right|} 
\newcommand{\id}{\mathds{1}}

\newcommand{\cH}{\mathcal{H}}

\newtheorem{theo}{Theorem}
\newtheorem{thm}[theo]{Theorem}
\newtheorem*{thm*}{Theorem}
\newtheorem{prop}[theo]{Proposition}
\newtheorem{lemma}[theo]{Lemma}
\newtheorem*{lemma*}{Lemma}
\newtheorem{lem}[theo]{Lemma}
\newtheorem{cor}[theo]{Corollary}
\newtheorem*{cor*}{Corollary}

\newtheorem{defn}[theo]{Definition}

\theoremstyle{definition} 

\begin{document}
\title {A channel-based framework for steering, \\ non-locality and beyond}
\author{Matty J. Hoban}
\affiliation{Clarendon Laboratory, Department of Physics, University of Oxford, UK}
\author{Ana Bel\'en Sainz}
\affiliation{Perimeter Institute for Theoretical Physics, 31 Caroline St. N, Waterloo, Ontario, N2L 2Y5, Canada}
\date{\today}
\maketitle
\begin{abstract}
Non-locality and steering are both non-classical phenomena witnessed in Nature as a result of quantum entanglement. It is now well-established that one can study non-locality independently of the formalism of quantum mechanics, in the so-called device-independent framework. With regards to steering, although one cannot study it completely independently of the quantum formalism, ``post-quantum steering" has been described, that is steering which cannot be reproduced by measurements on entangled states but do not lead to superluminal signalling. In this work we present a framework based on the study of quantum channels in which one can study steering (and non-locality) in quantum theory and beyond. In this framework, we show that kinds of steering, whether quantum or post-quantum, are directly related to particular families of quantum channels that have been previously introduced by Beckman, Gottesman, Nielsen, and Preskill [Phys. Rev. A 64, 052309 (2001)]. Utilising this connection we also demonstrate new analytical examples of post-quantum steering, give a quantum channel interpretation of almost quantum non-locality and steering, easily recover and generalise the celebrated Gisin-Hughston-Jozsa-Wootters theorem, and initiate the study of post-quantum Buscemi non-locality and non-classical teleportation. In this way, we see post-quantum non-locality and steering as just two aspects of a more general phenomenon. 
\end{abstract}

\tableofcontents

\bigskip

Entanglement is one of the most striking non-classical features of quantum mechanics. Given appropriately chosen measurements certain, but not all, entangled states can exhibit a violation of local realism (local causality), called ``non-locality" \cite{bell}. Apart from its fundamental interest, non-locality has also turned into a key resource for certain information-theoretic tasks, such as key distribution \cite{qkd} or certified quantum randomness generation \cite{rand}, and has been witnessed experimentally in a loophole-free manner \cite{lfb1, lfb2, lfb3}.

The non-classical implications of entanglement also manifest as a phenomenon called ``Einstein-Podolsky-Rosen steering'', henceforth referred to as solely ``steering''. There, one party, Alice, by performing appropriately chosen measurements on one half of an entangled state, remotely `steers' the states held by a distant party, Bob, in a way which has no local explanation \cite{sch}. A modern approach to steering describes it as a way to certify entanglement in cryptographic situations where some devices in the protocol are not characterised \cite{S35}. Steering hence allows for a ``one-sided device independent" implementation of several information-theoretic tasks, such as quantum key distribution \cite{sqkd}, randomness certification \cite{sr1,sr2}, measurement incompatibility certification \cite{smi1, smi2, smi3}, and self-testing of quantum states \cite{sst1, sst2}. 

Even though these phenomena arise naturally within quantum mechanics, they are not restricted to it. Non-local correlations and steering beyond what quantum theory allows are conceivable while still complying with natural physical assumptions, such as relativistic causality \cite{PR, pqsp}. By ``post-quantum'' we mean non-locality or steering that cannot be realised with local measurements made on an entangled quantum state\footnote{We do not mean post-quantum in the sense of post-quantum cryptography, where one designs cryptographic protocols that cannot be efficiently broken by quantum computers. Post-quantum in our sense could refer to non-locality and steering in generalised probabilistic theories.}. Post-quantum non-locality has been vastly explored, especially its implications in information-theoretic tasks \cite{pples}. Little is known, however, about post-quantum steering, mainly due to the lack of a clear formalism for studying this phenomenon beyond quantum theory. 

{It may be unclear why one would be interested in steering in theories beyond quantum theory, since it is a phenomenon that is defined within the quantum formalism. Indeed, if we are testing quantum theory against all possible, sensible classical descriptions of reality, a local hidden variable is the most general starting point.}
{One may however ask in which sensible ways Nature may differ from a world described by quantum theory. Here we argue that it makes sense to consider the picture where locally in our own laboratory everything is described according to quantum theory, however, the global process governing the interactions between laboratories is not, analogous to the study of indefinite causal order in Ref. \cite{ognyan}. The existence of post-quantum steering demonstrates that the global theory can deviate from quantum theory in intriguing ways, even if our own laboratory is restricted to quantum theory. In fact, because of this, we would argue that post-quantum steering is of more foundational interest than local hidden state models.} {We also note that in quantum information, bounding the set of quantum assemblages from the post-quantum set has also been studied in the guise of extended non-local games by Johnston et al \cite{johnston}.}

{To rectify the lack of a clear formalism for post-quantum steering, we} present a framework to study both non-locality and steering complying with the No-Signalling principle. Our formalism is based on quantum channels, i.e. completely positive trace preserving maps on density matrices. More specifically, we consider channels on multi-partite systems that satisfy a form of the No-Signalling principle, introduced first by Beckman, Gottesman, Nielsen, and Preskill \cite{Beck} in bipartite setups. Indeed, they defined two families of channels. On the one hand, ``causal channels", that do not permit superluminal quantum (and classical) communication between two parties. On the other, ``localizable channels", that can be described by parties sharing a quantum (entangled) state and performing local operations with respect to each party. Furthermore, the set of localizable channels is a strict sub-set of the causal channels \cite{Beck}. 

In this work, a given conditional probability distribution (correlations) in a non-locality scenario or a set of conditional quantum states (assemblage) in a steering scenario, is associated to a causal channel, and vice-versa. We identify the nature of the correlations, or assemblages, with the properties of the channels that may give rise to them. In particular, if correlations or assemblages are post-quantum then they can be associated with a causal, but not localizable, channel. {Utilising this connection we derive results in both the study of quantum channels and steering.}

We also show that our framework is not limited to the study of non-locality and steering. We show that non-locality studied from the perspective of channels can be expanded to other kinds of non-locality studied in the literature. In particular, Buscemi introduced the scenario of the semi-quantum non-local games \cite{Buscemi}, in which we can demonstrate a form of non-locality, denoted as ``Buscemi non-locality''. Buscemi showed that an entangled state can be used as a resource for demonstrating this form of non-locality. Here, we expand upon this original work to introduce post-quantum Buscemi non-locality, and show how it can be understood through quantum channels. Finally, we consider the analogue of steering for Buscemi non-locality, which is the study of non-classical teleportation, as initiated by Cavalcanti, Skrzypczyk, and \v{S}upi\'{c} \cite{teleportation}.

\section*{Summary of Results}

{This manuscript presents a variety of results which, to guide its more comprehensive reading, we now briefly outline. }

{First, in the study of quantum channels, we define a novel class of quantum channels called the ``almost localizable channels'' in Def. \ref{almostloc}, which are a generalisation of the set of localizable channels in \cite{Beck}. We show in Theorem \ref{stalmost} that the set of almost quantum assemblages (as defined in Ref. \cite{pqsp}) result from almost localizable channels, and almost localizable channels only give rise to \textit{almost quantum correlations}\footnote{Almost quantum correlations are defined as a particular relaxation of the set of quantum correlations in Bell scenarios. That is, the set of almost quantum correlations strictly contains those that are achievable by quantum mechanics. Almost quantum assemblages are defines as a particular relaxation of the set of quantum assemblages in steering scenarios. We revise the rigorous definition of these concepts in the next sections.} \cite{aqp} or assemblages. This is the first time that almost quantum assemblages are given a physical definition, rather than just being defined in terms of semi-definite programs.}

{Second, our framework provides a connection between the study of quantum channels and post-quantum steering, which is itself a novel observation. Starting from this connection, in Section \ref{postquantumex} we give new analytical examples of post-quantum steering constructed from non-localizable, yet causal, channels. In addition, Section \ref{constchanass} shows that a consequence of post-quantum steering is the existence of non-localizable channels that cannot be used to violate a Bell inequality through any local operations whatsoever. We moreover give a characterisation of non-signalling assemblages in terms of quantum states and unitary operations, which results in a diagramatic proof of the Gisin-Hughston-Jozsa-Wootters (GHJW) theorem in Corollary \ref{cor:GHJW}. We show in Section \ref{nonclassicaltel} that this proof of the GHJW theorem can be generalised to the study of non-classical teleportation, and we show in Corollary \ref{genGHJW} that post-quantum non-classical teleportation can only be witnessed if there are multiple black boxes in your network.} 

{Finally, we are the first to highlight the possibilities of studying forms of post-quantum Buscemi non-locality and post-quantum non-classical teleportation. Our framework further outlines how to approach these through the study of quantum channels.}

{The paper is structured as follows. In Section \ref{se:channels} we introduce a new family of quantum channels of utmost relevance in this work, while we review relevant known classes of channels in Appendix \ref{ap:channels}. In Sections \ref{se:chanbell} and \ref{se:chanstee} we (i) discuss the interpretation of Bell and steering scenarios in terms of quantum channels, and (iii) present some results that follow when looking at these non-classical phenomena from the scope of quantum channels. The traditional scope to these phenomena is briefly reviewed in Appendices \ref{ap:nonloc} and \ref{ap:eprstee}. 
Finally, in Section \ref{se:gennonloc} we discuss how our framework further includes the above mentioned Buscemi nonlocality \cite{Buscemi} and nonclassical teleportation \cite{teleportation}.  For clarity in the presentation, some of the proofs of results in the main body of the paper are presented in the Appendix.}

A quick note on notation. A Hilbert space will typically be denoted by $\mathcal{H}$, unless otherwise stated, and the set of positive operators acting on $\mathcal{H}$ with trace at most $1$ will be denoted as $\mathcal{D}(\mathcal{H})$. Furthermore, for the more general set of linear operators acting on $\mathcal{H}$, we will use the notation $\mathcal{L}(\mathcal{H})$.

\section{Quantum channels}\label{se:channels}

In the study of non-locality in quantum physics and beyond, a common approach is to have the fundamental objects being a black box associated with some stochastic behaviour: for a given set of inputs for each party, an output is generated stochastically. A stochastic process should be suitably normalised, i.e. the sum over all outcomes for a given input is $1$. The quantum analogue of such a process is a quantum channel. Recall that a channel $\Lambda$ is a trace-preserving, completely-positive (CPTP) map. That is, given an input quantum state described by the density matrix $\rho_{i}$, a channel $\Lambda$ acts on this system producing an output state with density matrix $\rho_{o}:=\Lambda(\rho_{i})$. The suitable normalisation condition is then that the trace of $\rho_{o}$ is $1$ whenever $\tr{\rho_i}{}=1$. A classical stochastic process can be encoded into a channel with respect to some orthonormal basis of the respective Hilbert space. To retrieve the probabilities in the stochastic process one only needs to prepare states in that basis as input, and then only measure in that basis.

Given these simple observations, one can readily relate quantum channels to the study of conditional probability distributions, and thus quantum non-locality. For example, we can ask which channels give rise to correlations that are compatible with a local hidden variable model, or otherwise. {Such non-local properties of quantum channels have been observed and utilised in previous works \cite{Beck,Pao}. There, the relevant objects of study are semicausal and causal multipartite quantum channels, in particular the subset of localizable ones, which we formally review in Appendix \ref{ap:channels}. {To sketch their definitions now, the causal channels are those where one party's output quantum state is the same for all input states for another party, and the localizable channels are those that are generated by local operations and shared entanglement between the parties.} In this section we introduce a new class of channels, called the \textit{almost localizable channels}, which will be pertinent when discussing non-locality and steering.} 

The general scenario we consider is that of multiple space-like separated parties such that they cannot use any particular physical system in their respective laboratories that could result in communication. In this way, the parties are subjected to the same conditions as in a Bell test. We can model the parties' global resources as a device with multiple input and output ports: an input and output port associated with each party. Therefore, each party can produce a local input quantum system, put it into their respective input port, and receive a quantum system from the output port. The global device can contain resources that are shared between distant parties, such as entanglement. For example, if we have two parties, and they each input a system into their respective devices, the output of both devices could be associated with an entangled quantum state. We will now make this picture more formal. 

We have $N$ parties labelled by an index $j\in\{1,2,...,N\}$, and each party has an input and output Hilbert space, $\mathcal{H}^{j}_{in}$ and $\mathcal{H}^{j}_{out}$ respectively, associated with the input and output ports of the parties' device\footnote{In this paper, all Hilbert spaces are assumed to be finite dimensional, unless otherwise specified.}. The input quantum systems have states that are associated with the density matrix $\rho^{j}_{in}\in\mathcal{D}(\mathcal{H}^{j}_{in})$. The $N$-partite device is then associated with a quantum channel $\Lambda_{1...N}:\bigotimes_{j}^{N}\mathcal{L}(\mathcal{H}^{j}_{in})\rightarrow\bigotimes_{j}^{N}\mathcal{L}(\mathcal{H}^{j}_{out})$ taking the input state $\rho^{1}_{in}\otimes\rho^{2}_{in}\otimes...\otimes\rho^{N}_{in}$ to $\rho_{out}:=\Lambda_{1...N}(\rho^{1}_{in}\otimes\rho^{2}_{in}\otimes...\otimes\rho^{N}_{in})$. It will be convenient at times to take bipartitions $S_{A}$, $S_{B}\subseteq\{1, \ldots, N\}$ of the $N$ parties, such that $S_A \cup S_B = \{1, \ldots, N\}$. With these bipartitions, we can then consider Hilbert spaces $\mathcal{H}^{S_{A}}_{in}$, $\mathcal{H}^{S_{B}}_{in}$ and $\mathcal{H}^{S_{A}}_{out}$, $\mathcal{H}^{S_{B}}_{out}$ associated with the input and output Hilbert spaces of $S_{A}$ and $S_{B}$ respectively.

{Given this set-up, we can informally sketch the definition of semicausal and localisable channels in the bipartite case (i.e. $N=2$). The formal definitions can be found in Appendix \ref{ap:channels}. A \textit{semicausal} channel is one where the output state for one particular party is independent of the input state of the other party. In other words, the reduced quantum state for one of the parties is well-defined since it is independent of the other party's input. For example, if a channel $\Lambda_{12}$ is semicausal from $1$ to $2$, denoted $1 \not\rightarrow 2$, then the output state is $\rho_{out}=\Lambda_{12}(\rho_{in}^{1}\otimes\rho_{in}^{2})$ and if we trace out party $1$, the output state of party $2$ is $\rho_{out}^{2}=\textrm{tr}_{1}(\Lambda_{12}(\rho_{in}^{1}\otimes\rho_{in}^{2}))=\Upsilon(\rho_{in}^{2})$, for $\Upsilon:\mathcal{L}(\mathcal{H}^{2}_{in})\rightarrow\mathcal{L}(\mathcal{H}^{2}_{out})$ being a channel. A bipartite channel is \textit{causal} if it is semicausal in both directions, i.e. $1\not\rightarrow 2$ and $2\not\rightarrow 1$. A bipartite channel is \textit{localisable} if there exists a joint quantum state shared between the two parties such all the parties' maps are only from the $j$th party's input and their share of the entangled state to the $j$th party's output.}

{In this work, we will use diagrammatic representations of quantum channels where input and output systems to a channel are represented by wires, and the channels as boxes connecting inputs and outputs. One can see an example of such a digram in Fig. \ref{f:bellmap}, where $\Lambda$ is the channel, and time (the flow from inputs to outputs) goes from bottom to top. Furthermore, later on, we will denote the preparation of states as triangles at the beginning of input wires, and measurements as triangles at the end of output wires.}


{Within this scenario we define a new class of channels called the \textit{almost localizable channels} as follows: 


\begin{defn}\label{almostloc}  \textbf{Almost localizable channels.} \\
A causal channel $\Lambda_{1 \ldots N}$ is almost localizable if and only if there exists {a global ancilla system $E$ with Hilbert space $\mathcal{H}_{E}$, a local ancilla system $E_{k}$ for each $k$th party, with input and output Hilbert spaces $\cH_{in}^{E_k}$ and $\cH_{out}^{E_k}$ respectively, and state $|\psi\rangle_{E}\in\mathcal{H}_{E}\otimes\cH_{in}^{E_1}\otimes\cH_{in}^{E_2}\otimes...\otimes\cH_{in}^{E_N}$ such that, for all states $\rho\in\mathcal{D}(\bigotimes_{j=1}^{N}\mathcal{H}_{in}^{j})$,}
\begin{equation*}
\Lambda_{1 \ldots N} [\rho] = \Tr_{EE_1 \ldots E_N}\left\{\prod_{j=1}^{N}U_{j E} \, (\rho \otimes |\psi\rangle\langle\psi|_{E}) \prod_{k=0}^{N-1}U_{N-k E}^{\dagger}\right\},
\end{equation*}
where $U_{kE}:\mathcal{H}_{in}^{k}\otimes\mathcal{H}_{E}\otimes\cH_{in}^{E_k}\rightarrow\mathcal{H}_{out}^{k}\otimes\mathcal{H}_{E}\otimes\cH_{out}^{E_k}$ is a unitary operator for all $k$, such that, for any permutation $\pi$ on the set $\{1,2,...,N\}$, 
\begin{align*}
\prod_{j=1}^{N}U_{j E} \, (\rho \otimes |\psi\rangle\langle\psi|_{E}) \, \prod_{k=0}^{N-1}U_{(N-k) E}^{\dagger} = \prod_{j=1}^{N}U_{\pi(j) E} \, (\rho \otimes |\psi\rangle\langle\psi|_{E}) \, \prod_{k=0}^{N-1}U_{\pi(N-k) E}^{\dagger}\,.
\end{align*}
\end{defn}

{Notice that in this definition for localizable channels, the ancilla $\sigma_{R}$ is the same for all inputs to the channel $\Lambda_{1 \ldots N}$.} If we compare this definition with that of localizable channels as given by Def.~\ref{def:locuni}, we see that almost localizable channels are a natural generalisation of the localizable ones. Indeed, in Def.~\ref{def:locuni}, the condition of the representation that for all permutations $\pi$, $\prod_{k=1}^{N}U_{k E}=\prod_{k=1}^{N}U_{\pi({k}) E}$ is equivalent to the constraint that $\prod_{k=1}^{N}U_{k E}|\psi\rangle=\prod_{k=1}^{N}U_{\pi({k}) E}|\psi\rangle$ for all possible states $|\psi\rangle\in\mathcal{H}_{in}^{E}\bigotimes_{j=1}^{N}\mathcal{H}_{in}^{j}$. This last universal quantifier over all ancilla states can be relaxed further to an existential quantifier, i.e. that there exists a state $|\psi\rangle$ such that the unitary operators' ordering is invariant under permutations of the parties. This relaxation precisely gives the set of almost localizable channels.

Note that localizable channels are by definition almost localizable, as well as causal. However, as we will show in Section \ref{construct}, there exist almost localizable channels that are not localizable. In showing this, we use the close connection between the so-called \textit{almost quantum correlations} defined in \cite{aqp} (see appendix \ref{ap:nonloc}), and the almost localizable channels. Indeed, the motivation for the name almost localizable comes from this connection. In this direction we also generalise this connection to the study of steering in Section \ref{se:substee}.}

\section{Non-locality from the scope of quantum channels}\label{se:chanbell}

In this section, we reinterpret the traditional Bell scenario \cite{bell} in terms of quantum channels. In particular, we connect every quantum channel to a family of correlations in a Bell test.  {We emphasize that non-locality can, in a sense, be studied independently of the quantum formalism, so considering all processes as fundamentally quantum may seem \textit{excessive}. Instead, one can see our review of non-locality from the point-of-view of quantum channels as just the beginning of a bigger story, as will hopefully become clear.} We review the traditional notion of a Bell scenario and its relevant sets of correlations in appendix \ref{ap:nonloc}.

\subsection{Non-locality via quantum channels}\label{se:sunclc}

{A Bell scenario is characterised by the parameters $(N,m,d)$, where $N$ denotes the number of parties, $m$ the number of measurements each party can choose from, each with $d$ possible outcomes. Consider now the parties to have input and output Hilbert spaces given by $\cH^{j}_{in} = \cH^{j'}_{in} = \cH_m$ and $\cH^{j}_{out} = \cH^{j'}_{out} = \cH_d$ for all $j\neq j'$, where $\cH_m$ has dimension $m$ and $\cH_d$ dimension $d$. Denote by $\{\ket{x}\}_{x=1:m}$ an orthonormal basis of $\cH_m$, and by $\{\ket{a}\}_{a=1:d}$ an orthonormal basis of $\cH_d$. 
In what follows, we relate channels of the form $\Lambda_{1...N}:\mathcal{L}(\mathcal{H}_{m}^{\otimes N})\rightarrow\mathcal{L}(\mathcal{H}_{d}^{\otimes N})$ to correlations in a Bell scenario.}

\begin{defn}
A conditional probability distribution $p(a_1 ... a_N | x_1 ... x_N)$ in a Bell scenario is \textbf{channel-defined} if there exists a channel $\Lambda_{1 ... N}:\mathcal{L}(\mathcal{H}_{m}^{\otimes N})\rightarrow\mathcal{L}(\mathcal{H}_{d}^{\otimes N})$, and some choice of orthonormal bases $\{\ket{x_{j}}\in\mathcal{H}_{m}\}_{x=1:m}$ and $\{\ket{a_{j}}\in\mathcal{H}_{d}\}_{a=1:d}$ for each $j$th party, such that
\begin{equation}\label{eq:NLequiv}
p(a_1 ... a_N | x_1 ... x_N) = \tr{\otimes_{k=1}^N \ket{a_k}\bra{a_k} \, \Lambda_{1 ... N} \left(\otimes_{k=1}^N \ket{x_k}\bra{x_k}\right) }\,.
\end{equation} 
\end{defn}

Given a channel $\Lambda_{1...N}$, it is always possible to define correlations resulting from it for a given choice of input and output orthonormal bases. Figure \ref{f:bellmap} sketches (in the bipartite case) this construction of correlations schematically. 
{Given this connection, we can now directly relate the families of correlations presented earlier to families of channels presented in Section \ref{se:channels} and Appendix \ref{ap:channels}. Although the results pertinent to non-signalling, quantum and classical correlations were noticed in previous works \cite{Beck, Piani, Pao}, we present all proofs in Appendix \ref{appsec21}.}

\begin{prop}\label{probNSprop} 
A conditional probability distribution $p(a_{1}...a_{N}|x_{1}...x_{N})$ is non-signalling if and only if there exists a causal channel $\Lambda_{1...N}^{\mathsf{C}}:\mathcal{L}(\mathcal{H}_{m}^{\otimes N})\rightarrow\mathcal{L}(\mathcal{H}_{d}^{\otimes N})$ such that the distribution is channel-defined by $\Lambda_{1...N}^{\mathsf{C}}$.
\end{prop}

In Fig.~\ref{f:PRmap} we show how the example of a Popescu-Rohrlich (PR) non-local box can be realised by a causal channel. {The PR non-local box is a device that can violate the Clauser-Horne-Shimony-Holt inequality beyond Tsirelson's bound, and thus cannot be realised by local measurements on an entangled state \cite{PR}. The statistics produced by a PR box, for binary inputs and outputs, are $p(a,b|x,y)=\frac{1}{2}\delta_{a\oplus b}^{xy}$, where $\oplus$ is addition modulo 2.} The channel in this figure is an entanglement-breaking channel \cite{Piani}, and thus its Choi state $\Omega$ is separable across the partition of Alice's and Bob's input and output Hilbert spaces. However, non-localizable causal channels that are not entanglement-breaking have been constructed in the literature \cite{Pao}, and we will refer to one such channel later. How can one detect non-localizability in a particular channel? One possible approach is through the correlations that are channel-defined by that channel, as described in the following result.
\begin{figure}
\begin{center}
\begin{tikzpicture}
\node at (-7.5,0) {$p(a,b|x,y)$};
\node at (-6,0) {$\equiv$};

\shade[draw, thick, ,rounded corners, inner color=white,outer color=gray!50!white] (-5,-1) rectangle (-1,1) ;
\node at (-3,0) {$\Lambda$};
\draw[thick] (-4,-1.5) -- (-4,-1);
\draw[thick] (-2,-1.5) -- (-2,-1);
\draw[thick] (-4,1) -- (-4,1.5);
\draw[thick] (-2,1) -- (-2,1.5);

\draw[thick, rounded corners] (-4,-2) -- (-3.5,-1.5) -- (-4.5, -1.5) -- cycle;
\node at (-4,-1.7) {$x$};
\draw[thick, rounded corners] (-2,-2) -- (-1.5,-1.5) -- (-2.5, -1.5) -- cycle;
\node at (-2,-1.7) {$y$};

\draw[thick, rounded corners] (-4,2) -- (-3.5,1.5) -- (-4.5, 1.5) -- cycle;
\node at (-4,1.7) {$a$};
\draw[thick, rounded corners] (-2,2) -- (-1.5,1.5) -- (-2.5, 1.5) -- cycle;
\node at (-2,1.7) {$b$};

\node at (-4,2.5) {Alice};
\node at (-2,2.5) {Bob};
\end{tikzpicture}
\end{center}
\caption{A conditional probability distribution $p(a,b|x,y)$ resulting from a quantum channel $\Lambda$.}
\label{f:bellmap}
\end{figure}
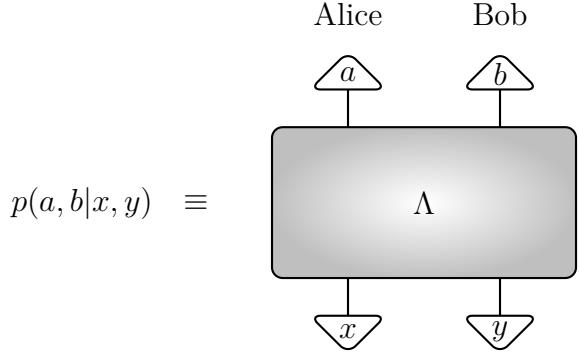
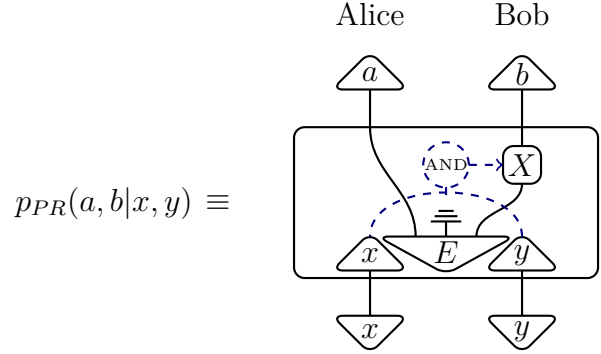
\begin{figure}
\begin{center}
\begin{tikzpicture}
\node at (-7.5,0) {$p_{PR}(a,b|x,y)$};
\node at (-6,0) {$\equiv$};

\draw[thick, ,rounded corners] (-5,-1) rectangle (-1,1) ;
\draw[thick] (-4,-1.5) -- (-4,-0.9);
\draw[thick] (-2,-1.5) -- (-2,-0.9);
\draw[thick] (-4,1) -- (-4,1.5);
\draw[thick] (-2,0.75) -- (-2,1.5);

\draw[thick, rounded corners] (-4,-2) -- (-3.5,-1.5) -- (-4.5, -1.5) -- cycle;
\node at (-4,-1.7) {$x$};
\draw[thick, rounded corners] (-2,-2) -- (-1.5,-1.5) -- (-2.5, -1.5) -- cycle;
\node at (-2,-1.7) {$y$};

\draw[thick, rounded corners] (-4,2) -- (-3.5,1.5) -- (-4.5, 1.5) -- cycle;
\node at (-4,1.7) {$a$};
\draw[thick, rounded corners] (-2,2) -- (-1.5,1.5) -- (-2.5, 1.5) -- cycle;
\node at (-2,1.7) {$b$};

\node at (-4,2.5) {Alice};
\node at (-2,2.5) {Bob};

\draw[thick, rounded corners, xshift=-6cm] (3,-0.95) -- (3.9,-0.45) -- (2.1, -0.45) -- cycle;
\node at (-3,-0.66) {$E$};

\draw[thick] (-3,-0.45) -- (-3,-0.25);
\draw[thick, xshift=-6cm, yshift=-1.1cm] (2.9,0.99) -- (3.1,0.99);
\draw[thick, xshift=-6cm, yshift=-1.1cm] (2.85,0.92) -- (3.15,0.92);
\draw[thick, xshift=-6cm, yshift=-1.1cm] (2.8,0.85) -- (3.2,0.85);

\draw[thick] (-3.4,-0.45) to [out=90,in=-90] (-4,1);
\draw[thick, rounded corners, yshift=-2.4cm] (-4,2) -- (-3.5,1.5) -- (-4.5, 1.5) -- cycle;
\node at (-4,-0.7) {$x$};
\draw[thick, rounded corners, yshift=-2.4cm] (-2,2) -- (-1.5,1.5) -- (-2.5, 1.5) -- cycle;
\node at (-2,-0.7) {$y$};

\draw[thick, rounded corners] (-2.25, 0.25) rectangle (-1.75, 0.75 );
\node at (-2, 0.5) {$X$};

\draw[thick] (-2.6,-0.45) to [out=90,in=-90] (-2,0.25);

\draw[thick, dashed, color=blue!50!black] (-2,-0.45) to [out=90,in=90] (-4,-0.45);

\node at (-3, 0.5) {\tiny{AND}};
\draw[thick, dashed, color=blue!50!black] (-3,0.5) circle [radius=0.3];
\draw[thick, dashed, color=blue!50!black] (-3, 0.1) -- (-3, 0.24);
\draw[thick, dashed, color=blue!50!black, ->] (-2.7, 0.5) -- (-2.25, 0.5);

\end{tikzpicture}
\end{center}
\caption{A causal channel that generates PR box correlations, as shown in Ref. \cite{Beck}. First, the inputs $\ket{x}$ and $\ket{y}$ are measured on the computational basis, obtaining outcomes $x$ and $y$. In addition, a bipartite ancilla state $\rho_{E} = \tfrac{1}{2} \left(\ket{00}\bra{00} + \ket{11}\bra{11} \right)$ is generated by preparing the pure state $\tfrac{1}{\sqrt{2}} (\ket{000}+\ket{111})$ and tracing out the third system. Then, the classical outputs of the first step are compared (gray dashed lines). Whenever $x \cdot y = 1$, an $X$ gate is performed on Bob's subsystem, flipping his qubit. Finally, Alice and Bob project the output state into the computational basis, and so obtain correlations that reproduce a PR box. This whole process can made into a unitary process by replacing the initial measurements with controlled unitaries that change the state of some  ancilla depending on the input. The $AND$ gate and controlled-$X$ gates can then be replaced by a Toffoli gate to get the unitary representation of this channel. Note also that we can interchange Alice and Bob's operations to get another causal channel that gives the PR box correlations.}
\label{f:PRmap}
\end{figure}

\begin{prop}\label{probQprop} 
A conditional probability distribution $p(a_{1}...a_{N}|x_{1}...x_{N})$ is quantum if and only if there exists a localizable channel $\Lambda_{1...N}^{\mathsf{Q}}:\mathcal{L}(\mathcal{H}_{m}^{\otimes N})\rightarrow\mathcal{L}(\mathcal{H}_{d}^{\otimes N})$ such that the distribution is channel-defined by $\Lambda^{\mathsf{Q}}_{1...N}$.
\end{prop}

In Fig.~\ref{f:mapsinglet} we present the example of a localizable channel that channel-defines the correlations $p_{sing}(a,b|x,y)$ which give Tsirelson's bound for the Clauser-Horne-Shimony-Holt (CHSH) inequality \cite{chsh}, i.e. the maximal violation for local measurements on an entangled state. We present the channel in terms of its unitary representation.

\begin{figure}
\begin{center}
\begin{tikzpicture}

\node at (-7.5,0) {$p_{sing}(a,b|x,y)$};
\node at (-6,0) {$\equiv$};

\draw[thick, ,rounded corners] (-5,-1) rectangle (-1,1) ;
\draw[thick] (-4,-1.5) -- (-4,-0.9);
\draw[thick] (-2,-1.5) -- (-2,-1);
\draw[thick] (-4,1) -- (-4,1.5);
\draw[thick] (-2,0.75) -- (-2,1.5);

\draw[thick, rounded corners] (-4,-2) -- (-3.5,-1.5) -- (-4.5, -1.5) -- cycle;
\node at (-4,-1.7) {$x$};
\draw[thick, rounded corners] (-2,-2) -- (-1.5,-1.5) -- (-2.5, -1.5) -- cycle;
\node at (-2,-1.7) {$y$};

\draw[thick, rounded corners] (-4,2) -- (-3.5,1.5) -- (-4.5, 1.5) -- cycle;
\node at (-4,1.7) {$a$};
\draw[thick, rounded corners] (-2,2) -- (-1.5,1.5) -- (-2.5, 1.5) -- cycle;
\node at (-2,1.7) {$b$};

\node at (-4,2.5) {Alice};
\node at (-2,2.5) {Bob};

\draw[thick, rounded corners, xshift=-6cm] (3,-0.95) -- (3.9,-0.45) -- (2.1, -0.45) -- cycle;
\node[xshift=-6cm] at (3,-0.66) {$E$};

\draw[thick, rounded corners, xshift=-6cm] (1.5,-0.35) rectangle (2.9,0.7);
\draw[thick, rounded corners, xshift=-6cm] (3.1,-0.35) rectangle (4.5,0.7);

\draw[thick, xshift=-6cm] (2,-1.5) -- (2,-0.35);
\draw[thick, xshift=-6cm] (2,0.7) -- (2,1.5);
\draw[thick, xshift=-6cm] (4,-1.5) -- (4,-0.35);
\draw[thick, xshift=-6cm] (4,0.7) -- (4,1.5);

\draw[thick, xshift=-6cm] (2.5,-0.45) -- (2.5,-0.2);
\draw[thick, xshift=-6cm] (3.5,-0.45) -- (3.5,-0.2);

\draw[thick, xshift=-6cm] (2.5,0.7) -- (2.5,0.81);
\draw[thick, xshift=-6.5cm, yshift=-0.05cm] (2.9,0.99) -- (3.1,0.99);
\draw[thick, xshift=-6.5cm, yshift=-0.05cm] (2.85,0.92) -- (3.15,0.92);
\draw[thick, xshift=-6.5cm, yshift=-0.05cm] (2.8,0.85) -- (3.2,0.85);

\draw[thick, xshift=-6cm] (3.5,0.7) -- (3.5,0.81);
\draw[thick, xshift=-5.5cm, yshift=-0.05cm] (2.9,0.99) -- (3.1,0.99);
\draw[thick, xshift=-5.5cm, yshift=-0.05cm] (2.85,0.92) -- (3.15,0.92);
\draw[thick, xshift=-5.5cm, yshift=-0.05cm] (2.8,0.85) -- (3.2,0.85);

\draw[thick, rounded corners] (-3.75,-0.2) rectangle (-3.25,0.35);
\node at (-3.5,0.05) {$U_t$};
\draw[thick] (-4, -0.35) -- (-4, 0.05);
\node[draw,shape=circle,fill, scale=0.5] at (-4, 0.05) {};
\draw[thick] (-4, 0.05) -- (-3.75, 0.05);
\draw[thick] (-4, 0.05) to [out=90,in=-90] (-3.5, 0.7);
\draw[thick] (-3.5, 0.35) to [out=90,in=-90] (-4, 0.7);

\draw[thick, rounded corners] (-2.75,-0.2) rectangle (-2.25,0.35);
\node at (-2.5,0.05) {$V_t$};
\draw[thick] (-2, -0.35) -- (-2, 0.05);
\node[draw,shape=circle,fill, scale=0.5] at (-2, 0.05) {};
\draw[thick] (-2, 0.05) -- (-2.25, 0.05);
\draw[thick] (-2, 0.05) to [out=90,in=-90] (-2.5, 0.7);
\draw[thick] (-2.5, 0.35) to [out=90,in=-90] (-2, 0.7);

\end{tikzpicture}
\end{center}
\caption{A localizable map that generates singlet correlations that violate the CHSH inequality maximally, i.e. up to Tsirelson's bound $2\sqrt{2}$. All basis states are in the computational basis. First each party performs a unitary operation on their share of the ancilla $E$ (initialised in state $\ket{E}=\tfrac{\ket{00}+\ket{11}}{\sqrt{2}}$) controlled on their input qubits $\ket{x}$ and $\ket{y}$. The controlled unitaries inside the boxes are $U_{t} = \ket{0}\bra{0}_c \otimes H_t + \ket{1}\bra{1}_c \otimes \id _t$ and $V_{t} = \ket{0}\bra{0}_c \otimes R_{Y,t}(-\tfrac{\pi}{4}) + \ket{1}\bra{1}_c \otimes R_{Y,t}(\tfrac{\pi}{4})$, for $H$ and $R_{Y}$ being a Hadamard and a rotation about the $Y$-axis in the Bloch sphere respectively. The indices $c$ and $t$ denote the control and target qubits respectively.}
\label{f:mapsinglet}
\end{figure}
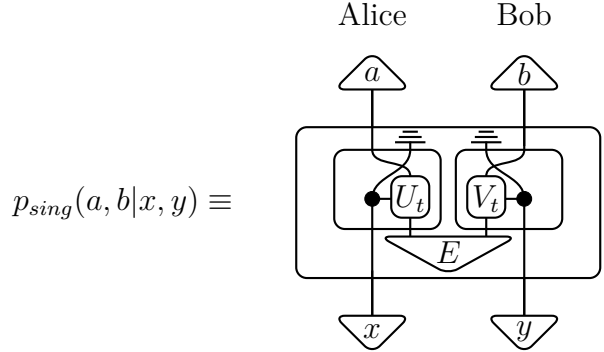

Given a channel $\Lambda$, if the correlations that are channel-defined by it are not compatible with quantum correlations, i.e. they are post-quantum correlations, then the channel was not localizable. For example, if one obtains correlations that are channel-defined by a channel $\Lambda$, and then computes their CHSH value, if this exceeds Tsirelson's bound, the channel $\Lambda$ is non-localizable. Indeed, this is how it is shown that the channel in Fig. \ref{f:PRmap} is non-localizable, as well as the channel given in Ref. \cite{Pao}.

\begin{prop}\label{probLHSprop} 
A conditional probability distribution $p(a_{1}...a_{N}|x_{1}...x_{N})$ is classical if and only if there exists a local channel $\Lambda_{1...N}^{\mathsf{L}}:\mathcal{L}(\mathcal{H}_{m}^{\otimes N})\rightarrow\mathcal{L}(\mathcal{H}_{d}^{\otimes N})$ such that the distribution is channel-defined by $\Lambda_{1...N}^{\mathsf{L}}$.
\end{prop}

It should be noted that there can exist non-local but localizable channels that will only channel-define classical correlations. A simple bipartite example of such a channel is one where the maximally entangled two-qubit state is prepared in the ancilla register, the input systems are discarded (or traced out), and each party's output is one half of the two-qubit register. For this channel, correlations are generated by each party is measuring one half of a maximally entangled state in a fixed basis, which can be reproduced by classical correlations. 

Finally, we now address the set of almost quantum correlations. We have included the proof of the following result, since it will be useful for our subsequent discussion.

\begin{prop} \label{thmalmost}
A conditional probability distribution $p(a_{1}...a_{N}|x_{1}...x_{N})$ is almost quantum if and only if there exists an almost localizable channel $\Lambda^{\tilde{\mathsf{Q}}}_{1...N}:\mathcal{L}(\mathcal{H}_{m}^{\otimes N})\rightarrow\mathcal{L}(\mathcal{H}_{d}^{\otimes N})$ such that the distribution is channel-defined by $\Lambda_{1...N}^{\tilde{\mathsf{Q}}}$.
\end{prop}

\begin{proof}
First we take a probability distribution $p(a_{1}...a_{N}|x_{1}...x_{N})$ that is in the set of almost quantum correlations.  Let $\Pi^{(i)}_{a_i|x_i}$ and $\ket{\psi}$ be the projectors and state that realise\footnote{The state $\rho$ that realises an almost quantum correlation can be taken to be a pure state, without loss of generality when there are no restrictions on the dimension of the Hilbert space.} this distribution, which have an associated Hilbert space $\mathcal{H}_{d^\prime}$ of dimension $d^\prime$. 
From these we will define an almost localizable channel $\Lambda^{\tilde{\mathsf{Q}}}_{1...N}:\mathcal{L}(\mathcal{H}_{m}^{\otimes N})\rightarrow\mathcal{L}(\mathcal{H}_{d}^{\otimes N})$ such that the correlations are channel-defined for this channel. Let the ancilla be a quantum system initialised on the state $\ket{\Psi}= \ket{\psi}\otimes  \ket{0} ^{\otimes N}$, where these extra $N$ systems are qu$d$its in the Hilbert space $\mathcal{H}_{d}$, one for each party, initialised in $\ket{0}$. Now define the following operators for each party:
\begin{align}\label{eq:theOaq}
O^{(i)}_x = \sum_{j=1}^d \Pi^{(i)}_{j|x} \otimes A_{j} \,,
\end{align}
where $A_1 = \id_{d} $ and $A_j = \ket{j}\bra{0} + \ket{0}\bra{j} + \sum_{\alpha=1, \alpha \neq j}^{d-1} \ket{\alpha}\bra{\alpha}$ for $j = 1:d-1$. The operator $\Pi^{(i)}_{j|x}$ acts on $\mathcal{H}_{d^\prime}$, $A_{j}$ on the ancillary qu$d$it that corresponds to party $i$.
Then define, for each party, controlled unitary operations on their input qu$m$it in Hilbert space $\mathcal{H}_{m}$ together with their ancillary system as follows:
\begin{align}\label{eq:AQuni}
U^{(i)} = \sum_{x=1}^m \ket{x}\bra{x} \otimes O^{(i)}_x. 
\end{align}
Now the almost localizable channel can be defined, as in Fig.~\ref{f:theAQconstr}. Each party has as input system a qu$m$it, the unitary representation of the channel is given by the $U^{(i)}$ followed by a swap on the ancilla qu$d$it and the input qu$m$it for each party. 
Finally, the output system for each party is their corresponding qu$d$it, and the input qumits and the ancillary subsystem on $\mathcal{H}_{d^\prime}$ are traced out. The commutation relations of the $ \Pi^{(i)}_{j|x}$ on the state $\ket{\psi}$ imply that the unitaries $U^{i}$ associated with the different parties commute on the ancilla, thus implying the channel is almost localizable. It is straightforward to check that the correlation $p(a_{1}...a_{N}|x_{1}...x_{N})$ is recovered by the parties inputting their measurement settings $\ket{x_i}$, and measuring their output systems in the basis $\{\ket{a_{i}}\}$. 

So far we have seen that an almost localizable channel can be constructed from almost quantum correlations. Given an almost localizable channel, it is relatively straightforward to see that the channel-defined correlations will be almost quantum. Note that the action at each $j$th party of preparing an input state, followed by a unitary and then a projective measurement can be simulated by the projective measurements $\Pi_{a_{j}|x_{j}}$ on the ancilla state $|E\rangle$ as per the definition of almost localizable channels. These projectors will then satisfy the properties required to produce almost quantum correlations given the definition of an almost localizable channel.
\end{proof}

\begin{figure}
\begin{center}
\begin{tikzpicture}
\draw[thick, rounded corners] (-0.5,-1) rectangle (4.5,1.6);
\draw[thick, rounded corners] (2,-0.99) -- (2.5,-0.49) -- (1.5, -0.49) -- cycle;
\node at (2,-0.69) {\tiny{$\ket{\psi}$}};
\draw[thick, rounded corners, xshift=-1cm] (2,-0.99) -- (2.5,-0.49) -- (1.5, -0.49) -- cycle;
\node[xshift=-1cm] at (2,-0.69) {\tiny{$\ket{0}$}};
\draw[thick, rounded corners, xshift=1cm] (2,-0.99) -- (2.5,-0.49) -- (1.5, -0.49) -- cycle;
\node[xshift=1cm] at (2,-0.69) {\tiny{$\ket{0}$}};

\draw[thick, rounded corners] (0,-0.4) rectangle (2.5,0.4);
\draw[thick, rounded corners] (1.5,0.5) rectangle (4,1.3);
\node at (1.5,-0.1) {\tiny{$O_x^{(A)}$}};
\node at (2.5,0.8) {\tiny{$O_y^{(B)}$}};

\draw[thick, rounded corners] (0.8,-0.3) rectangle (2.3,0.15);
\draw[thick, rounded corners] (1.7,0.6) rectangle (3.2,1.05);

\draw[thick] (0.5,-1.5) -- (0.5,0.15);
\draw[thick] (0.5,0.4) -- (0.5,2);
\draw[thick] (3.5,-1.5) -- (3.5,1.05);
\draw[thick] (3.5,1.3) -- (3.5,2);

\draw[thick,xshift=-1cm] (2,-0.49) -- (2,-0.3);
\draw[thick,xshift=1cm] (2,-0.49) -- (2,0.6);

\node[draw, circle, fill, scale=0.5] at (0.5,-0.1) {};
\node[draw, circle, fill, scale=0.5] at (3.5,0.8) {};
\draw[thick] (0.5,-0.1) -- (0.8,-0.1);
\draw[thick] (3.5,0.8) -- (3.2,0.8);
\draw[thick] (1,0.15) to [out=90, in=-90] (0.5,0.4);
\draw[thick] (0.5,0.15) to [out=90, in=-90] (1,0.4);

\draw[thick] (3,1.05) to [out=90, in=-90] (3.5,1.3);
\draw[thick] (3.5,1.05) to [out=90, in=-90] (3,1.3);
\draw[thick, xshift=-1cm] (2,0.4) -- (2,1.45);
\draw[thick, yshift=0.6cm, xshift=-1cm] (1.8,0.85) -- (2.2,0.85);
\draw[thick, yshift=0.6cm, xshift=-1cm] (1.85,0.9) -- (2.15,0.9);
\draw[thick, yshift=0.6cm, xshift=-1cm] (1.9,0.95) -- (2.1,0.95);
\draw[thick, xshift=1cm] (2,1.3) -- (2,1.45);
\draw[thick, yshift=0.6cm, xshift=1cm] (1.8,0.85) -- (2.2,0.85);
\draw[thick, yshift=0.6cm, xshift=1cm] (1.85,0.9) -- (2.15,0.9);
\draw[thick, yshift=0.6cm, xshift=1cm] (1.9,0.95) -- (2.1,0.95);

\draw[thick] (2,-0.49) -- (2,-0.3);
\draw[thick] (2,0.15) -- (2,0.6);
\draw[thick] (2,1.05) -- (2,1.45);
\draw[thick, yshift=0.6cm] (1.8,0.85) -- (2.2,0.85);
\draw[thick, yshift=0.6cm] (1.85,0.9) -- (2.15,0.9);
\draw[thick, yshift=0.6cm] (1.9,0.95) -- (2.1,0.95);

\draw[thick, rounded corners] (0.5,-2) -- (0,-1.5) -- (1, -1.5) -- cycle;
\node at (0.5,-1.7) {$x$};
\draw[thick, rounded corners, xshift=5.5cm] (-2,-2) -- (-1.5,-1.5) -- (-2.5, -1.5) -- cycle;
\node[xshift=5.5cm] at (-2,-1.7) {$y$};

\draw[thick, rounded corners, xshift=4.5cm, yshift=0.5cm] (-4,2) -- (-3.5,1.5) -- (-4.5, 1.5) -- cycle;
\node[xshift=4.5cm, yshift=0.5cm] at (-4,1.7) {$a$};
\draw[thick, rounded corners,xshift=5.5cm, yshift=0.5cm] (-2,2) -- (-1.5,1.5) -- (-2.5, 1.5) -- cycle;
\node[xshift=5.5cm, yshift=0.5cm] at (-2,1.7) {$b$};

\node at (0.5,3) {Alice};
\node at (3.5,3) {Bob};

\node at (-1,0.2) {$=$};
\node at (-2.5,0.2) {$p_{\text{AQ}}(a,b|x,y)$};

\end{tikzpicture}
\end{center}
\caption{An almost quantum map constructed from the realisation of an almost quantum correlation. For simplicity we depict the case of two parties, Alice and Bob. First each party performs a unitary operation on their share of the ancilla $R$ (initialised in state $\ket{\Psi}=\ket{0}\ket{\psi}\ket{0}$) controlled on their input qumits $\ket{x}$ and $\ket{y}$ (see eqs.~\eqref{eq:AQuni} and \eqref{eq:theOaq}). Then, the input systems together with the part of the ancilla on Hilbert space $\cH_{d^\prime}$ are traced out, and the ancilla qudits measured on the computational basis $\ket{a}$ and $\ket{b}$. }
\label{f:theAQconstr}
\end{figure}
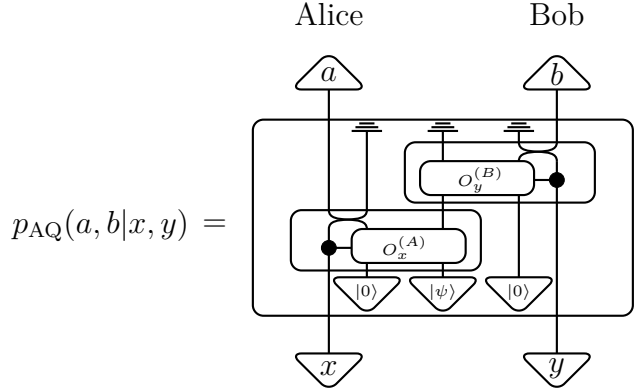

\subsection{Connections between channels and correlations}\label{construct}

{In this section, we first comment on how, given some correlations in a Bell scenario, one can find a canonical channel that channel-defines them. Then, we elaborate on further ways one may use a channel to generate correlations. }

{We have previously considered how correlations result from channels. One can then readily ask how channels can be constructed once we are given a set of correlations. Given correlations $p(a_{1}...a_{N}|x_{1}...x_{N})$, there is a \textit{canonical channel} that channel-defines them, which amounts to a controlled preparation of a quantum system. In particular, for a given choice of input and output orthonormal bases $\{|a_{j}\rangle\}$, $\{|x_{j}\rangle\}$ for all $N$ parties, such a channel is defined as
\begin{equation}\label{canonical}
\Lambda^{c}_{1...N}(\cdot)=\sum_{x_{1},...,x_{N}}\sum_{a_{1},...,a_{N}}p(a_{1}...a_{N}|x_{1}...x_{N})|a_{1}...a_{N}\rangle\langle x_{1}...x_{N}|(\cdot)|x_{1}...x_{N}\rangle\langle a_{1}...a_{N}|.
\end{equation}}
{It can be readily seen that $\Lambda^{c}_{1...N}$ channel defines the correlations $p(a_{1}...a_{N}|x_{1}...x_{N})$ for the choice of input and output orthonormal bases $\{|a_{j}\rangle\}$, $\{|x_{j}\rangle\}$ for all $N$ parties. We also remark that one can take any channel $\Lambda$ that channel-defines some correlations with given preparations and measurements, and then construct the canonical channel $\Lambda^{c}$ from $\Lambda$ with those preparations and measurements. This intuitively amounts of the taking the original channel and applying fully decoherent channels to the inputs and outputs.}

{Now, let us elaborate further on how correlations may arise from the use of quantum channels.}
If we are given a particular channel $\Lambda_{1...N}$, indeed, choosing a set of orthonormal bases such that correlations are channel-defined by $\Lambda_{1...N}$ may not be \textit{optimal} for witnessing non-locality. That is, given access to a channel, correlations can be generated through more elaborate means than just preparing a state from an orthonormal basis, plugging it into a local port of the channel, and then measuring in another basis. For example, one party could prepare a bipartite system and send one half of it into the channel, then after the system has emerged from the channel, one can jointly measure this output and the remaining half of the bipartite system\footnote{A more general strategy would be to apply an instrument with a quantum memory to the channel. That is, preparing a bipartite state, and then sequentially using the channel, in between each use a party applies an operation to the output of the channel and the other half of the bipartite system (stored in a memory). This would be in analogy to performing a Bell test through collective measurements on a number of quantum states.}. More formally, with each party $j$, in addition to the Hilbert spaces associated with the $j$th input and output ports of the channel, we associate an auxiliary Hilbert space $\mathcal{H}_{aux}^{j}$. Then for a given input $x_{j}$ for the $j$th party, without loss of generality, this party can prepare a state $\rho_{x_j}\in\mathcal{D}(\mathcal{H}_{in}^{j}\otimes\mathcal{H}_{aux}^{j})$, and then the output $a_j$ is associated with some POVM element $M_{a_j}\in\mathcal{L}(\mathcal{H}_{out}^{j}\otimes\mathcal{H}_{aux}^{j})$, such that $\sum_{a_j}M_{a_j}=\id$ and every element $M_{a_j}$ is a positive operator\footnote{We do not need to explicitly consider choices of different measurements for $\{M_{a}\}$, since the state $\rho_{x}$ carries the information about the input.}. Putting this together, given a channel $\Lambda_{1...N}$, correlations are generated as:
\begin{equation}\label{gencorr}
p(a_{1},a_{2},...,a_{N}|x_{1},x_2,...,x_{N})=\Tr\left\{\bigotimes_{j=1}^{N}M_{a_{j}}\left(\Lambda_{1...N}\otimes\id_{aux}(\bigotimes_{j=1}^{N}\rho_{x_{j}})\right)\right\},
\end{equation}
where $\id_{aux}$ is the identity operator acting on all Hilbert spaces $\mathcal{H}_{aux}^{j}$. {This allows us to explore whether a particular channel may result in non-local correlations, as we now formalise in the following definition.} 

\begin{defn}\textbf{Local-limited channels.}\\
{If for all states $\rho_{x_{j}}$ and measurements $\{M_{a_{j}}\}$, a channel $\Lambda_{1...N}$ never produces non-local correlations via eq.~\eqref{gencorr}, then the channel is \textit{local-limited}.}
\end{defn}

{There are several channels which are local-limited. As an example, consider the entangled quantum states that can never produce non-local correlations for all general measurements \cite{barrett,quintino}.}
These quantum states can give rise to localizable channels that are not local yet are local-limited. The construction goes as follows. Take a localizable channel where the ancillary system is initiated in such an entangled quantum state. In addition, the `unitary operations' between the input and ancillary ports of each party simply trace out the input states. For all practical purposes then, this channel only prepares a fixed quantum state among the parties, which then goes to the output ports. It follows that event though such channel is not local, it is however local-limited.

In general, if $\Lambda^{c}_{1...N}$ is a \textit{canonical} channel for local correlations $p(a_{1}...a_{N}|x_{1}...x_{N})$, we have the following result, which is proven in Section \ref{appseccon} of the Appendix.

\begin{prop}\label{canon}
Given $\Lambda^{c}_{1...N}(\cdot)$ from $p(a_{1}...a_{N}|x_{1}...x_{N})$, for all measurements $M_{a'_{j}}$, and all states $\rho_{x'_{j}}$, the correlations
\begin{equation*}
p(a'_{1},a'_{2},...,a'_{N}|x'_{1},x'_2,...,x'_{N})=\Tr\left\{\bigotimes_{j=1}^{N}M_{a'_{j}}\left(\Lambda^{c}_{1...N}\otimes\id_{aux}(\bigotimes_{j=1}^{N}\rho_{x'_{j}})\right)\right\}
\end{equation*}
are local if the correlations $p(a_{1}...a_{N}|x_{1}...x_{N})$ are local.
\end{prop}

{Another interesting question is that of constructing almost localizable channels that are non-localizable. The following method works for any general Bell scenario as a starting point, depending on which type of channel one wishes to construct, and goes beyond the canonical form previously discussed. For the sake of simplicity, however, we focus on a bipartite Bell scenario with two dichotomic measurements per party. 

First, take an almost quantum correlation $p(a,b|x,y)$ with no quantum realisation. Such  correlations can be found by taking those that violate Bell inequalities beyond a Tsirelson-like bound, as presented in \cite{aqp}. Then, obtain a state and measurements that reproduce the correlations as outlined in \cite{NPApaper}. Using the protocol described in the proof of Prop.~\ref{thmalmost}, also depicted in Fig.~\ref{f:theAQconstr}, an almost localizable channel that channel-defines these correlations can be hence constructed from these `state and measurements'. 
This almost localizable channel is hence provably non-localizable, since it channel-defines Bell correlations beyond what quantum theory allows, and completes the picture of the hierarchy of channels in Theorem \ref{channelhierarchy}.
}

{Finally, while Theorem \ref{thmalmost} tells us that almost localisable channels channel-define the almost quantum correlations, does this mean the correlations in eq. \ref{gencorr} that are generated by an almost localisable channel $\Lambda_{1...N}$ will necessarily be almost quantum correlations? The answer does not follow immediately from the statement of Theorem \ref{thmalmost}, but the proof of this theorem can be slightly extended to give an answer in the affirmative. To sketch this extension, first note that all states $\rho_{x_{j}}$ and measurements $M_{a_{j}}$ can be made pure and projective respectively by introducing a large enough auxiliary system for each party. That is, $\rho_{x_{j}}$ can be replaced by a pure state $\ket{\psi_{x_{j}}}$ in a larger space, and then we can rewrite these states to be $\ket{\psi_{x_{j}}}=V_{x_{j}}\ket{\psi_{0}^{j}}$ for some fixed state $\ket{\psi_{0}}$. Now if we apply an almost localisable channel to (part of) these input states, the whole process can be modelled as preparing the state $\ket{\psi}_{E}\bigotimes\ket{\psi_{0}^{j}}$, then applying $V_{x_{j}}$ to the input states, followed by the unitaries in the almost localisable channel. Finally, a projective measurement is made on the output qubits. This whole process is equivalent to applying the inverse of the unitaries to these projective measurements to form new projective measurements which act on the state $\ket{\psi}_{E}\bigotimes\ket{\psi_{0}^{j}}$. These new projective measurements, due to the definition of the almost localisable channel, will ``commute'' for the particular state $\ket{\psi}_{E}\bigotimes\ket{\psi_{0}^{j}}$, and thus will generate almost quantum correlations by definition. Note that due to Theorem \ref{thmalmost}, given almost quantum correlations, we can always find states and measurements and an almost localisable channel that reproduce these correlations.}

\section{Steering from the scope of quantum channels}\label{se:chanstee}

Steering refers to the phenomenon where one party, Alice, by performing measurements on one half of a shared state, seemingly remotely `steers' the states held by a distant party, Bob, in a way which has no classical explanation \cite{S35}. This resembles the phenomenon of non-locality presented in last section, but with a slight change: now one party describes its system as a quantum system. In this section, we discuss an approach to studying steering via quantum channels. 
{Here, we review the traditional notion of a steering scenario, while its relevant sets of assemblages are presented in Appendix \ref{ap:eprstee}.}

{In a bipartite steering scenario, the actions of one party (here Alice, also referred to as `untrusted' or `uncharacterised'\footnote{The sense in which the parties are untrusted is that whatever is used to produce classical data is some black box.}) are described solely by $m$ possible classical inputs to her system, labelled by $x \in \{1 ... m\}$, each of which results in one of $d$ possible classical outputs, labelled $a = \in\{1 ... d\}$. The second party (Bob, also referred to as `trusted' or `characterised') fully describes the state of his share of the system by a sub-normalised quantum state $\sigma_{a|x}\in\mathcal{L}(\mathcal{H}_{B})$, where $\mathcal{H}_{B}$ is the Hilbert space associated with Bob's quantum system with dimension $d_{B}$. The set of sub-normalised conditional states Alice prepares on Bob's side $\{ \sigma_{a|x} \}_{a,x}$ is usually called \textit{assemblage}, and $p(a|x)=\tr{\sigma_{a|x}}$ denotes the probability that such a sub-normalised state is prepared, i.e. the probability that Alice obtains $a$ when measuring $x$. }

{In this work we go beyond the bipartite definition of steering, and consider a setting with $N$ untrusted parties and a single trusted party, still called Bob, who has some associated Hilbert space $\mathcal{H}_{B}$. Now, we have $N$ Alices, where for the $j$th Alice, her input is $x_{j}\in\{1...m\}$ and output is $a_{j}\in\{1...d\}$. As a result Bob obtains an assemblage $\{\sigma_{a_{1}...a_{N}|x_{1}...x_{N}}\}_{a_{1}...a_{N},x_{1},...x_{N}}$ with elements $\sigma_{a_{1}...a_{N}|x_{1}...x_{N}}\in\mathcal{L}(\mathcal{H}_{B})$ such that $p(a_1...a_{N}|x_1...x_{N})=\tr{\sigma_{a_{1}...a_{N}|x_{1}...x_{N}}}$.  As with Bell scenarios, for the case of $N\leq 2$ we will use the same notation of inputs being $x$ and $y$ and outputs being $a$ and $b$.}

\subsection{Steering via quantum channels}\label{se:substee}

Here we extend the ideas of Section \ref{se:sunclc} to steering scenarios, which provides a novel way to understanding the phenomenon. First we introduce the formalism and then characterise the channels that give rise to each set of assemblages. 

{A steering scenario is characterised by $N$ untrusted parties, each of which generate one of $m$ possible inputs, of $d$ possible outcomes each}, and a trusted party Bob with Hilbert space $\mathcal{H}_{B}$ with dimension $d_{B}$. Consider now all $(N+1)$ parties (including Bob) to have input and output Hilbert spaces. For the $N$ untrusted parties, these Hilbert spaces are $\cH^{j}_{in} = \cH^{j'}_{in} = \cH_m$ and $\cH^{j}_{out} = \cH^{j'}_{out} = \cH_d$ for all $j\neq j'$, where $\cH_m$ has dimension $m$ and $\cH_d$ dimension $d$. Denote by $\{\ket{x}\}_{x=1:m}$ an orthonormal basis of $\cH_m$, and by $\{\ket{a}\}_{a=1:d}$ an orthonormal basis of $\cH_d$. For Bob, he has Hilbert spaces $\mathcal{H}_{B_{in}}$ and $\mathcal{H}_{B_{out}}$, which are taken to be equal\footnote{For the specific purposes of studying steering, we could equally take $\mathcal{H}_{B_{in}}$ to be $\mathbb{C}$, i.e. the scalars, but for the sake of simplicity in our presentation, we have this more symmetric set-up.}. 
{In what follows, we relate channels of the form $\Lambda_{1...N,B}:\mathcal{L}(\mathcal{H}_{m}^{\otimes N}\otimes\mathcal{H}_{B_{in}})\rightarrow\mathcal{L}(\mathcal{H}_{d}^{\otimes N}\otimes\mathcal{H}_{B_{out}})$ to assemblages in a steering scenario as in Fig. \ref{f:Stmap}.}

\begin{defn}
An assemblage $\{\sigma_{a_1 ... a_N | x_1 ... x_N}\in\mathcal{D}(\mathcal{H}_{B})\}$ in a steering scenario is \textbf{channel-defined} if there exists a channel $\Lambda_{1...N,B}:\mathcal{L}(\mathcal{H}_{m}^{\otimes N}\otimes\mathcal{H}_{B_{in}})\rightarrow\mathcal{L}(\mathcal{H}_{d}^{\otimes N}\otimes\mathcal{H}_{B_{out}})$, some choice of orthonormal bases $\{\ket{x_{j}}\in\mathcal{H}_{m}\}_{x=1:m}$ and $\{\ket{a_{j}}\in\mathcal{H}_{d}\}_{a=1:d}$ for each  party $j$, and a state $\ket{0}\in\mathcal{H}_{B_{in}}$, such that
\begin{equation}\label{eq:Sequiv}
\sigma_{a_1 ... a_N | x_1 ... x_N} = \Tr_{1 ... N}\left\{\otimes_{k=1}^N \ket{a_k}\bra{a_k} \otimes \id_B \,\, \Lambda_{1 ... N,B} [\otimes_{k=1}^n \ket{x_k}\bra{x_k} \otimes \ket{0}\bra{0}] \right\}\,,
\end{equation}
where the partial trace is taken over all $N$ untrusted systems.
\end{defn}

Note that the main difference between the correlations and assemblages from the point-of-view of channels is that one of the outputs of the channels is left unmeasured, and Bob has a fixed input state $\ket{0}$. {We now relate channels to the families of assemblages presented in Appendix \ref{ap:eprstee}, starting with the local hidden state assemblages.}

\begin{figure}
\begin{center}
\begin{tikzpicture}
\node at (-7.5,0) {$\sigma_{a|x}$};
\node at (-6,0) {$\equiv$};

\shade[draw, thick, ,rounded corners, inner color=white,outer color=gray!50!white] (-5,-1) rectangle (-1,1) ;
\node at (-3,0) {$\Lambda_{A \not\leftrightarrow B}$};
\draw[thick] (-4,-1.5) -- (-4,-1);
\draw[thick] (-2,-1.5) -- (-2,-1);
\draw[thick] (-4,1) -- (-4,1.5);
\draw[thick] (-2,1) -- (-2,1.5);

\node at (-2, 1.7) {$\sigma_{a|x}$};

\draw[thick, rounded corners] (-4,-2) -- (-3.5,-1.5) -- (-4.5, -1.5) -- cycle;
\node at (-4,-1.7) {$x$};
\draw[thick, rounded corners] (-2,-2) -- (-1.5,-1.5) -- (-2.5, -1.5) -- cycle;
\node at (-2,-1.7) {$0$};

\draw[thick, rounded corners] (-4,2) -- (-3.5,1.5) -- (-4.5, 1.5) -- cycle;
\node at (-4,1.7) {$a$};

\node at (-4,2.5) {Alice};
\node at (-2,2.5) {Bob};
\end{tikzpicture}
\end{center}
\caption{An assemblage $\sigma_{a|x}$ seen as generated by a quantum channel $\Lambda_{A \not\leftrightarrow B}$.}
\label{f:Stmap}
\end{figure}
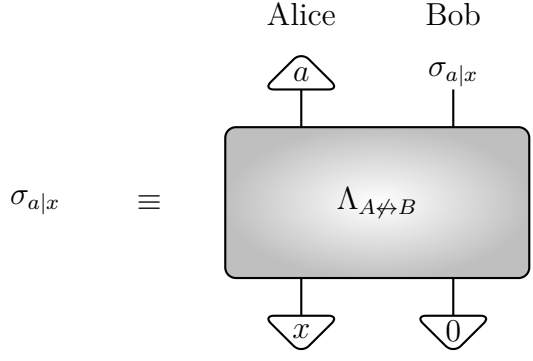

\begin{prop}\label{assLHSprop} 
An assemblage $\{\sigma_{a_{1}...a_{N}|x_{1}...x_{N}}\}$ is a local hidden state assemblage if and only if there exists a local channel $\Lambda_{1...N,B}^{\mathsf{L}}:\mathcal{L}(\mathcal{H}_{m}^{\otimes N}\otimes\mathcal{H}_{B_{in}})\rightarrow\mathcal{L}(\mathcal{H}_{d}^{\otimes N}\otimes\mathcal{H}_{B_{out}})$ such that the assemblage is channel-defined by $\Lambda_{1...N,B}^{\mathsf{L}}$.
\end{prop}

In the literature, {most of the focus has been on detecting whether an assemblage has a local hidden state model, thus revealing entanglement in a shared resource}. It is therefore reassuring that the local hidden state assemblages do not involve entanglement when viewed through the channels that define them. Note that it is possible to have an assemblage which does not have a local hidden state model, yet the correlations resulting any measurement Bob makes on the assemblage can be local. In other words, steering is a distinct phenomenon from non-locality.

\begin{prop}\label{prop11} 
An assemblage $\{\sigma_{a_{1}...a_{N}|x_{1}...x_{N}}\}$ is non-signalling if and only if there exists a causal channel $\Lambda_{1...N,B}^{\mathsf{C}}:\mathcal{L}(\mathcal{H}_{m}^{\otimes N}\otimes\mathcal{H}_{B_{in}})\rightarrow\mathcal{L}(\mathcal{H}_{d}^{\otimes N}\otimes\mathcal{H}_{B_{out}})$ such that the assemblage is channel-defined by $\Lambda_{1...N,B}^{\mathsf{C}}$.
\end{prop}

Given this definition, if one is given a non-signalling assemblage then it is straightforward to find a causal channel that channel-defines the assemblage if the input states and output measurements are fixed. In fact it is an SDP that is a slight modification of the SDP that decides whether a channel is causal as outlined {in Appendix \ref{ap:channels}.} Given elements of the assemblage, since they are channel-defined, this just results in linear constraints made on the channel.

A consequence of the above proposition and the unitary representation of causal channels is the following theorem. 

\begin{thm}\label{nonsigass} \textbf{Unitary representation of non-signalling assemblages}\\
Let $\{\sigma_{a_{1}...a_{N}|x_{1}...x_{N}}\}$ be a non-signalling assemblage. {Then,
the assemblage is channel-defined by a channel $\Lambda_{1...N,B}^{\mathsf{C}}:\mathcal{L}(\mathcal{H}_{m}^{\otimes N}\otimes\mathcal{H}_{B_{in}})\rightarrow\mathcal{L}(\mathcal{H}_{d}^{\otimes N}\otimes\mathcal{H}_{B_{out}})$ if and only if 
there exist
\begin{itemize}
\item auxiliary systems $E$ and $E'$ with input and output Hilbert spaces, $\mathcal{H}^{E}_{in}$  and $\mathcal{H}^{E}_{out}$ for $E$, with $\mathcal{H}^{E'}_{in}=\mathcal{H}_{B}$ and  $\mathcal{H}^{E'}_{out}=\mathcal{H}_{B_{out}}$ for $E'$, that is the output Hilbert space of $E'$ and $B$ coincide;
\item quantum state $|R\rangle\in\mathcal{H}^{E}_{in}\otimes\mathcal{H}^{E'}_{in}$;
\item unitary operator $V:\mathcal{H}_{d}^{\otimes N}\otimes\mathcal{H}^{E}_{in}\rightarrow\mathcal{H}_{d}^{\otimes N}\otimes\mathcal{H}^{E}_{out}$,
\end{itemize}
which produce a unitary representation of the channel $\Lambda_{1...N,B}^{\mathsf{C}}$ via 
\begin{equation*}
\Lambda_{1...N,B}^{\mathsf{C}}(\cdot)=\Tr_{E_{out}}\{V\otimes\id_{E'}(\Tr_{B_{in}}(\cdot)\otimes|R\rangle\langle  R|_{E,E'})V^{\dagger}\otimes\id_{E'}\}.
\end{equation*}
}

Futhermore the unitary $V$ can be decomposed into a sequence of unitaries $U_{k,E}:\mathcal{H}_m\otimes\mathcal{H}^{E}_{1}\rightarrow\mathcal{H}_d\otimes\mathcal{H}^{E}_{2}$ for appropriately chosen Hilbert spaces $\mathcal{H}^{E}_{1}$ and $\mathcal{H}^{E}_{2}$, where for any given permutation $\pi$ of the set $\{1,...,N\}$, we have that
\begin{equation*}
V=U^{\pi}_{\pi(1),E}U^{\pi}_{\pi(2),E}...U^{\pi}_{\pi(N),E}
\end{equation*}
where $U^{\pi}_{k,E}$ is not necessarily the same as $U^{\pi'}_{k,E}$ for two different permutations $\pi$ and $\pi'$.
\end{thm}
A pictorial representation of this theorem for $N=2$ is given in Fig.~\ref{f:matty}. 

Given this characterisation of the set of non-signalling assemblages, we now turn to the set of quantum assemblages.

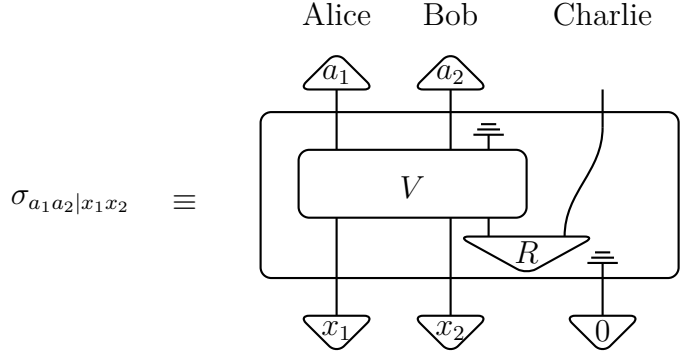
\begin{figure}
\begin{center}
\begin{tikzpicture}
\node at (-7.5,0) {$\sigma_{a_{1}a_{2}|x_{1}x_{2}}$};
\node at (-6,0) {$\equiv$};

\draw[thick, ,rounded corners] (-5,-1) rectangle (0.5,1.2) ;
\draw[thick] (-4,-1.5) -- (-4,-0.2);
\draw[thick] (-2.5,-1.5) -- (-2.5,-0.2);
\draw[thick] (-4,0.7) -- (-4,1.5);
\draw[thick] (-2.5,0.7) -- (-2.5,1.5);
\draw[thick] (-0.5,1) -- (-0.5,1.5);
\draw[thick] (-0.5,-1.5) -- (-0.5,-1);

\draw[thick] (-0.5,-1) -- (-0.5,-0.8);
\draw[thick, xshift=-3.5cm, yshift=-1.65cm] (2.9,0.99) -- (3.1,0.99);
\draw[thick, xshift=-3.5cm, yshift=-1.65cm] (2.85,0.92) -- (3.15,0.92);
\draw[thick, xshift=-3.5cm, yshift=-1.65cm] (2.8,0.85) -- (3.2,0.85);

\draw[thick, rounded corners] (-4,-2) -- (-3.5,-1.5) -- (-4.5, -1.5) -- cycle;
\node at (-4,-1.7) {$x_{1}$};
\draw[thick, rounded corners] (-2.5,-2) -- (-2,-1.5) -- (-3, -1.5) -- cycle;
\node at (-2.5,-1.7) {$x_{2}$};
\draw[thick, rounded corners] (-0.5,-2) -- (0,-1.5) -- (-1, -1.5) -- cycle;
\node at (-0.5,-1.7) {$0$};

\draw[thick] (-2,0.7) -- (-2,0.9);
\draw[thick, xshift=-0.5cm, yshift=-0.25cm] (-1.6,1.29) -- (-1.4,1.29); 
\draw[thick, xshift=-0.5cm, yshift=-0.25cm] (-1.65,1.22) -- (-1.35,1.22); 
\draw[thick, xshift=-0.5cm, yshift=-0.25cm] (-1.7,1.15) -- (-1.3,1.15); 

\draw[thick, rounded corners] (-4,2) -- (-3.5,1.5) -- (-4.5, 1.5) -- cycle;
\node at (-4,1.7) {$a_{1}$};
\draw[thick, rounded corners] (-2.5,2) -- (-2,1.5) -- (-3, 1.5) -- cycle;
\node at (-2.5,1.7) {$a_{2}$};

\node at (-4,2.5) {Alice};
\node at (-2.5,2.5) {Bob};
\node at (-0.5,2.5) {Charlie};
\node at (-0.5,1.7) {};

\draw[thick, rounded corners, xshift=-4.5cm] (3,-0.95) -- (3.9,-0.45) -- (2.1, -0.45) -- cycle;
\node at (-1.5,-0.66) {$R$};

\draw[thick, rounded corners] (-4.5,-0.2) rectangle (-1.5,0.7) ;
\node at (-3,0.2) {$V$};
\draw[thick] (-2,-0.45) -- (-2,-0.2);
\draw[thick] (-1,-0.45) to [out=90,in=-90] (-0.5,1);

\end{tikzpicture}
\end{center}
\caption{Unitary representation of non-signalling assemblages for $N=2$.}
\label{f:matty}
\end{figure}

\begin{prop}\label{assQprop} 
An assemblage $\{\sigma_{a_{1}...a_{N}|x_{1}...x_{N}}\}$ is quantum if and only if there exists a localizable channel $\Lambda_{1...N,B}^{\mathsf{Q}}:\mathcal{L}(\mathcal{H}_{m}^{\otimes N}\otimes\mathcal{H}_{B_{in}})\rightarrow\mathcal{L}(\mathcal{H}_{d}^{\otimes N}\otimes\mathcal{H}_{B_{out}})$ such that the assemblage is channel-defined by $\Lambda_{1...N,B}^{\mathsf{Q}}$.
\end{prop}

\begin{figure}
\begin{center}
\begin{tikzpicture}
\node at (-1.5,0) {$\sigma^{Q}_{a|x}$};
\node at (0,0) {$\equiv$};
\draw[thick, rounded corners] (1,-1) rectangle (5,1);
\draw[thick, rounded corners] (3,-0.95) -- (3.9,-0.45) -- (2.1, -0.45) -- cycle;
\node at (3,-0.66) {$R$};

\draw[thick, rounded corners] (1.5,-0.1) rectangle (2.9,0.4);
\draw[thick, rounded corners] (3.1,-0.1) rectangle (4.5,0.4);
\node at (2.3,0.1) {$U_{AR}$};
\node at (3.7,0.1) {$U_{RB}$};

\draw[thick] (2,-1.5) -- (2,-0.1);
\draw[thick] (2,0.4) -- (2,1.5);
\draw[thick] (4,-1.5) -- (4,-0.1);
\draw[thick] (4,0.4) -- (4,1.5);

\draw[thick] (2.5,-0.45) -- (2.5,-0.1);
\draw[thick] (3.5,-0.45) -- (3.5,-0.1);

\draw[thick] (2.5,0.4) -- (2.5,0.6);
\draw[thick, xshift=-0.5cm, yshift=-0.25cm] (2.9,0.99) -- (3.1,0.99);
\draw[thick, xshift=-0.5cm, yshift=-0.25cm] (2.85,0.92) -- (3.15,0.92);
\draw[thick, xshift=-0.5cm, yshift=-0.25cm] (2.8,0.85) -- (3.2,0.85);

\draw[thick] (3.5,0.4) -- (3.5,0.6);
\draw[thick, xshift=0.5cm, yshift=-0.25cm] (2.9,0.99) -- (3.1,0.99);
\draw[thick, xshift=0.5cm, yshift=-0.25cm] (2.85,0.92) -- (3.15,0.92);
\draw[thick, xshift=0.5cm, yshift=-0.25cm] (2.8,0.85) -- (3.2,0.85);

\node at (2,2.5) {Alice};
\node at (4,2.5) {Bob};
\draw[thick, rounded corners] (4,-2) -- (3.5,-1.5) -- (4.5, -1.5) -- cycle;
\node at (4,-1.7) {$0$};
\draw[thick, rounded corners] (2,-2) -- (1.5,-1.5) -- (2.5, -1.5) -- cycle;
\node at (2,-1.7) {$x$};
\draw[thick, rounded corners] (2,2) -- (1.5,1.5) -- (2.5, 1.5) -- cycle;
\node at (2,1.7) {$a$};
\node at (4,1.7) {$\sigma_{a|x}$};
\end{tikzpicture}
\end{center}
\caption{A quantum assemblage in a bipartite steering scenario resulting from a localizable map.}
\label{f:Slocaliz}
\end{figure}
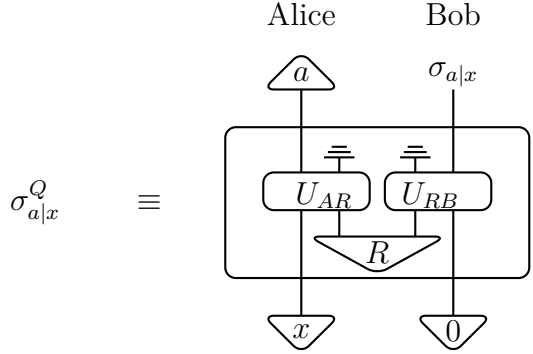

In Fig. \ref{f:Slocaliz} we give a pictorial representation of a channel-defined quantum assemblage. At this point we should point out the following corollary of this proposition along with the previous theorem, which was first proven by Gisin and Hughston, Jozsa, and Wootters (GHJW). We note that our proof is structurally very different from the previous proofs, and is a simple consequence of the fact that, for $N=1$, the unitary $V$ in Thm. \ref{nonsigass} acts only on the input Hilbert space of the untrusted party and the ancillary register. {The full proof of this corollary can be found in Appendix \ref{appc}.}

\begin{cor} \label{cor:GHJW}
For $N=1$, all non-signalling assemblages are also quantum assemblages.
\end{cor}

It is important to note that this is only true for the case of $N=1$, i.e. a single untrusted party. In subsection \ref{postquantumex}, we use causal channels to give examples of post-quantum steering, i.e. non-signalling assemblages that are not quantum. We know that post-quantum correlations witness a non-localizable channel, then any assemblage that gives post-quantum correlations must have an associated non-localizable channel. However, there exist non-quantum assemblages that will never give rise to non-quantum correlations \cite{pqsp}: there are assemblages that cannot be channel-defined by a localizable channel, but for any measurement made on the Bob's system the corresponding Bell correlations are channel-defined by a local channel. This highlights that post-quantum steering is distinct from post-quantum non-locality, and indeed from non-locality itself.

For $N\geq 2$, as pointed out in \cite{pqsp}, characterising the set of quantum assemblages is difficult, and at least as hard as characterising the set of quantum correlations. However, the almost quantum assemblages are a superset of the quantum assemblages, and for the former there is a characterisation in terms of a semi-definite program. In the next result we give a \textit{physical interpretation} for the almost quantum assemblages.

\begin{thm} \label{stalmost}
An assemblage $\{\sigma_{a_{1}...a_{N}|x_{1}...x_{N}}\}$ is almost quantum if and only if there exists an almost localizable channel $\Lambda_{1...N,B}^{\mathsf{\tilde{Q}}}:\mathcal{L}(\mathcal{H}_{m}^{\otimes N}\otimes\mathcal{H}_{B_{in}})\rightarrow\mathcal{L}(\mathcal{H}_{d}^{\otimes N}\otimes\mathcal{H}_{B_{out}})$ such that the assemblage is channel-defined by $\Lambda_{1...N,B}^{\mathsf{\tilde{Q}}}$.
\end{thm}

The full proof is in Appendix \ref{ap:AQCSDP}, but is essentially a consequence of the following lemma, which is also proven in Appendix \ref{ap:AQCSDP}. Given this lemma, one can essentially use the proof of Theorem \ref{thmalmost} to obtain the result in Theorem \ref{stalmost}. 

\begin{lemma}\label{lemmaAQ}
An assemblage $\{\sigma_{a_{1}...a_{N}|x_{1}...x_{N}}\}$ is almost quantum if and only if there exists a Hilbert space $\mathcal{H}\cong\mathcal{K}\otimes\mathcal{H}_{B}$, quantum state $|\psi\rangle\in\mathcal{H}$, and projective measurements $\{\Pi_{a_{j}|x_{j}}\in\mathcal{L}(\mathcal{K})\}$ for each $j$th party where $\sum_{a_{j}}\Pi_{a_{j}|x_{j}}=\id$ and for all permutations $\pi$ of $\{1,...,N\}$, $\prod_{j=1}^{N}\Pi_{a_{\pi(j)}|x_{\pi(j)}}|\psi\rangle=\prod_{j=1}^{N}\Pi_{a_{j}|x_{j}}|\psi\rangle$, such that
\begin{equation*}
\sigma_{a_{1}...a_{N}|x_{1}...x_{N}}=\Tr_{\mathcal{K}}\{\prod_{j=1}^{N}\Pi_{a_{j}|x_{j}}\otimes\id_{{B}}|\psi\rangle\langle\psi|\}.
\end{equation*}
\end{lemma}

\subsection{Connections between channels and assemblages}\label{constchanass}

In Section \ref{construct}, we indicated the general way to obtain correlations given a channel, and then we gave a canonical way of constructing a channel from correlations. In this section, we will do exactly the same for the case of steering. 

In analogy with the case of Bell non-locality, we will first describe a general way to generate an assemblage from a channel. As in the case of Bell non-locality, the $N$ untrusted parties can prepare a state $\rho_{x_{j}}\in\mathcal{D}(\mathcal{H}_{in}^{j}\otimes\mathcal{H}_{aux}^{j})$ indexed by their input $x_{j}$ for $j\in\{1,...,N\}$, put one of its subsystems (living in $\mathcal{H}_{in}^{j}$) into the channel, and jointly measure the output of the channel and other subsystems associated with initial state $\rho_{x_{j}}$.  The measurements then are the operators $M_{a_{j}}\in\mathcal{L}(\mathcal{H}_{out}^{j}\otimes\mathcal{H}_{aux}^{j})$, which have outcomes $a_{j}$. The novelty in steering is the trusted party, and there is a potential ambiguity in how to generate an assemblage from a channel with $N$ input port and $N$ output ports. We could restrict to channels that trace out the input of the trusted party (or, equivalently, there is no input port), or the trusted party just always inputs the same quantum state into the channel. The second approach is more general when one considers the possibility that the trusted party has an auxiliary sytem with Hilbert space $\mathcal{H}_{B_{aux}}$, and prepares the state $\sigma_{B}\in\mathcal{D}(\mathcal{H}_{B_{in}}\otimes\mathcal{H}_{B_{aux}})$; there could be correlations between the input system and auxiliary system that would be destroyed by tracing out the input system. This more general approach results in the assemblage being a set of operators that act on the Hilbert space $\mathcal{H}_{B_{out}}\otimes\mathcal{H}_{B_{in}}$, and is in the spirit of \textit{channel steering} \cite{pianichan}, which we touch upon later. 

To summarise this discussion, given a causal channel $\Lambda_{1...N,B}$, each $j$th untrusted party will prepare the state $\rho_{x_{j}}\in\mathcal{D}(\mathcal{H}_{in}^{j}\otimes\mathcal{H}_{aux}^{j})$, and obtain measurement outcomes corresponding to the operators $M_{a_{j}}\in\mathcal{L}(\mathcal{H}_{out}^{j}\otimes\mathcal{H}_{aux}^{j})$. The trusted party with Hilbert space $\mathcal{H}_{B}$, will prepare the state $\sigma_{B}\in\mathcal{D}(\mathcal{H}_{B_{in}}\otimes\mathcal{H}_{B_{aux}})$, and thus generate assemblage elements $\sigma_{a_{1}...a_{N}|x_{1}...x_{N}}\in\mathcal{D}(\mathcal{H}_{B_{out}}\otimes\mathcal{H}_{B_{aux}})$, which can be obtained as
\begin{equation}\label{genass}
\sigma_{a_{1}...a_{N}|x_{1}...x_{N}}=\Tr_{1...N}\left\{\bigotimes_{j=1}^{N}M_{a_{j}}\otimes\id_{B_{out}}\left(\Lambda_{1...N,B}\otimes\id_{aux}(\bigotimes_{j=1}^{N}\rho_{x_{j}}\otimes\sigma_{B})\right)\right\},
\end{equation}
where $\id_{aux}$ is the identity operator acting on all Hilbert spaces $\mathcal{H}_{aux}^{j}$.

Let us now move on to the case of constructing a generic channel from an assemblage. That is, given an assemblage with elements $\sigma_{a_{1}...a_{N}|x_{1}...x_{N}}\in\mathcal{D}(\mathcal{H}_{B})$, we specify a canonical channel $\Sigma^{c}_{1...N}:\mathcal{L}(\mathcal{H}_{m}^{\otimes N}\otimes\mathcal{H}_{B_{in}})\rightarrow\mathcal{L}(\mathcal{H}_{d}^{\otimes N}\otimes\mathcal{H}_{B_{out}})$, with $\mathcal{H}_{B_{out}}=\mathcal{H}_{B}$, that will reproduce that assemblage, given appropriate choices of preparations and measurements. This canonical channel is defined as
\begin{equation*}
\Sigma^{c}_{1...N,B}(\cdot)=\sum_{x_{1},...,x_{N}}\sum_{a_{1},...,a_{N}}|a_{1}...a_{N}\rangle\langle x_{1}...x_{N}|\Tr_{B_{in}}(\cdot)|x_{1}...x_{N}\rangle\langle a_{1}...a_{N}|\otimes\sigma_{a_{1}...a_{N}|x_{1}...x_{N}},
\end{equation*}
and can be seen as a channel which completely decoheres the input and output systems with respect to a basis, traces out the trusted party's input, and then produces assemblage elements in the trusted party's output of the channel. Notice that the assemblage elements $\sigma_{a_{1}...a_{N}|x_{1}...x_{N}}$ are channel defined by $\Sigma^{c}_{1...N,B}$, as long as appropriate elements of an orthonormal basis are chosen. 

{The channel $\Sigma^{c}_{1...N,B}$ can moreover be used to} generate correlations, and not just assemblages. This is done by the method outlined in Section \ref{construct}, where the correlations are obtained as 
\begin{equation*}
p(a'_{1},...,a'_{N},a'_{B}|x'_{1},...,x'_{N},x'_{B})=\Tr\left\{\bigotimes_{j=1}^{N}M_{a'_{j}}\otimes M_{a'_{B}}\left(\Sigma^{c}_{1...N,B}\otimes\id_{aux}(\bigotimes_{j=1}^{N}\rho_{x'_{j}}\otimes\rho_{x'_{B}})\right)\right\},
\end{equation*}
from the local measurements $M_{a'_{j}}$ and states $\rho_{x'_{j}}$, where $x'_{B}$ and $a'_{B}$ represent the trusted party's inputs and outputs respectively. We can now ask when this channel gives non-local correlations, or conversely, when is a channel  $\Sigma^{c}_{1...N,B}$ local-limited. The following result addresses this, and is proven in Section \ref{appseccon} of the Appendix.
\begin{prop}\label{canonass}
Given $\Sigma^{c}_{1...N,B}(\cdot)$ from assemblage elements $\sigma_{a_{1}...a_{N}|x_{1}...x_{N}}$, this channel is local-limited if for all measurements $P_{a_{B}|x_{B}}\in\mathcal{L}(\mathcal{H}_{B})$ indexed by the choice $x_{B}$ and outcomes $a_{B}$, the correlations {$p(a_{1},...,a_{N},a_{B}|x_{1},...,x_{N},x_{B}):=\Tr\{P_{a_{B}|x_{B}}\sigma_{a_{1}...a_{N}|x_{1}...x_{N}}\}$} are local.
\end{prop}

A direct consequence of this result is that the canonical channel that one would construct for the post-quantum assemblage given in Ref.~\cite{pqsp} is a local-limited, yet non-localizable channel. Furthermore, this channel is actually not even almost localizable \cite{pqsp}. We summarise all of these observations in Fig.~\ref{figsets}. 

{One can define moreover the set of channels restricted to producing only quantum correlations, and call them the \textit{quantum-limited} channels, where these correlations can be non-local, therefore defining a larger set than the set of local-limited channels. We can then take, for instance, the post-quantum assemblages from Ref.~\cite{newpaper} that can result in non-local but quantum correlations, and from their canonical channels give quantum-limited channels that are not almost localizable.}

\begin{figure}\label{figsets}
\includegraphics[width=0.8\textwidth]{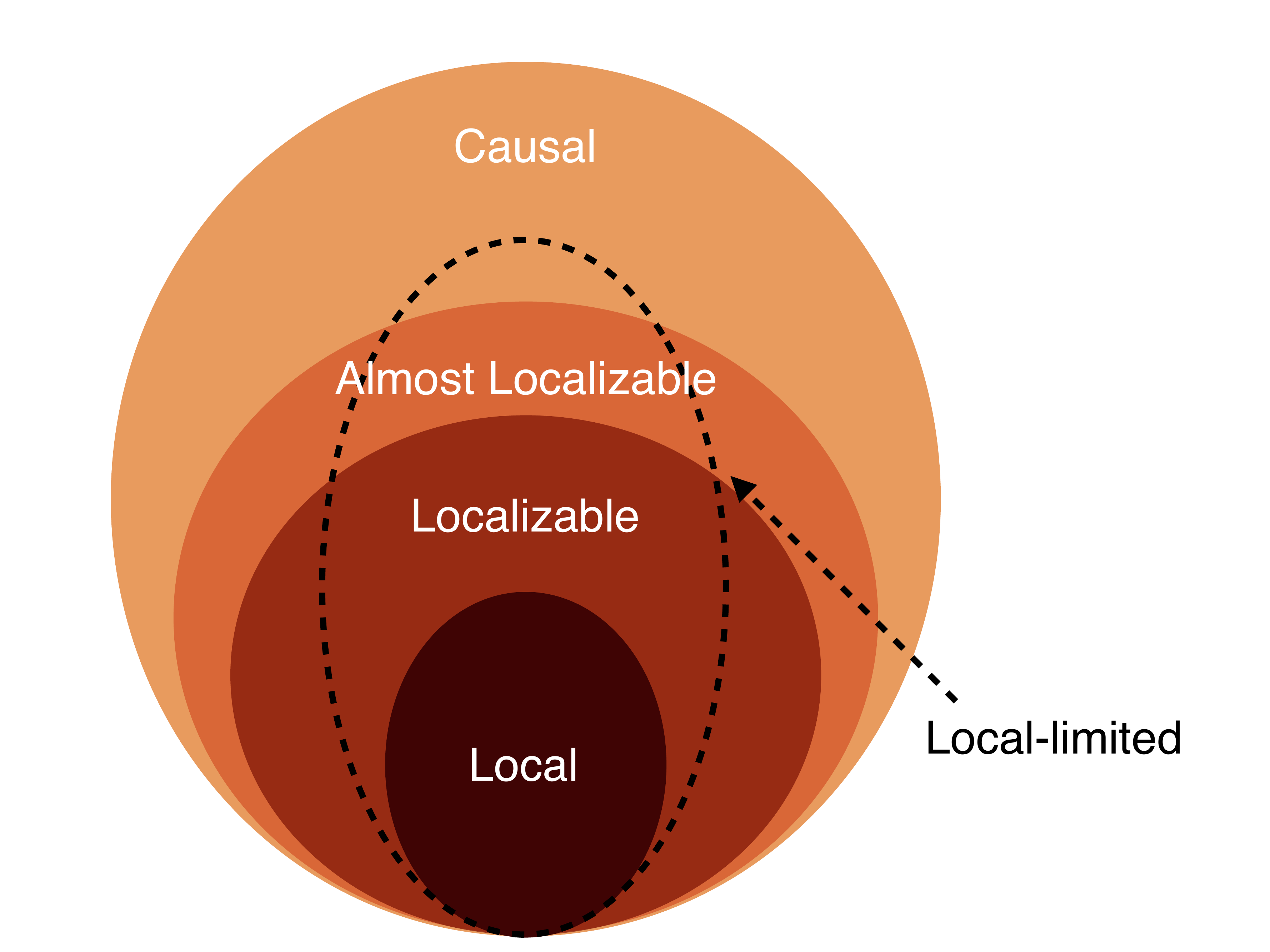}
\caption{A schematic of the sets of causal channels. We have set out the hierarchy in the paper, where the set of local-limited channels (those that do not result in Bell non-locality) intersects all families of channels; this results from knowing that there are non-almost localizable channels that are local-limited, and by taking convex combinations.} 
\end{figure}

\subsection{Examples of post-quantum steering}\label{postquantumex}

In this section we have outlined a constructive way to understand post-quantum steering: assemblages that cannot be channel-defined by localizable channels.  We give a couple of examples of post-quantum steering that are a simple consequence of Theorem \ref{nonsigass}. 

The first example of post-quantum steering is depicted in Fig.~\ref{f:SPRmap}, and is in a tripartite scenario where Alice and Bob steer Charlie, whose Hilbert space has dimension $d_{C}=2$. There, an ancilla is initialised on state $\rho_R = \tfrac{\ket{00}\bra{00} + \ket{11}\bra{11} }{2}_{AB} \otimes \tfrac{\ket{0}\bra{0}+\ket{1}\bra{1}}{2}_{C}$, where $AB$ and $C$ denote Alice and Bob's and Charlie's share of the ancilla system. Then, the part of the ancilla shared by Alice and Bob is used as the ancilla in the channel that generates PR box correlations, while the qubit on state $\tfrac{\ket{0}\bra{0}+\ket{1}\bra{1}}{2}$ is output by Charlie. This map is causal since it has exactly the same form as described in Theorem \ref{nonsigass} (after one locally dilates all processes to be unitary). Once that Alice and Bob input qubits in the computational basis and measure their output systems, the following assemblage elements are then prepared in Charlie's lab: 
$$
\sigma_{ab|xy} = p_{PR}(ab|xy) \, \tfrac{\id}{2}\,.
$$
This assemblage is a non-signalling one which has no quantum realisation \cite{pqsp}.  {However, note that we can have post-quantum steering without any entanglement (across any of the bipartitions) in the shared ancilla state $\rho_{R}$ in the causal channel. In our next example, the ancilla in the channel does consist of entanglement, and the channel generates pure state entanglement between three parties.} 
\begin{figure}
\begin{center}
\begin{tikzpicture}
\node at (-7.5,0) {$\sigma^{PR}_{ab|xy}$};
\node at (-6,0) {$\equiv$};

\draw[thick, ,rounded corners] (-5,-1) rectangle (0.5,1) ;
\draw[thick] (-4,-1.5) -- (-4,-0.9);
\draw[thick] (-2,-1.5) -- (-2,-0.9);
\draw[thick] (-4,1) -- (-4,1.5);
\draw[thick] (-2,0.75) -- (-2,1.5);
\draw[thick] (-0.5,1) -- (-0.5,1.5);
\draw[thick] (-0.5,-1.5) -- (-0.5,-1);

\draw[thick] (-0.5,-1) -- (-0.5,-0.8);
\draw[thick, xshift=-3.5cm, yshift=-1.65cm] (2.9,0.99) -- (3.1,0.99);
\draw[thick, xshift=-3.5cm, yshift=-1.65cm] (2.85,0.92) -- (3.15,0.92);
\draw[thick, xshift=-3.5cm, yshift=-1.65cm] (2.8,0.85) -- (3.2,0.85);
\draw[thick] (-2.4,-0.45) to [out=35,in=-90] (-0.5,1);

\draw[thick, rounded corners] (-4,-2) -- (-3.5,-1.5) -- (-4.5, -1.5) -- cycle;
\node at (-4,-1.7) {$x$};
\draw[thick, rounded corners] (-2,-2) -- (-1.5,-1.5) -- (-2.5, -1.5) -- cycle;
\node at (-2,-1.7) {$y$};
\draw[thick, rounded corners] (-0.5,-2) -- (0,-1.5) -- (-1, -1.5) -- cycle;
\node at (-0.5,-1.7) {$0$};

\draw[thick, rounded corners] (-4,2) -- (-3.5,1.5) -- (-4.5, 1.5) -- cycle;
\node at (-4,1.7) {$a$};
\draw[thick, rounded corners] (-2,2) -- (-1.5,1.5) -- (-2.5, 1.5) -- cycle;
\node at (-2,1.7) {$b$};

\node at (-4,2.5) {Alice};
\node at (-2,2.5) {Bob};
\node at (-0.5,2.5) {Charlie};
\node at (-0.5,1.7) {$\sigma_{ab|xy}$};

\draw[thick, rounded corners, xshift=-6cm] (3,-0.95) -- (3.9,-0.45) -- (2.1, -0.45) -- cycle;
\node at (-3,-0.66) {$R$};

\draw[thick] (-3,-0.45) -- (-3,-0.25);
\draw[thick, xshift=-6cm, yshift=-1.1cm] (2.9,0.99) -- (3.1,0.99);
\draw[thick, xshift=-6cm, yshift=-1.1cm] (2.85,0.92) -- (3.15,0.92);
\draw[thick, xshift=-6cm, yshift=-1.1cm] (2.8,0.85) -- (3.2,0.85);

\draw[thick] (-3.4,-0.45) to [out=90,in=-90] (-4,1);
\draw[thick, rounded corners, yshift=-2.4cm] (-4,2) -- (-3.5,1.5) -- (-4.5, 1.5) -- cycle;
\node at (-4,-0.7) {$x$};
\draw[thick, rounded corners, yshift=-2.4cm] (-2,2) -- (-1.5,1.5) -- (-2.5, 1.5) -- cycle;
\node at (-2,-0.7) {$y$};

\draw[thick, rounded corners] (-2.25, 0.25) rectangle (-1.75, 0.75 );
\node at (-2, 0.5) {$X$};

\draw[thick] (-2.6,-0.45) to [out=90,in=-90] (-2,0.25);

\draw[thick, dashed, color=gray!50!white] (-2,-0.45) to [out=90,in=90] (-4,-0.45);

\node at (-3, 0.5) {$\wedge$};
\draw[thick, dashed, color=gray!50!white] (-3,0.5) circle [radius=0.25];
\draw[thick, dashed, color=gray!50!white] (-3, 0.1) -- (-3, 0.24);
\draw[thick, dashed, color=gray!50!white, ->] (-2.7, 0.5) -- (-2.25, 0.5);

\end{tikzpicture}
\end{center}
\caption{A causal channel that generates a non-signalling assemblage given particular input states and measurements. A tripartite ancilla state $\rho_R = \tfrac{\ket{00}\bra{00} + \ket{11}\bra{11} }{2} \otimes \tfrac{\ket{0}\bra{0}+\ket{1}\bra{1}}{2}$ is generated by preparing the pure state $\ket{R}$ and tracing out part of it. Charlie's output is his part of the ancilla system, which is in the state $\tfrac{\ket{0}\bra{0}+\ket{1}\bra{1}}{2}$. Alice and Bob then implement, on their ancillary systems, the channel that generates PR box correlations for given particular measurements. Therefore, the assemblage prepared in Charlie's lab is $\sigma_{ab|xy} = p_{PR}(ab|xy) \, \tfrac{\id}{2}$. }
\label{f:SPRmap}
\end{figure}
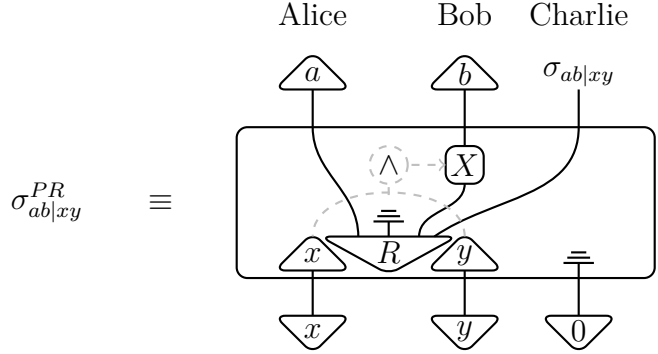

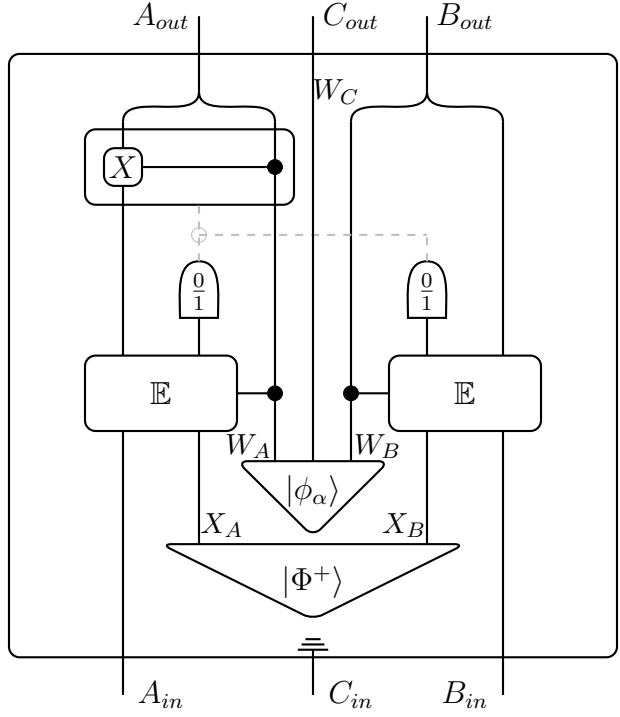
\begin{figure}
\begin{center}
\begin{tikzpicture}

\draw[thick, rounded corners] (-2,0) -- (0,-1) -- (2,0) -- cycle;
\node at (0,-0.5) {$\ket{\Phi^{+}}$};
\draw[thick] (-1.5,0) -- (-1.5,1.5);
\draw[thick] (1.5,0) -- (1.5,1.5);

\draw[thick, rounded corners] (-1,1.1) -- (0,0.1) -- (1,1.1) -- cycle;
\node at (0,0.7) {$\ket{\phi_\alpha}$};
\draw[thick] (-0.5,1.1) -- (-0.5,5.6);
\draw[thick] (0.5,1.1) -- (0.5,5.6);

\draw[thick] (-2.5,-2) -- (-2.5,1.5);
\draw[thick] (2.5,-2) -- (2.5,1.5);

\draw[thick, rounded corners] (-3,1.5) rectangle (-1,2.5);
\draw[thick, rounded corners] (3,1.5) rectangle (1,2.5);
\node at (-2,2) {$\mathbb{E}$};
\node at (2,2) {$\mathbb{E}$};
\draw[thick] (-1,2) -- (-0.5,2);
\draw[thick] (1,2) -- (0.5,2);
\node[draw, scale=0.5, circle, fill] at (-0.5,2) {};
\node[draw, scale=0.5, circle, fill] at (0.5,2) {};

\draw[thick] (-1.25,3) to [out=90, in=0] (-1.5,3.75) to [out=180, in=90] (-1.75,3)--cycle;
\draw[thick] (1.25,3) to [out=90, in=180] (1.5,3.75) to [out=0, in=90] (1.75,3)--cycle;
\draw[thick] (-1.5,2.5)--(-1.5,3);
\draw[thick] (1.5,2.5)--(1.5,3);
\node at (-1.5,3.35) {$\frac{0}{1}$};
\node at (1.5,3.35) {$\frac{0}{1}$};

\draw[thick, rounded corners] (-3,4.5) rectangle (-0.25,5.5);
\draw[thick] (-2.5,2.5) -- (-2.5,4.75);
\draw[thick, rounded corners] (-2.75,4.75) rectangle (-2.25,5.25);
\node at (-2.5,5) {$X$};
\draw[thick] (-2.25,5) -- (-0.5,5);
\node[draw, scale=0.5, circle, fill] at (-0.5,5) {};
\draw[thick] (-2.5,5.25) -- (-2.5,5.6);

\draw[thick] (-2.5,5.6) to [out=90, in=-90] (-1.5,6); 
\draw[thick] (-0.5,5.6) to [out=90, in=-90] (-1.5,6);

\draw[thick] (2.5,2.5) -- (2.5,5.6);
\draw[thick] (2.5,5.6) to [out=90, in=-90] (1.5,6); 
\draw[thick] (0.5,5.6) to [out=90, in=-90] (1.5,6);

\draw[thick] (-1.5,6) -- (-1.5,7);
\draw[thick] (1.5,6) -- (1.5,7);
\draw[thick, rounded corners] (-4,-1.5) rectangle (4,6.5);

\node at (-2,-2) {$A_{in}$};
\node at (2,-2) {$B_{in}$};
\node at (-2,7) {$A_{out}$};
\node at (2,7) {$B_{out}$};
\node at (-1.2,0.25) {\small{$X_A$}}; 
\node at (1.2,0.25) {\small{$X_B$}}; 
\node at (-0.85,1.3) {\small{$W_A$}}; 
\node at (0.85,1.3) {\small{$W_B$}}; 

\draw[thick, dashed, color=gray!50!white] (-1.5,4.1) -- (1.5,4.1); 
\draw[thick, dashed, color=gray!50!white] (1.5,3.75) -- (1.5,4.1);
\draw[thick, dashed, color=gray!50!white] (-1.5,3.75) -- (-1.5,4.5);
\node[draw, scale=0.5, circle,color=gray!50!white] at (-1.5,4.1) {};

\draw[thick] (0,1.1) -- (0,7);
\node at (0.5,7) {$C_{out}$};
\node at (0.3,6) {\small{$W_C$}};

\node at (0.5,-2) {$C_{in}$};
\draw[thick] (0,-2) -- (0,-1.4);
\draw[thick, xshift=-3cm, yshift=-2.25cm] (2.9,0.99) -- (3.1,0.99);
\draw[thick, xshift=-3cm, yshift=-2.25cm] (2.85,0.92) -- (3.15,0.92);
\draw[thick, xshift=-3cm, yshift=-2.25cm] (2.8,0.85) -- (3.2,0.85);
\end{tikzpicture}
\end{center}
\caption{A tripartite causal channel that is not localizable. The ancilla is initialised in the state $\ket{\psi} = \ket{\Phi^+}_{X_AX_B} \otimes \ket{\phi_\alpha}_{W_AW_BW_C}$, where $\ket{\Phi^{+}} = 1/\sqrt{2}(\ket{00} + \ket{11})$ and $\ket{\phi_\alpha} = \sqrt{\alpha} \ket{000} +\sqrt{1-\alpha} \ket{111}$
. Alice (Bob) performs a controlled swap $\mathbb{E}$ on the qubits $A$ and $X_A$ ($B$ and $X_B$), where $W_A$ ($W_B$) is the control qubit. Then, qubits $X_A$ and $X_B$ are measured in the computational basis and the logical $\texttt{AND}$ of the results computed. Whenever this is 1, Alice performs a controlled-NOT gate on qubit $A$, with $W_A$ as the control qubit. The output systems are two quqarts: $AW_A$ for Alice and $BW_B$ for Bob, and a qubit $W_C$ for Charlie. }
\label{f:nonloctri}
\end{figure}

The second example of post-quantum steering also comprises a causal channel that is not localizable, and relies on the results of \cite{Pao}. The steering scenario consists of Alice and Bob, who by performing two dichotomic measurements, steer Charlie, whose Hilbert space has dimension $d_{C}=2$. The channel used by the three parties is depicted in Fig.\ref{f:nonloctri}. Each party's input system is given by a qubit labelled by $A$, $B$ and $C$ respectively. Then, the channel makes use of a five qubit ancilla ($X_AW_AW_CW_BX_B$) initialised on the state: 
$$
\ket{\psi} = \tfrac{\ket{00}_{X_AX_B} + \ket{11}_{X_AX_B}}{\sqrt{2}} \otimes \left(\sqrt{\alpha} \ket{000}_{W_AW_BW_C} +\sqrt{1-\alpha} \ket{111}_{W_AW_BW_C}  \right). 
$$
First, Alice (Bob) makes a controlled swap on the input and $X_A$ ($X_B$) qubits, with $W_A$ ($W_B$) being the control qubit. Then, both the $X_A$ and $X_B$ qubits are measured in the computational basis, and their logical $\texttt{AND}$ computed. Finally, if this measurement result is 1, a controlled-NOT is performed by Alice on qubit $A$, with $W_A$ as the control qubit. The output systems are then a ququart $AW_A$ for Alice, another one  $BW_B$ for Bob, and a qubit $W_C$ for Charlie. Since the marginal channel for Alice and Bob is causal \cite{Pao}, this tripartite extension is also causal.
 {To construct an assemblage, $\{\sigma_{ab|xy}\}$ for $x$, $y$, $a$, $b\in\{0,1\}$, from this channel, we prepare states $|x\rangle$ and $|y\rangle$ in the computational basis for Alice and Bob respectively, and then measure in the computational basis, where Alice's outcome is $a$ and Bob's is $b$. To check that this resulting assemblage is post-quantum, one can check that if we trace out Charlie's system, the correlations between Alice and Bob violate the CHSH inequality beyond Tsirelson's bound.}
Since we trace out Charlie's output system, the ancilla's state for Alice and Bob is given by 
\begin{align*}
\rho_{ancilla} = & \frac{\left(\ket{00}_{X_AX_B} + \ket{11}_{X_AX_B}\right) \left(\bra{00}_{X_AX_B} + \bra{11}_{X_AX_B}\right)}{2} \otimes \\
& \left(\alpha  \ket{00}\bra{00}_{W_AW_B} +(1-\alpha) \ket{11}\bra{11}_{W_AW_B}  \right) .
\end{align*}
For a choice of parameter $\alpha = \tfrac{1}{6}$, the correlations can be shown to give a value of $3$ for the CHSH inequality, which is larger than Tsirelson's bound $2\sqrt{2}$. Therefore, the map is definitely not localizable for that choice of $\alpha$. This channel can hence be used for Alice and Bob to channel-define a post-quantum assemblages on Charlie's subsystem.  {Not only this, but since almost quantum correlations cannot violate Tsirelson's bound either \cite{aqp}, then this assemblage is not even almost quantum, and thus the channel is not almost localizable.}

Finally, we discuss how certifying the post-quantumness of the Bell correlations that are channel-defined by a causal map is not a necessary condition for such a channel to be non-localizable. For this, consider the post-quantum assemblage given in the main result of \cite{pqsp}. We can construct a canonical channel that is not localizable and that channel-defines this post-quantum assemblage. Now, this particular assemblage has the property that the Bell correlations it produces are quantum, or more precisely, local \cite{pqsp}. Hence, we can construct a provably non-localizable channel that can only channel-define local correlations in Bell scenarios. 

\section{Teleportation and Buscemi non-locality} \label{se:gennonloc}

Inspired by the connection between forms of non-locality and quantum channels, in this section we initiate the study of \textit{post-quantum non-classical teleportation}, and \textit{post-quantum Buscemi non-locality}. Non-classical teleportation \cite{teleportation} and Buscemi non-locality \cite{Buscemi} (or semi-quantum non-locality\footnote{We do not use this terminology so as not to confuse between semi-quantum and post-quantum.}) have been introduced very recently within the quantum information community as generalisations of steering and Bell non-locality respectively. We will review each of these notions, and then relate their study to our study of channels, and this will naturally give a framework in which to study their post-quantum generalisations.

\subsection{Buscemi non-locality}

The pioneering work by Buscemi consisted in defining a semi-quantum non-local game and arguing that any entangled state is more useful than a separable one for winning at it \cite{Buscemi}. It should be noted that the kind of game Buscemi describes is subtly distinct to the one hinted by Leung, Toner and Watrous \cite{leung}. 
In this section, we will study the kind of non-locality that is witnessed in these games, and we begin by presenting the general setup. 

Consider $N$ parties, each of which has a quantum system with Hilbert space $\mathcal{K}_{j}$ and can prepare it in one out of $m$ quantum states. For each $j \leq N$, the states in which party $j$ may prepare their system are $\rho_{x_{j}}\in\mathcal{D}(\mathcal{K}_{j})$, with $x_{j}\in\{1,...,m\}$ being the classical label of the particular preparation. The parties then locally plug the system into some device (it can be a black box in analogy with Bell non-locality), and then receive a classical output from the device. Let $a_{j}\in\{1,...,d\}$ denote the classical output for the $j$th party, where $d$ is the total number of possible outputs the device can locally produce. 

Effectively, this whole process just described is a measurement on the preparations made by the $N$ parties. By means of a set of tomographically-complete preparations at each site, the parties can hence generate a description of this measurement. For convenience, we now introduce a new piece of terminology to describe this measurement.

\begin{defn}
In a Buscemi non-locality experiment, for a set of classical outputs $\{(a_{1},...,a_{N})\}$, a \textbf{distributed measurement} is $\{M_{a_{1},...,a_{N}}\in\mathcal{L}(\bigotimes_{j=1}^{N}\mathcal{K}_{j})\}_{a_{1},...,a_{N}}$ where $M_{a_{1},...,a_{N}}\geq 0$ and $\sum_{a_{1},...,a_{N}}M_{a_{1},...,a_{N}}=\id_{1,...,N}$.
\end{defn}

Given this distributed measurement, it is straightforward to generate conditional probabilities from its elements and certain state preparations $\{\rho_{x_{j}}\}$ as
\begin{equation*}
p(a_{1},...,a_{N}|x_{1},...,x_{N})=\Tr\{M_{a_{1},...,a_{N}}\rho_{x_{1}}\otimes...\otimes\rho_{x_{N}}\}.
\end{equation*}
For the purposes of Buscemi's original work, we need to define the set of distributed measurements that result from the set of local operations and shared randomness, which we call the \textit{local distributed measurements}.

\begin{defn}
A distributed measurement is \textbf{local} if there exist $N$ auxiliary systems $R_{j}$ for $j\in\{1,...,N\}$ with Hilbert spaces $\bigotimes_{j=1}^{N}\mathcal{H}_{R_{j}}$ such that 
\begin{equation*}
M_{a_{1},...,a_{N}}=\Tr_{R_{1},...,R_{N}}\{\rho_{R_{1},...,R_{N}}\bigotimes_{j=1}^{N}\Pi_{a_{j}}\},
\end{equation*}
where $\{\Pi_{a_{j}}\in\mathcal{L}(\mathcal{K}_{j}\otimes\mathcal{H}_{R_{j}})\}_{a_{j}}$ is a complete projective measurement, and $\rho_{R_{1},...,R_{N}}\in\mathcal{D}(\bigotimes\mathcal{H}_{R_{j}})$ is a separable state, i.e. for $|\phi_{\lambda}^{j}\rangle\in\mathcal{H}_{R_{j}}$
\begin{equation*}
\rho_{R_{1},...,R_{N}}=\sum_{\lambda}p_{\lambda}|\phi_{\lambda}^{1}\rangle\langle\phi_{\lambda}^{1}|_{R_{1}}\otimes|\phi_{\lambda}^{2}\rangle\langle\phi_{\lambda}^{2}|_{R_{2}}\otimes ... \otimes |\phi_{\lambda}^{N}\rangle\langle\phi_{\lambda}^{N}|_{R_{N}},
\end{equation*}
with $p_{\lambda}\geq 0$ and $\sum_{\lambda}p_{\lambda}=1$.
\end{defn}

Without loss of generality the local measurements can be taken to be projective, since the dimension of the Hilbert spaces $R_{j}$ is finite, but not constrained. Clearly, the state $\rho_{R_{1},...,R_{N}}$ could, in principle, be entangled, and thus we now define the set of \textit{quantum distributed measurements}.

\begin{defn}
A distributed measurement is \textbf{quantum} if there exist $N$ auxiliary systems $R_{j}$ for $j\in\{1,...,N\}$ with Hilbert spaces $\bigotimes_{j=1}^{N}\mathcal{H}_{R_{j}}$ such that 
\begin{equation*}
M_{a_{1},...,a_{N}}=\Tr_{R_{1},...,R_{N}}\{\rho_{R_{1},...,R_{N}}\bigotimes_{j=1}^{N}\Pi_{a_{j}}\},
\end{equation*}
where $\{\Pi_{a_{j}}\in\mathcal{L}(\mathcal{K}_{j}\otimes\mathcal{H}_{R_{j}})\}_{a_{j}}$ is a complete projective measurement, and $\rho_{R_{1},...,R_{N}}\in\mathcal{D}(\bigotimes\mathcal{H}_{R_{j}})$ is any quantum state, entangled or otherwise.
\end{defn}

The main result of Buscemi in \cite{Buscemi} can then be restated as: for every non-separable state $\rho_{R_{1},...,R_{N}}$, there exists a set of projective measurements $\{\Pi_{a_{j}}\in\mathcal{L}(\mathcal{K}_{j}\otimes\mathcal{H}_{R_{j}})\}_{a_{j}}$ such that the distributed measurement is not local. A corollary of this is that the set of local distributed measurements is strictly contained in the set of quantum distributed measurements.

In complete analogy with the study of Bell non-locality and steering, we can ask what are the most general distributed measurements that do not permit superluminal signalling. The following definition formalises the answer to this.

\begin{defn}
Given a bipartition $S_1 \cup S_2 = \{1, ..., N\}$ of $N$ parties where $S_{1}=\{i_1,..., i_s\}$ and $S_{2}=\{i_{s+1},... ,i_N\}$, a distributed measurement $\{M_{a_{1}...a_{N}}\}$ does not permit signalling across this bipartition if there exist sets of complete measurements $\{M_{a_{i_{1}}, ..., a_{i_s}}\in\mathcal{L}(\mathcal{K}_{i_1}\otimes...\otimes\mathcal{K}_{i_s})\}$ and $\{M_{a_{i_{s+1}}, ..., a_{i_N}}\in\mathcal{L}(\mathcal{K}_{i_{s+1}}\otimes...\otimes\mathcal{K}_{i_N})\}$ such that
\begin{eqnarray}
\sum_{a_{i_1}, ..., a_{i_s}}M_{a_{i_1}, ..., a_{i_s}}&=&\id_{i_1,...,i_s}\\
\sum_{a_{i_{s+1}}, ..., a_{i_N}}M_{a_{i_{s+1}}, ..., a_{i_N}}&=&\id_{i_{s+1},...,i_N}\\
\sum_{a_{i_{s+1}}, ..., a_{i_N}}M_{a_1, ..., a_N}&=&M_{a_{i_1}, ..., a_{i_s}}\otimes\id_{i_{s+1},...,i_N}\\
\sum_{a_{i_1}, ..., a_{i_s}}M_{a_1, ..., a_N}&=&\id_{i_1,...,i_s}\otimes M_{a_{i_{s+1}}, ..., a_{i_N}}.
\end{eqnarray}
A distributed measurement $\{M_{a_{1}...a_{N}}\}$ belongs to the set of \textbf{non-signalling distributed measurements} if and only if it does not permit signalling across any bipartition of the $N$ parties.
\end{defn}

If a distributed measurement is non-signalling but not quantum then we refer to this as \textit{post-quantum Buscemi non-locality}. We are not the first to describe the set of non-signalling distributed measurements, \v{S}upi\'{c}, Skrzypczyk, and Cavalcanti \cite{Skrzypczyk} defined this set in the bipartite setting, although the terminology ``distributed measurement" is of our creation. We believe we are, however, the first to point out the possibility of post-quantum Buscemi non-locality. Indeed, in the next section we point this out in a clear fashion.

\subsection{Buscemi non-locality via quantum channels}

In this section we take our channels-based perspective and apply it to the study of Buscemi non-locality. This indeed proceeds similarly to the study of steering and Bell non-locality.
The Buscemi non-locality scenario consists of $N$ parties, where party $j$th (for each $j\leq N$) acts on the Hilbert space $\mathcal{K}_{j}$, and outputs data $a_{j}\in\{1,...,d\}$. 
To study such a Buscemi scenario, consider $N$ parties to have input Hilbert spaces $\cH^{j}_{in}=\mathcal{K}_{j}$, and output Hilbert spaces $\cH^{j}_{out} = \cH_d$ for all $j$, where $\cH_d$ has dimension $d$. Denote by $\{\ket{a}\}_{a=1:d}$ an orthonormal basis of $\cH_d$. We now consider channels $\Lambda_{1...N}:\mathcal{L}(\bigotimes_{j}\cH^{j}_{in})\rightarrow\mathcal{L}(\mathcal{H}_{d}^{\otimes N})$ and relate channels of this form to distributed measurements.

\begin{defn}
A distributed measurement $\{M_{a_{1},...,a_{N}}\in\mathcal{L}(\bigotimes_{j=1}^{N}\mathcal{K}_{j})\}$ is \textbf{channel-defined} if there exists a channel $\Lambda_{1 ... N}:\mathcal{L}(\bigotimes_{j}\cH^{j}_{in}=\bigotimes_{j}\mathcal{K}_{j})\rightarrow\mathcal{L}(\bigotimes_{j=1}^{n}\cH^{j}_{out}=\mathcal{H}_{d}^{\otimes N})$, and some choice of orthonormal bases $\{\ket{a_{j}}\in\mathcal{H}_{d}\}_{a=1:d}$ for each party, such that
\begin{equation*}
M_{a_{1},...,a_{N}} = \Lambda^{\dagger}_{1 ... N} \left(\bigotimes_{k=1}^N \ket{a_k}\bra{a_k}\right)\,,
\end{equation*} 
where $\Lambda^{\dagger}_{1 ... N}$ is the dual of  $\Lambda_{1 ... N}$.
\end{defn}

\begin{figure}
\begin{center}
\begin{tikzpicture}
\node at (-7.5,0) {$M_{a_1, \ldots, a_N}$};
\node at (-6,0) {$\equiv$};

\shade[draw, thick, ,rounded corners, inner color=white,outer color=gray!50!white] (-5,-1) rectangle (-1,1) ;
\node at (-3,0) {$\Lambda$};
\draw[thick] (-4,-1.5) -- (-4,-1);
\draw[thick] (-2,-1.5) -- (-2,-1);
\draw[thick] (-4,1) -- (-4,1.5);
\draw[thick] (-2,1) -- (-2,1.5);

\draw[thick, rounded corners] (-4,2) -- (-3.5,1.5) -- (-4.5, 1.5) -- cycle;
\node at (-4,1.65) {$a_1$};
\draw[thick, rounded corners] (-2,2) -- (-1.5,1.5) -- (-2.5, 1.5) -- cycle;
\node at (-2,1.65) {$a_N$};
\node at (-3,1.7) {$\cdots$};
\node at (-3,-1.4) {$\cdots$};
\end{tikzpicture}
\end{center}
\caption{A distributed measurement $M_{a_1,\ldots,a_N}$ viewed as a causal channel $\Lambda$.}
\label{fig:distributed measurement}
\end{figure}
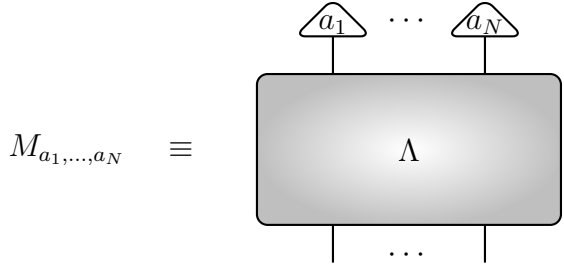

Fig.~\ref{fig:distributed measurement} presents a pictorial representation of  distributed measurements as quantum channels. Given this definition, as before, we can now give alternative definitions of local, quantum, and non-signalling distributed measurements. 
\begin{prop}\label{buslocprop} 
A distributed measurement is local if and only if there exists a local channel $\Lambda_{1...N}^{\mathsf{L}}:\mathcal{L}(\bigotimes_{j}\cH^{j}_{in})\rightarrow\mathcal{L}(\mathcal{H}_{d}^{\otimes N})$ such that the distributed measurement is channel-defined by $\Lambda_{1...N}^{\mathsf{L}}$.
\end{prop}

\begin{prop}\label{busqprop} 
A distributed measurement is quantum if and only if there exists a localizable channel $\Lambda_{1...N}^{\mathsf{Q}}:\mathcal{L}(\bigotimes_{j}\cH^{j}_{in})\rightarrow\mathcal{L}(\mathcal{H}_{d}^{\otimes N})$ such that the distributed measurement is channel-defined by $\Lambda_{1...N}^{\mathsf{Q}}$.
\end{prop}

\begin{prop} \label{busNSprop}
A distributed measurement is non-signalling if and only if there exists a causal channel $\Lambda_{1...N}^{\mathsf{C}}:\mathcal{L}(\bigotimes_{j}\cH^{j}_{in})\rightarrow\mathcal{L}(\mathcal{H}_{d}^{\otimes N})$ such that the distributed measurement is channel-defined by $\Lambda_{1...N}^{\mathsf{C}}$.
\end{prop}

Given these definitions of distributed measurements, it is straightforward to see that if each party were to prepare pure states from an orthonormal basis, then we recover the Bell non-locality setting. This then implies that local, quantum, and post-quantum non-locality implies a local, quantum, and post-quantum distributed measurement. A simple consequence of this is that the set of non-signalling distributed measurements is strictly larger than the set of quantum distributed measurements. For example, we can take the channel that produces the PR box correlations, and generate a post-quantum distributed measurement. However, do there exist post-quantum distributed measurements that will never produce post-quantum correlations? We leave this question as open, but in this direction we now define the set of almost quantum distributed measurements as the analogue of almost quantum correlations and assemblages.

\begin{defn}
A distributed measurement is \textbf{almost quantum} if there exists an almost localizable channel $\Lambda_{1...N}^{\tilde{\mathsf{Q}}}:\mathcal{L}(\bigotimes_{j}\cH^{j}_{in})\rightarrow\mathcal{L}(\mathcal{H}_{d}^{\otimes N})$ such that the distributed measurement is channel-defined by $\Lambda_{1...N}^{\tilde{\mathsf{Q}}}$.
\end{defn}

Given that the set of almost quantum correlations is larger than the set of quantum correlations, it follows that the set of almost quantum distributed measurements is larger than the set of quantum distributed measurements. In future work we will investigate whether this set has a useful characterisation in terms of semi-definite programming. Given such a characterisation, we should be able to address the question of whether post-quantum Buscemi non-locality implies post-quantum Bell non-locality.

\subsection{Non-classical teleportation}

The final scenario we consider is a generalisation of the steering scenario as first outlined by Cavalcanti, Skrzypczyk and \v{S}upi\'{c} \cite{teleportation}. In this scenario the original motivation was to consider two parties, and have one party ``teleport'' quantum information to the other, even if their resources are noisy. In particular, Alice is given one out of $m$ possible quantum states $\rho_{j}$ for $j\in\{1,...,m\}$, and produces some classical data (using a measurement on this input state and some other shared resource with Bob), and Bob has a quantum system upon which he can perform state tomography. {Importantly, the set of states $\{\rho_{j}\}$ is known to all parties, unlike in conventional teleportation where there is a single unknown state that is to be teleported.} {Once Bob knows the choice of state $\rho_{j}$ and the classical data (that resulted from Alice's measurement), Bob can deduce their (subnormalised) state conditioned upon this information. This is analogous to the assemblage in a steering scenario, which is a collection of (subnormalised) states conditioned on the classical information generated by the untrusted party.} {In the case where Alice and Bob share a maximally entangled state, Alice can make an entangled measurement on her input state $\rho_{j}$ and her half of this maximally entangled state. Conditioned on the outcome of the measurement, the state in Bob's laboratory will be $\rho_{j}$ with some unitary applied that depends on the outcome.} In general, given Bob's conditional (subnormalised) quantum state, they wish to establish if ``non-classical teleportation" took place.

We now extend this scenario to mimic closer the case of multipartite steering experiments. Consider $N$ parties, each of which has a quantum system with Hilbert space $\mathcal{K}_{j}$ and can prepare it in one out of $m$ quantum states. For each $j \leq N$, the states in which party $j$ may prepare their system are $\rho_{x_{j}}\in\mathcal{D}(\mathcal{K}_{j})$, with $x_{j}\in\{1,...,m\}$ being the classical label of the particular preparation.
In addition, consider another party, Bob, who has a quantum system with Hilbert space $\mathcal{K}_{B}$ and can perform quantum state tomography on his part of his system.
The first $N$ parties generate classical data locally from their system, and we denote by $a_{j}\in\{1,...,d\}$ the classical output obtained by party $j$, for each $j \leq N$.

{Since the first $N$ parties could prepare their input system in an arbitrary state before plugging it into their unknown device, they could each choose states from a tomographically complete set of states, as with the Buscemi non-locality scenario. That is, enough states that span the space $\mathcal{D}(\mathcal{K}_{j})$. The difference now in this scenario from the Buscemi non-locality scenario is that we have Bob's quantum system with Hilbert space $\mathcal{K}_{B}$, and upon which he can perform any quantum operation he likes. Therefore if we consider the whole process in terms of known quantum systems, we have the input Hilbert spaces $\mathcal{K}_{j}$ and an ``output" Hilbert space $\mathcal{K}_{B}$ in Bob's laboratory. Therefore the process of producing classical data and a (subnormalised) quantum state in Bob's laboratory can be described in terms of an object, which we call a \textit{teleportage}. This object can be characterised as a map from space of operators over $\bigotimes_{j=1}^{N}\mathcal{K}_{j}$ to the space of operators on $\mathcal{K}_{B}$, and it is characterised by the fact that a tomographically complete set of input states can be generated, and a tomographically complete measurement can be made on Bob's system.}

%

\begin{defn}
For a set of classical outputs $\{(a_{1},...,a_{N})\}$ and the Hilbert space $\mathcal{K}_{B}$ of Bob's system, a \textbf{teleportage} is an instrument $\{T_{a_{1},...,a_{N}}\in\mathcal{L}(\bigotimes_{j=1}^{N}\mathcal{K}_{j})\rightarrow\mathcal{L}(\mathcal{K}_{B})\}_{a_{1},...,a_{N}}$ such that, for all quantum states $\rho\in\mathcal{D}(\bigotimes_{j=1}^{N}\mathcal{K}_{j})$, $\Tr\{\sum_{a_{1},...,a_{N}}T_{a_{1},...,a_{N}}(\rho)\}=\Tr\{\rho\}$.
\end{defn}

We note that one can obtain an assemblage $\{\sigma_{a_{1},...,a_{N}|x_{1},...,x_{N}}\in\mathcal{D}(\mathcal{K}_{B}) \}$ for any set of input quantum states $\{\rho_{x_{j}}\}$, in the following way:
\begin{equation*}
\sigma_{a_{1},...,a_{N}|x_{1},...,x_{N}}=T_{a_{1},...,a_{N}}\left(\bigotimes_{j=1}^{N}\rho_{x_{j}}\right).
\end{equation*}
This is actually slightly distinct from the assemblages in the standard steering scenario, since the states $\rho_{x_{j}}$ have some quantum information, and so the classical labels $x_{j}$ do not capture everything. 

As with the study of Buscemi non-locality and steering, we can define the physically meaningful sets of teleportages.

\begin{defn}
A teleportage is \textbf{local} if there exist $N$ auxiliary systems $R_{j}$ for $j\in\{1,...,N\}$ with Hilbert spaces $\bigotimes_{j=1}^{N}\mathcal{H}_{R_{j}}$ such that 
\begin{equation*}
T_{a_{1},...,a_{N}}=\Tr_{1,...,N,R_{1},...,R_{N}}\{\rho_{R_{1},...,R_{N},B}\bigotimes_{j=1}^{N}\Pi_{a_{j}}\otimes\id_{B}\},
\end{equation*}
where $\{\Pi_{a_{j}}\in\mathcal{L}(\mathcal{K}_{j}\otimes\mathcal{H}_{R_{j}})\}_{a_{j}}$ is a complete projective measurement, and $\rho_{R_{1},...,R_{N},B}\in\mathcal{D}(\bigotimes\mathcal{H}_{R_{j}}\otimes\mathcal{K}_{B})$ is a separable state, i.e. for $|\phi_{\lambda}^{j}\rangle\in\mathcal{H}_{R_{j}}$ and $|\phi_{\lambda}^{B}\rangle\in\mathcal{K}_{B}$
\begin{equation*}
\rho_{R_{1},...,R_{N},B}=\sum_{\lambda}p_{\lambda}|\phi_{\lambda}^{1}\rangle\langle\phi_{\lambda}^{1}|_{R_{1}}\otimes|\phi_{\lambda}^{2}\rangle\langle\phi_{\lambda}^{2}|_{R_{2}}\otimes ... \otimes |\phi_{\lambda}^{N}\rangle\langle\phi_{\lambda}^{N}|_{R_{N}}\otimes|\phi_{\lambda}^{B}\rangle\langle\phi_{\lambda}^{B}|,
\end{equation*}
with $p_{\lambda}\geq 0$ and $\sum_{\lambda}p_{\lambda}=1$.
\end{defn}

\begin{defn}
A teleportage is \textbf{quantum} if there exist $N$ auxiliary systems $R_{j}$ for $j\in\{1,...,N\}$ with Hilbert spaces $\bigotimes_{j=1}^{N}\mathcal{H}_{R_{j}}$ such that 
\begin{equation*}
T_{a_{1},...,a_{N}}=\Tr_{1,...,N,R_{1},...,R_{N}}\{\rho_{R_{1},...,R_{N},B}\bigotimes_{j=1}^{N}\Pi_{a_{j}}\otimes\id_{B}\},
\end{equation*}
where $\{\Pi_{a_{j}}\in\mathcal{L}(\mathcal{K}_{j}\otimes\mathcal{H}_{R_{j}})\}_{a_{j}}$ is a complete projective measurement, and $\rho_{R_{1},...,R_{N},B}\in\mathcal{D}(\bigotimes\mathcal{H}_{R_{j}}\otimes\mathcal{K}_{B})$ is any quantum state, entangled or otherwise.
\end{defn}

\begin{defn}
Given a bipartition $S_1 \cup S_2 = \{1, ..., N\}$ of $N$ parties where $S_{1}=\{i_1,..., i_s\}$ and $S_{2}=\{i_{s+1},... ,i_N\}$, a teleportage $\{T_{a_{1},...,a_{N}}:\mathcal{L}(\bigotimes_{j=1}^{N}\mathcal{K}_{j})\rightarrow\mathcal{L}(\mathcal{K}_{B})\}$ does not permit signalling across this bipartition if there exist further teleportages $\{T_{a_{i_{1}}, ..., a_{i_s}}:\mathcal{L}(\bigotimes_{j=1}^{s}\mathcal{K}_{j})\rightarrow\mathcal{L}(\mathcal{K}_{B})\}$, $\{T_{a_{i_{s+1}}, ..., a_{i_N}}:\mathcal{L}(\bigotimes_{j=s+1}^{N}\mathcal{K}_{j})\rightarrow\mathcal{L}(\mathcal{K}_{B})\}$, and quantum state $\rho_{B}\in\mathcal{D}(\mathcal{H}_{B})$ such that
\begin{eqnarray*}
\sum_{a_{i_1}, ..., a_{i_s}}T_{a_{i_1}, ..., a_{i_s}}(\cdot)&=&\rho_{B}\\
\sum_{a_{i_{s+1}}, ..., a_{i_N}}T_{a_{i_{s+1}}, ..., a_{i_N}}(\cdot)&=&\rho_{B}\\
\sum_{a_{i_{s+1}}, ..., a_{i_N}}T_{a_1, ..., a_N}(\cdot)&=&T_{a_{i_1}, ..., a_{i_s}}(\cdot)\\
\sum_{a_{i_1}, ..., a_{i_s}}T_{a_1, ..., a_N}(\cdot)&=&T_{a_{i_{s+1}}, ..., a_{i_N}}(\cdot).
\end{eqnarray*}
A teleportage $\{T_{a_{1}...a_{N}}\}$ belongs to the set of \textbf{non-signalling teleportages} if and only if it does not permit signalling across any bipartition of the $N$ parties.
\end{defn}

We say a teleportage demonstrates \textit{post-quantum non-classical teleportation} if it is a non-signalling teleportage that is not a quantum teleportage. As far as we know, we are the first to define the set of non-signalling teleportages, in addition to introducing the nomenclature.

\subsection{Non-classical teleportation via quantum channels}\label{nonclassicaltel}

Finally, we look at non-classical teleportation through the lens of channels. Recall that we have a black box device with $N$ inputs for quantum systems, and $N$ classical outputs $a_{j}\in\{1,...,d\}$ for $j\in\{1,...,N\}$. For the $N$ ports of the black box device we associate input and output Hilbert spaces with each party such that $\mathcal{H}_{in}^{j}=\mathcal{K}_{j}$ and $\mathcal{H}_{out}^{j}=\mathcal{H}_{out}^{j'}=\mathcal{H}_{d}$ for all $j$, $j'\in\{1,...,N\}$, where $\mathcal{H}_{d}$ is a Hilbert space of dimension $d$. For Bob, we have an input and output Hilbert space denoted $\mathcal{H}_{in}^{B}$ and $\mathcal{H}_{out}^{B}$ respectively, where we have that $\mathcal{H}_{in}^{B}=\mathcal{H}_{out}^{B}=\mathcal{K}_{B}$. Therefore the channels of interest will be $\Lambda_{1...N}:\mathcal{L}(\bigotimes_{j}\cH^{j}_{in}\otimes\mathcal{K}_{B})\rightarrow\mathcal{L}(\mathcal{H}_{d}^{\otimes N}\otimes\mathcal{K}_{B})$, and as before we can define teleportages in terms of these channels.

\begin{defn}
A teleportage $\{T_{a_{1},...,a_{N}}:T_{a_{1},...,a_{N}}\in\mathcal{L}(\bigotimes_{j=1}^{N}\mathcal{K}_{j})\rightarrow\mathcal{L}(\mathcal{K}_{B})\}_{a_{1},...,a_{N}}$ is \textbf{channel-defined} if there exists a channel $\Lambda_{1 ... N}:\mathcal{L}(\bigotimes_{j}\cH^{j}_{in}=\bigotimes_{j}\mathcal{K}_{j}\otimes\mathcal{K}_{B})\rightarrow\mathcal{L}(\bigotimes_{j=1}^{n}\mathcal{H}_{d}^{\otimes N}\otimes\mathcal{K}_{B})$, and some choice of orthonormal bases $\{\ket{a_{j}}\in\mathcal{H}_{d}\}_{a=1:d}$ for each party, such that
\begin{equation*}
T_{a_{1},...,a_{N}}(\cdot) = \Tr_{out_{1},...,out_{N},B}\{\otimes_{k=1}^N \ket{a_k}\bra{a_k} \, \Lambda_{1 ... N} \left(\cdot\otimes|0\rangle\langle 0|_{B}\right) \}\, ,
\end{equation*} 
where $|0\rangle_{B}\in\mathcal{K}_{B}$.
\end{defn}

\begin{figure}
\begin{center}
\begin{tikzpicture}
\node at (-7.5,0) {$T_{a_1,\ldots,a_N}$};
\node at (-6,0) {$\equiv$};

\shade[draw, thick, ,rounded corners, inner color=white,outer color=gray!50!white] (-5,-1) rectangle (0.5,1) ;
\draw[thick] (-4,-1.5) -- (-4,-1);
\draw[thick] (-2,-1.5) -- (-2,-1);
\draw[thick] (-4,1) -- (-4,1.5);
\draw[thick] (-2,1) -- (-2,1.5);
\draw[thick] (-0.5,1) -- (-0.5,1.5);
\draw[thick] (-0.5,-1.5) -- (-0.5,-1);

\node at (-3,1.7) {$\cdots$};
\node at (-3,-1.4) {$\cdots$};

\draw[thick, rounded corners] (-0.5,-2) -- (0,-1.5) -- (-1, -1.5) -- cycle;
\node at (-0.5,-1.7) {$0$};

\draw[thick, rounded corners] (-4,2) -- (-3.5,1.5) -- (-4.5, 1.5) -- cycle;
\node at (-4,1.65) {$a_1$};
\draw[thick, rounded corners] (-2,2) -- (-1.5,1.5) -- (-2.5, 1.5) -- cycle;
\node at (-2,1.65) {$a_N$};
\node at (-2.5,0) {$\Lambda$};
\end{tikzpicture}
\end{center}
\caption{A Teleportage $T_{a_1,\ldots,a_N}$ viewed as a causal channel $\Lambda$.}
\label{fig:teleportage}
\end{figure}
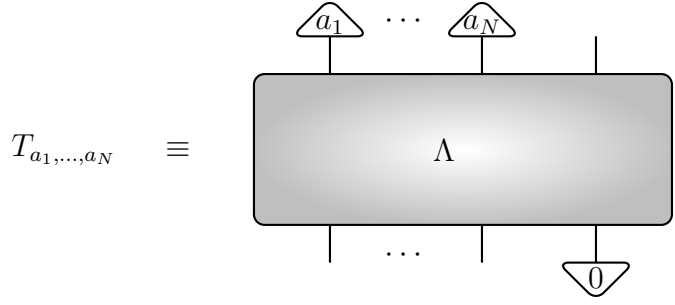

Fig.~\ref{fig:teleportage} depicts a teleportage as a quantum channel. Given this definition, as before we obtain the following results:

\begin{prop}\label{tellocprop} 
A teleportage is local if and only if there exists a local channel $\Lambda_{1...N}^{\mathsf{L}}:\mathcal{L}(\bigotimes_{j}\cH^{j}_{in}=\bigotimes_{j}\mathcal{K}_{j}\otimes\mathcal{K}_{B})\rightarrow\mathcal{L}(\bigotimes_{j=1}^{n}\mathcal{H}_{d}^{\otimes N}\otimes\mathcal{K}_{B})$ such that the teleportage is channel-defined by $\Lambda_{1...N}^{\mathsf{L}}$.
\end{prop}

\begin{prop}\label{telqprop} 
A teleportage is quantum if and only if there exists a localizable channel $\Lambda_{1...N}^{\mathsf{Q}}:\mathcal{L}(\bigotimes_{j}\cH^{j}_{in}=\bigotimes_{j}\mathcal{K}_{j}\otimes\mathcal{K}_{B})\rightarrow\mathcal{L}(\bigotimes_{j=1}^{n}\mathcal{H}_{d}^{\otimes N}\otimes\mathcal{K}_{B})$ such that the teleportage is channel-defined by $\Lambda_{1...N}^{\mathsf{Q}}$.
\end{prop}

\begin{prop}\label{telNSprop}
A teleportage is non-signalling if and only if there exists a causal channel $\Lambda_{1...N}^{\mathsf{C}}:\mathcal{L}(\bigotimes_{j}\cH^{j}_{in}=\bigotimes_{j}\mathcal{K}_{j}\otimes\mathcal{K}_{B})\rightarrow\mathcal{L}(\bigotimes_{j=1}^{n}\mathcal{H}_{d}^{\otimes N}\otimes\mathcal{K}_{B})$ such that the teleportage is channel-defined by $\Lambda_{1...N}^{\mathsf{C}}$.
\end{prop}

It should be clear that post-quantum steering implies post-quantum non-classical teleportation, since if an assemblage is post-quantum then it is channel-defined by a non-localizable channel, this non-localizable channel will then channel-define a teleportage that is post-quantum. 

For the study of steering we had an alternative characterisation of non-signalling assemblages in terms of a unitary representation. This result can be generalised to the set of non-signalling teleportages as follows. 

\begin{thm}\label{nonsigtel} \textbf{Unitary representation of non-signalling teleportages}\\
Let $\{T_{a_{1}...a_{N}}\}$ be a non-signalling teleportage. Then, the teleportage is channel-defined by a channel $\Lambda_{1...N,B}^{\mathsf{C}}:\mathcal{L}(\bigotimes_{j}\cH^{j}_{in}\otimes\mathcal{K}_{B})\rightarrow\mathcal{L}(\mathcal{H}_{d}^{\otimes N}\otimes\mathcal{K}_{B})$ if and only if there exist

\begin{itemize}
\item auxiliary systems $E$ and $E'$ with input and output Hilbert spaces, $\mathcal{H}^{E}_{in}$  and $\mathcal{H}^{E}_{out}$ for $E$, with $\mathcal{H}^{E'}_{in}=\mathcal{K}_{B}$ and  $\mathcal{H}^{E'}_{out}=\mathcal{K}_{B}$ for $E'$, that is the output Hilbert space of $E'$ and $B$ coincide;
\item quantum state $|R\rangle\in\mathcal{H}^{E}_{in}\otimes\mathcal{H}^{E'}_{in}$;
\item unitary operator $V:\bigotimes_{j}\cH^{j}_{in}\otimes\mathcal{H}^{E}_{in}\rightarrow\mathcal{H}_{d}^{\otimes N}\otimes\mathcal{H}^{E}_{out}$,
\end{itemize}
which produce a unitary representation of the channel $\Lambda_{1...N,B}^{\mathsf{C}}:\mathcal{L}(\bigotimes_{j}\cH^{j}_{in}\otimes\mathcal{K}_{B})\rightarrow\mathcal{L}(\mathcal{H}_{d}^{\otimes N}\otimes\mathcal{K}_{B})$ via
\begin{equation*}
\Lambda_{1...N,B}^{\mathsf{C}}(\cdot)=\Tr_{E_{out}}\{V(\Tr_{B_{in}}(\cdot)\otimes|R\rangle\langle  R|_{E,E'})V^{\dagger}\}.
\end{equation*}
Futhermore the unitary $V$ can be decomposed into a sequence of unitaries $U_{k,E}:\mathcal{H}_m\otimes\mathcal{H}^{E}_{1}\rightarrow\mathcal{H}_d\otimes\mathcal{H}^{E}_{2}$ for appropriately chosen Hilbert spaces $\mathcal{H}^{E}_{1}$ and $\mathcal{H}^{E}_{2}$, where for any given permutation $\pi$ of the set $\{1,...,N\}$, we have that
\begin{equation*}
V=U^{\pi}_{\pi(1),E}U^{\pi}_{\pi(2),E}...U^{\pi}_{\pi(N),E}
\end{equation*}
where $U^{\pi}_{k,E}$ is not necessarily the same as $U^{\pi'}_{k,E}$ for two different permutations $\pi$ and $\pi'$.
\end{thm}

Given this last result about non-signalling teleportages, we can actually generalise the GHJW theorem from the case of steering to the study of non-classical teleportation. 

\begin{cor}\label{genGHJW}
For $N=1$, all non-signalling teleportages are also quantum teleportages.
\end{cor}

That is, for the original context in which non-classical teleportation was studied, the bipartite setting, the No-Signalling principle is already enough to characterise exactly \textit{everything} that can be done quantum mechanically in the experiment. 

Finally, in analogy with everything that has gone before, we can define the set of almost quantum teleportages as follows.

\begin{defn} 
A teleportage is \textbf{almost quantum} if there exists an almost localizable channel $\Lambda_{1...N}^{\tilde{\mathsf{Q}}}:\mathcal{L}(\bigotimes_{j}\cH^{j}_{in}=\bigotimes_{j}\mathcal{K}_{j}\otimes\mathcal{K}_{B})\rightarrow\mathcal{L}(\bigotimes_{j=1}^{n}\mathcal{H}_{d}^{\otimes N}\otimes\mathcal{K}_{B})$ such that the teleportage is channel-defined by $\Lambda_{1...N}^{\tilde{\mathsf{Q}}}$.
\end{defn}

\subsection{Connections between all forms of post-quantum non-locality}

The relationship between entanglement, steering, and non-locality is now well-studied within the scope of quantum states. Since non-locality implies steering the non-trivial question is which entangled states demonstrate steering, but not non-locality. It has been shown that for all possible measurements on a quantum state, entanglement, steering, and non-locality are all inequivalent \cite{quintino}. In post-quantum non-locality, obviously we cannot automatically associate a process with measurements on a quantum state. Furthermore, due to the GHJW theorem, post-quantum steering cannot be demonstrated when there are only two parties, although post-quantum non-locality can be demonstrated with only two parties. Therefore, the relationship between post-quantum non-locality and post-quantum steering is somewhat subtle. The resolution is, of course, to consider a steering scenario with two (or more) uncharacterised parties and then generate correlations by making a measurement on Bob's system. If these correlations demonstrate post-quantum non-locality then this implies post-quantum steering, since the whole process cannot be associated with local measurements on a quantum system. However, post-quantum steering does not imply post-quantum non-locality, as demonstrated in Ref.~\cite{pqsp}.

The relationship between post-quantum non-locality and post-quantum Buscemi non-locality was discussed at length in the previous section. In particular, if we take a distributed measurement and for a combination of local preparations of states, we obtain post-quantum correlations, then this implies post-quantum Buscemi non-locality. As mentioned above, we leave it open whether there are post-quantum distributed measurements that do not result in post-quantum non-locality for all possible preparations. 

The next point to consider is the relationship between post-quantum Buscemi non-locality and post-quantum non-classical teleportation. As in the relationship between non-locality and steering, if we take a teleportage and make a measurement on Bob's system, we obtain a distributed measurement. If the distributed measurement is post-quantum, then clearly the teleportage was itself post-quantum. Likewise, one can obtain an assemblage from a teleportage by preparing certain quantum systems for each of the uncharacterised parties. If the assemblage demonstrates post-quantum steering then the teleportage was post-quantum. 
We see then that all these different forms of post-quantum non-locality are somehow related to each other as summarised in Fig.~\ref{fig:relations}.

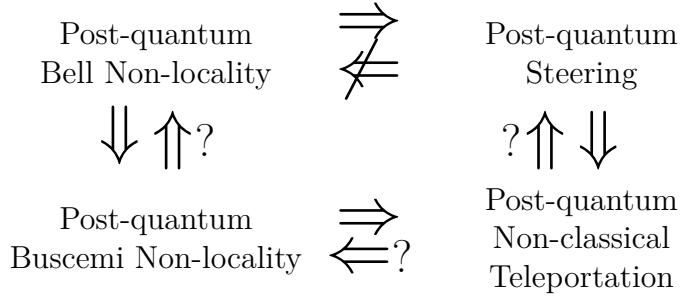
\begin{figure}
\begin{center}
\begin{tikzpicture}[scale=0.7]

\node[align=center] at (4,2) {Post-quantum \\ Steering};
\node[align=center] at (-4,2) {Post-quantum \\ Bell Non-locality};
\node[align=center] at (4,-1.5) {Post-quantum \\ Non-classical \\ Teleportation};
\node[align=center] at (-4,-1.5) {Post-quantum \\ Buscemi Non-locality};

\node[align=center] at (0,2) {\Huge{$\Rightarrow$} \\ \Huge{$\not\Leftarrow$}};
\node[align=center] at (0,-1.5) {\Huge{$\Rightarrow$} \\ \Huge{${\Leftarrow}$}\Large{?} };
\node[align=center] at (4.5,0.4) {\Huge{$\Downarrow$} };
\node[align=center] at (3.5,0.4) {\Huge{$\Uparrow$} };
\node[align=center] at (2.85,0.4) {\Large{?} };
\node[align=center] at (-4.5,0.5) {\Huge{$\Downarrow$} };
\node[align=center] at (-3.5,0.4) {\Huge{$\Uparrow$} };
\node[align=center] at (-2.9,0.4) {\Large{?} };

\end{tikzpicture}
\end{center}
\caption{Implication relations among the different forms of post-quantum non-locality. {Where there is a question mark next to an implication, this means that it is open whether there is an implication. One can also infer from the diagram that Post-quantum Bell Non-locality infers Post-quantum Non-classical Teleportation, but the reverse implication definitely does not hold.}}
\label{fig:relations}
\end{figure}

What is the relationship between post-quantum steering and post-quantum Buscemi non-locality? At first sight it seems difficult to relate the two, since in one scenario measurements are made, but preparations are made in the other. However, given our picture of non-locality from the perspective of quantum channels we can find a resolution. One way of generating an assemblage from a distributed measurement would be the following (see Fig.~\ref{fig:theconnection}): encode the classical inputs $\{|x_{j}\rangle\}$ as elements of an orthonormal basis $\mathcal{H}_{m}$ for $j\in\{1,...,N\}$, and take a localizable channel $\Lambda:\mathcal{L}(\mathcal{H}_{m}^{\otimes N}\otimes\mathcal{H}_{B})\rightarrow\bigotimes_{j=1}^{N}\mathcal{K}_{j}\otimes\mathcal{H}_{B}$ where $\mathcal{K}_{j}$ is the Hilbert space associated with the $j$th party's input to a distributed measurement, and $\mathcal{H}_{B}$ is an auxiliary Hilbert space associated with Bob's system; apply the channel $\Lambda$ to the input states $\{|x_{1}x_{2},...,x_{N}\rangle\}$, and then apply the distributed measurement to the systems now living in the Hilbert space $\bigotimes_{j=1}^{N}\mathcal{K}_{j}$, which results in an assemblage $\{\sigma_{a_{1}...a_{N}|x_{1}...x_{N}}\}$, i.e. operators acting on Bob's system. Since this extra element is a localizable channel, it will not introduce any post-quantum elements in its own right. Therefore, if we take a distributed measurement and turn it into an assemblage in this fashion, if the assemblage is post-quantum then the original distributed measurement itself was post-quantum.

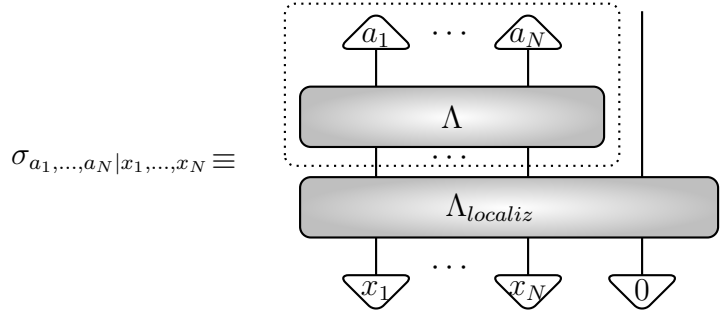
\begin{figure}
\begin{center}
\begin{tikzpicture}
\node at (-7.5,0) {$\sigma_{a_1,\ldots,a_N | x_1, \ldots, x_N}$};
\node at (-6,0) {$\equiv$};

\shade[draw, thick, ,rounded corners, inner color=white,outer color=gray!50!white] (-5,-1) rectangle (0.5,-0.2) ;
\node at (-2.5,-0.6) {$\Lambda_{localiz}$};
\shade[draw, thick, ,rounded corners, inner color=white,outer color=gray!50!white] (-5,0.2) rectangle (-1,1) ;
\node at (-3,0.6) {$\Lambda$};

\draw[thick] (-4,-1.5) -- (-4,-1);
\draw[thick] (-4,-0.2) -- (-4,0.2);
\draw[thick] (-2,-1.5) -- (-2,-1);
\draw[thick] (-4,1) -- (-4,1.5);
\draw[thick] (-2,1) -- (-2,1.5);
\draw[thick] (-2,-0.2) -- (-2,0.2);
\draw[thick] (-0.5,-0.2) -- (-0.5,2);
\draw[thick] (-0.5,-1.5) -- (-0.5,-1);

\node at (-3,1.7) {$\cdots$};
\node at (-3,-1.4) {$\cdots$};
\node at (-3,0.05) {$\cdots$};

\draw[thick, rounded corners] (-0.5,-2) -- (0,-1.5) -- (-1, -1.5) -- cycle;
\node at (-0.5,-1.7) {$0$};

\draw[thick, rounded corners] (-4,2) -- (-3.5,1.5) -- (-4.5, 1.5) -- cycle;
\node at (-4,1.65) {$a_1$};
\draw[thick, rounded corners] (-2,2) -- (-1.5,1.5) -- (-2.5, 1.5) -- cycle;
\node at (-2,1.65) {$a_N$};
\draw[thick, rounded corners] (-4,-2) -- (-3.5,-1.5) -- (-4.5, -1.5) -- cycle;
\node at (-4,-1.7) {$x_1$};
\draw[thick, rounded corners] (-2,-2) -- (-1.5,-1.5) -- (-2.5, -1.5) -- cycle;
\node at (-2,-1.7) {$x_N$};

\draw[thick, dotted ,rounded corners] (-5.2,-0.05) rectangle (-0.8, 2.1);

\end{tikzpicture}
\end{center}
\caption{A steering experiment constructed from a Buscemi nonlocality one. The distributed measurement is depicted within the dotted box.}
\label{fig:theconnection}
\end{figure}

Given Fig.~\ref{fig:relations}, we immediately see that post-quantum non-classical teleportation cannot imply post-quantum non-locality, since post-quantum steering does not imply post-quantum non-locality. That is, if post-quantum non-classical teleportation and post-quantum non-locality were equivalent then, post-quantum steering would imply post-quantum non-locality, which is not true. Furthermore, this also implies that either post-quantum Buscemi non-locality does not imply post-quantum non-locality, or post-quantum non-classical teleportation does not imply post-quantum Buscemi non-locality, or both. {In Fig.~\ref{fig:relations} we indicate these main open questions between all forms of post-quantum non-locality with a question mark next to the implication. To prove, for example, that Post-quantum Buscemi Non-locality does not imply Post-quantum Bell Non-locality, one would need to find a distributed measurement that cannot be realised via a localisable channel, yet this channel does not give post-quantum correlations, e.g. it could be local-limited. We conjecture that all four notions of post-quantum non-locality are inequivalent.}


\section{Discussion}

In this work we have shown that the study of post-quantum non-locality and steering can be seen as two facets of the study of quantum channels that do not permit superluminal signalling. We further showed that other scenarios can be readily approached within this scope, and hence initiated the study of post-quantum Buscemi non-locality and post-quantum non-classical teleportation. This general perspective allows us generate new examples of post-quantum steering, and allow us to generate novel kinds of non-signalling, but non-localizable channels. Furthermore, we have expanded the definition of almost quantum correlations to the domain of quantum channels (with no reference made to measurements), allowing us to recover almost quantum correlations and almost quantum assemblages in an appropriate domain.

Another channel-based perspective on the study of steering has led to so-called \textit{channel steering} \cite{pianichan}, as briefly mentioned in Section \ref{constchanass}. Channel steering is a generalisation of standard bipartite steering (involving Alice and Bob), where now there is a third party, Charlie, that inputs a quantum system into a channel, and Alice and Bob have systems that are the outputs of this channel. In a sense, this channel is then a broadcast channel. Alice can perform a measurement on her system to demonstrate to Bob that she can steer his output of the channel. Channel steering is distinct from the forms of non-locality considered here, but we can extend our channels to include this third party, and then study causal, but non-localizable, channels for post-quantum channel steering. We leave this for future work.

The characterisation of quantum non-locality is not only of foundational interest, but it is also of use in quantum information theory. In particular, characterising the set of quantum correlations is useful for device-independent quantum information, since it allows for a way to practically constrain what, say, a malicious agent can do in the preparation of devices. The study of Buscemi non-locality is of relevance to measurement-device-independent quantum information, hence this paradigm may profit from the characterisation of what is quantum mechanically allowed in the setup, with direct consequences regarding randomness certification and entanglement quantification \cite{Skrzypczyk,Branciard}.

One of the main open problems of this work is to further probe the relationships between the different kinds of post-quantum non-locality and steering. For example, as discussed, we know that post-quantum steering does not always imply post-quantum non-locality, but does post-quantum Buscemi non-locality imply post-quantum Bell non-locality? In Fig.~\ref{fig:relations} we summarised all the known relationships between all forms of post-quantum non-locality. We conjecture that all of these different notions of post-quantum non-locality are inequivalent, just as post-quantum non-locality is inequivalent to post-quantum steering.

On the way to proving our conjecture, it may be relevant to first study possible characterisations of all forms of almost quantum non-locality in terms of a semi-definite program. Such a connection was crucial in \cite{pqsp} when showing that post-quantum steering does not imply post-quantum non-locality.  Indeed, in this work we gave an interpretation in terms of quantum channels to the original SDP characterisation of almost quantum asseblages, thus giving a physical underpinning of this set. This interpretation allowed us to generalise the notion of almost quantum nonclassicality, hence now it would be interesting to relate back these general notions to SDPs when possible.

Since we have shown that post-quantum non-locality and steering are two aspects of a more general study of quantum channels, we hope this work motivates a \textit{resource theory of post-quantumness}. This resource theory could be approached from the point-of-view quantum channels, where the non-localizability of a channel is a resource. This relates directly to the study of zero-error communication with quantum channels \cite{duan}. Given this connection, we expect to find applications of post-quantum steering, just as we find that post-quantum non-locality can be used to trivialise communication complexity. Furthermore, we might be able to find applications of post-quantum Buscemi non-locality and post-quantum non-classical teleportation. Going further, there are other possibilities for non-locality scenarios. In particular, one can consider scenarios where all parties' outputs are quantum systems.

Our work could fit neatly within the study of quantum combs \cite{combs}, quantum strategies \cite{watrous}, quantum causal models \cite{Allen} and process matrices \cite{ognyan}. Indeed, since in certain scenarios in our work it is assumed that one party has access to a quantum system but the global system may be incompatible with quantum mechanics, it has a similar motivation to the study of indefinite causal order \cite{ognyan}. It would be interesting to see how our non-signalling processes interact with processes that could include signalling, and whether this interaction could be used to understand the structure of post-quantum non-locality. 

Last but not least, the resource theory of non-locality has been studied by only thinking of systems as black boxes. That is, one does not need to consider Hilbert spaces, or other features of quantum mechanics, but only consider the correlations associated with particular devices. The approach in this paper has been couched in the language of quantum theory. Can we consider generalising our framework further to consider trusted (and characterised) devices that may not be quantum, but are objects that can be described within a broad family of, say, generalised probabilistic theories \cite{gpt}? Indeed, steering has already been studied within the broad framework of these theories \cite{jencova,toytheory}. The study of non-signalling channels in general theories is left for future work and could then shed insight onto what is so special about quantum theory.

\section*{Acknowledgements}
We thank Marco Piani, Paul Skrzypczyk and Ivan \v{S}upi\'{c} for fruitful discussions and comments. MJH acknowledges funding from the EPSRC Networked Quantum Information Technologies (NQIT) Hub and from the EPSRC through grant Building Large Quantum States out of Light (EP/K034480/1). This research was supported by Perimeter Institute for Theoretical Physics. Research at Perimeter Institute is supported by the Government of Canada through the Department of Innovation, Science and Economic Development Canada and by the Province of Ontario through the Ministry of Research, Innovation and Science.

\appendix

\section{Relevant concepts from quantum channels}\label{ap:channels}

In this section we review families of multipartite quantum channels which are pertinent when discussing non-locality and steering. 
The general scenario we consider is that outlined in section \ref{se:channels}. 

Beckman, Gottesman, Nielsen, and Preskill \cite{Beck} considered quantum channels in such a set-up of space-like separated laboratories, especially those channels that are compatible with relativistic causality. That is, if two parties are space-like separated, the channel mapping their input states to their output states do not permit communication between them, called the \textit{causal channels}. There are multiple equivalent mathematical definitions capturing this concept \cite{SW}, and we shall present the definition of semicausal and causal channels.

\begin{defn}\textbf{Semicausal and Causal channels} \\
Given a multipartite system and a bipartition $S_A \cup S_B = \{1, \ldots, N\}$ of its elements, a map $\Lambda:\mathcal{L}(\mathcal{H}^{S_{A}}_{in}\otimes\mathcal{H}^{S_{B}}_{in})\rightarrow\mathcal{L}(\mathcal{H}^{S_{A}}_{out}\otimes\mathcal{H}^{S_{B}}_{out})$ is \textbf{semicausal} from $S_B$ to $S_A$, denoted $S_B \not\rightarrow S_A$, there exists a channel $\Gamma: \mathcal{L}(\cH^{S_A}_{in}) \rightarrow \mathcal{L}(\cH^{S_A}_{out})$, such that for all states $\rho_{in} \in \mathcal{D}(\mathcal{H}^{S_{A}}_{in}\otimes\mathcal{H}^{S_{B}}_{in})$, $\Tr_{S_B}\{ \Lambda[\rho_{in}] \} = \Gamma(\Tr_{S_B}[\rho_{in}])$. A map that is semicausal for all bipartitions is called \textbf{causal}. 
\end{defn}

For every channel there exists a unitary operator $U$ acting on a system and ancilla $E$, such that $\Lambda [\rho] = \Tr_E \left\{ U (\rho \otimes \ket{0} \bra{0}_E) U^\dagger \right\}$ for some state $|0\rangle_{E}\in\mathcal{H}_{E}$. What is the form of a unitary dilation of (semi)causal channels? For bipartite semicausal maps, works by Schumacher and Westmoreland \cite{SW}, D'Ariano et al \cite{Pao}, and Piani et al \cite{Piani} provide the following characterisation:
\begin{thm}\label{lem:bip} \textbf{Unitary representation for bipartite semicausal channels {\cite{SW,Pao,Piani}}.}\\
Let $\Lambda_{AB}$ be a CPTP map, with $B \not\rightarrow A$. Then, there exists an auxiliary system $E$ with input, intermediate and output Hilbert spaces, $\mathcal{H}^{E}_{in}$, $\mathcal{H}^{E}_{int}$, and $\mathcal{H}^{E}_{out}$ respectively, and  quantum state $|0\rangle\in\mathcal{H}^{E}_{in}$, producing a unitary representation of the map as $\Lambda_{AB} [\rho] = \Tr_E \left\{ U (\rho \otimes \ket{0} \bra{0}_E) U^\dagger \right\}$, where  the unitary $U$ can be decomposed as a unitary acting on $A$ and $E$ followed by a unitary acting on $B$ and $E$ alone. That is, 
\begin{equation*}
U = U_{EB} \, U_{AE}\,,
\end{equation*}
for $U_{EB}:\cH_{int}^{E}\otimes\mathcal{H}_{in}^{B}\rightarrow\cH_{out}^{E}\otimes\mathcal{H}_{out}^{B}$ and $U_{AE}:\cH_{in}^{A}\otimes\mathcal{H}_{in}^{E}\rightarrow\cH_{out}^{A}\otimes\mathcal{H}_{int}^{E}$.
\end{thm}
The statement of this theorem is depicted in Fig.~\ref{f:causal}. When the map is fully causal, then there exist both a decomposition in terms of $U=U_{EB} \, U_{AE}$ and one in terms of $V=V_{AE} \, V_{EB}$, where the $U_j$ and $V_j$ are unitaries, for $j\in\{AE,EB\}$. 

\begin{figure}
\begin{center}
\subfigure[Semicausal map]{

\begin{tikzpicture}
\shade[draw, thick, ,rounded corners, inner color=white,outer color=gray!50!white] (-5,-1) rectangle (-1,1) ;
\node at (-3,0) {$\Lambda_{B \not\rightarrow A}$};
\draw[thick] (-4,-2) -- (-4,-1);
\draw[thick] (-2,-2) -- (-2,-1);
\draw[thick] (-4,1) -- (-4,2);
\draw[thick] (-2,1) -- (-2,2);
\node at (-4,2.5) {Alice};
\node at (-2,2.5) {Bob};

\node at (0,0) {$\equiv$};
\draw[thick, rounded corners] (1,-1) rectangle (5,1);
\draw[thick, rounded corners] (3,-0.99) -- (3.5,-0.49) -- (2.5, -0.49) -- cycle;
\node at (3,-0.69) {$E$};
\draw[thick, rounded corners] (1.5,-0.39) rectangle (3.5,0.11);
\draw[thick, rounded corners] (2.5,0.21) rectangle (4.5,0.71);
\node at (2.5,-0.15) {$U_{AE}$};
\node at (3.5,0.46) {$U_{EB}$};

\draw[thick] (2,-2) -- (2,-0.39);
\draw[thick] (2,0.11) -- (2,2);
\draw[thick] (4,-2) -- (4,0.21);
\draw[thick] (4,0.71) -- (4,2);

\draw[thick] (3,-0.49) -- (3,-0.39);
\draw[thick] (3,0.11) -- (3,0.21);
\draw[thick] (3,0.71) -- (3,0.85);
\draw[thick] (2.8,0.85) -- (3.2,0.85);
\draw[thick] (2.85,0.9) -- (3.15,0.9);
\draw[thick] (2.9,0.95) -- (3.1,0.95);
\node at (2,2.5) {Alice};
\node at (4,2.5) {Bob};

\end{tikzpicture}}

\subfigure[Causal map]{
\begin{tikzpicture}

\shade[draw, thick, ,rounded corners, inner color=white,outer color=gray!50!white] (-5,-1) rectangle (-1,1) ;
\node at (-3,0) {$\Lambda_{B \not\leftrightarrow A}$};
\draw[thick] (-4,-2) -- (-4,-1);
\draw[thick] (-2,-2) -- (-2,-1);
\draw[thick] (-4,1) -- (-4,2);
\draw[thick] (-2,1) -- (-2,2);
\node at (-4,2.5) {Alice};
\node at (-2,2.5) {Bob};
\node at (-0.5,0) {$\equiv$};

\draw[thick, rounded corners] (0,-1) rectangle (4,1);
\draw[thick, rounded corners] (2,-0.99) -- (2.5,-0.49) -- (1.5, -0.49) -- cycle;
\node at (2,-0.69) {$E$};
\draw[thick, rounded corners] (0.5,-0.39) rectangle (2.5,0.11);
\draw[thick, rounded corners] (1.5,0.21) rectangle (3.5,0.71);
\node at (1.5,-0.15) {$U_{AE}$};
\node at (2.5,0.46) {$U_{EB}$};

\draw[thick] (1,-2) -- (1,-0.39);
\draw[thick] (1,0.11) -- (1,2);
\draw[thick] (3,-2) -- (3,0.21);
\draw[thick] (3,0.71) -- (3,2);

\draw[thick] (2,-0.49) -- (2,-0.39);
\draw[thick] (2,0.11) -- (2,0.21);
\draw[thick] (2,0.71) -- (2,0.85);
\draw[thick] (1.8,0.85) -- (2.2,0.85);
\draw[thick] (1.85,0.9) -- (2.15,0.9);
\draw[thick] (1.9,0.95) -- (2.1,0.95);

\node at (1,2.5) {Alice};
\node at (3,2.5) {Bob};

\node at (4.5,0) {$\equiv$};
\draw[thick, rounded corners] (5,-1) rectangle (9,1);
\draw[thick, rounded corners] (7,-0.99) -- (7.5,-0.49) -- (6.5, -0.49) -- cycle;
\node at (7,-0.69) {$E$};
\draw[thick, rounded corners] (5.5,0.21) rectangle (7.5,0.71);
\draw[thick, rounded corners] (6.5,-0.39) rectangle (8.5,0.11);
\node at (7.5,-0.15) {$V_{EB}$};
\node at (6.5,0.46) {$V_{AE}$};

\draw[thick] (8,-2) -- (8,-0.39);
\draw[thick] (8,0.11) -- (8,2);
\draw[thick] (6,-2) -- (6,0.21);
\draw[thick] (6,0.71) -- (6,2);

\draw[thick] (7,-0.49) -- (7,-0.39);
\draw[thick] (7,0.11) -- (7,0.21);
\draw[thick] (7,0.71) -- (7,0.85);
\draw[thick] (6.8,0.85) -- (7.2,0.85);
\draw[thick] (6.85,0.9) -- (7.15,0.9);
\draw[thick] (6.9,0.95) -- (7.1,0.95);

\node at (6,2.5) {Alice};
\node at (8,2.5) {Bob};
\end{tikzpicture}}
\end{center}
\caption{(a) A semicausal CPTP channel, where Bob does not signal Alice, has an equivalent representation where Alice performs a unitary operation $U_{AE}$ on her system plus an ancilla, and afterwards Bob performs a unitary operation $U_{EB}$ on his system plus the shared ancilla. (b) A causal map, where there exist a unitary decomposition for each ordering of the parties. }
\label{f:causal}
\end{figure}
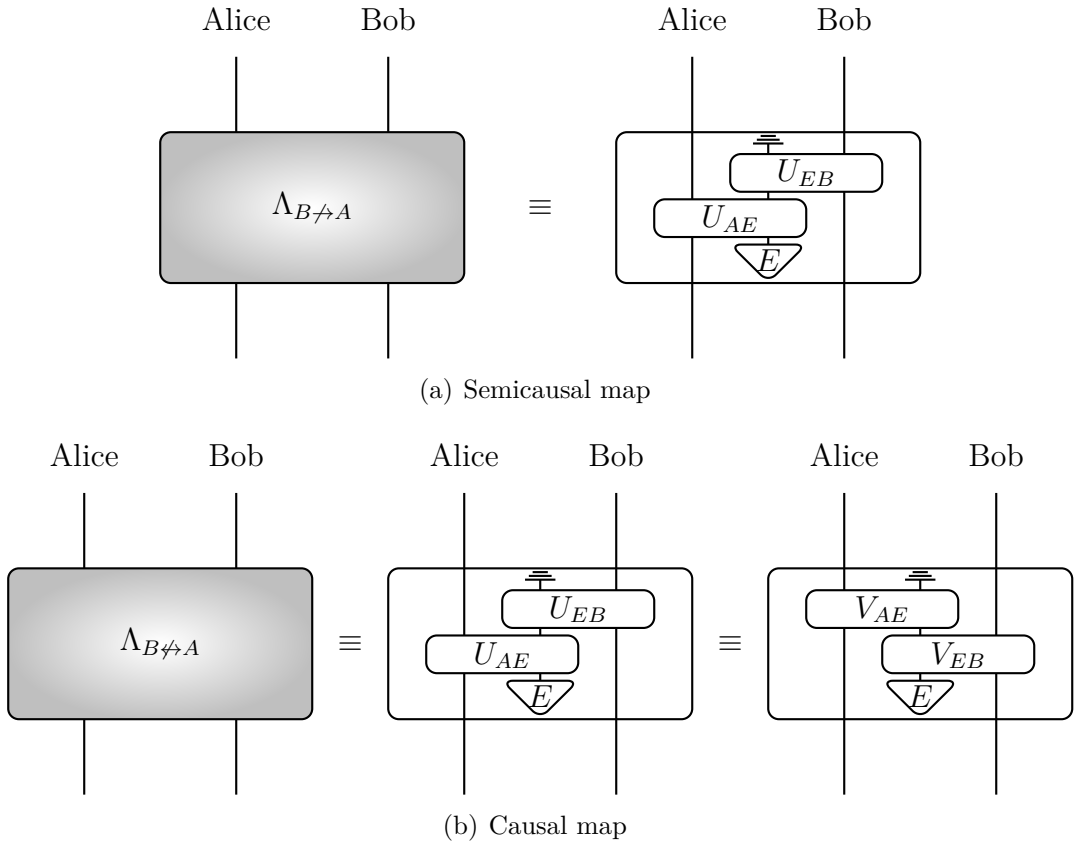
\begin{figure}
\begin{center}
\begin{tikzpicture}
\shade[draw, thick, ,rounded corners, inner color=white,outer color=gray!50!white] (-5,-1) rectangle (-1,1) ;
\node at (-3,0) {$\Lambda_{localizable}$};
\draw[thick] (-4,-2) -- (-4,-1);
\draw[thick] (-2,-2) -- (-2,-1);
\draw[thick] (-4,1) -- (-4,2);
\draw[thick] (-2,1) -- (-2,2);
\node at (-4,2.5) {Alice};
\node at (-2,2.5) {Bob};

\node at (0,0) {$\equiv$};
\draw[thick, rounded corners] (1,-1) rectangle (5,1);
\draw[thick, rounded corners] (3,-0.95) -- (3.9,-0.45) -- (2.1, -0.45) -- cycle;
\node at (3,-0.66) {$E$};

\draw[thick, rounded corners] (1.5,-0.1) rectangle (2.9,0.4);
\draw[thick, rounded corners] (3.1,-0.1) rectangle (4.5,0.4);
\node at (2.3,0.1) {$U_{AE}$};
\node at (3.7,0.1) {$U_{EB}$};

\draw[thick] (2,-2) -- (2,-0.1);
\draw[thick] (2,0.4) -- (2,2);
\draw[thick] (4,-2) -- (4,-0.1);
\draw[thick] (4,0.4) -- (4,2);

\draw[thick] (2.5,-0.45) -- (2.5,-0.1);
\draw[thick] (3.5,-0.45) -- (3.5,-0.1);

\draw[thick] (2.5,0.4) -- (2.5,0.6);
\draw[thick, xshift=-0.5cm, yshift=-0.25cm] (2.9,0.99) -- (3.1,0.99);
\draw[thick, xshift=-0.5cm, yshift=-0.25cm] (2.85,0.92) -- (3.15,0.92);
\draw[thick, xshift=-0.5cm, yshift=-0.25cm] (2.8,0.85) -- (3.2,0.85);

\draw[thick] (3.5,0.4) -- (3.5,0.6);
\draw[thick, xshift=0.5cm, yshift=-0.25cm] (2.9,0.99) -- (3.1,0.99);
\draw[thick, xshift=0.5cm, yshift=-0.25cm] (2.85,0.92) -- (3.15,0.92);
\draw[thick, xshift=0.5cm, yshift=-0.25cm] (2.8,0.85) -- (3.2,0.85);

\node at (2,2.5) {Alice};
\node at (4,2.5) {Bob};

\end{tikzpicture}
\end{center}
\caption{A bipartite localizable channel, decomposed as local unitaries acting on a shared ancilla.}
\label{f:localiz}
\end{figure}
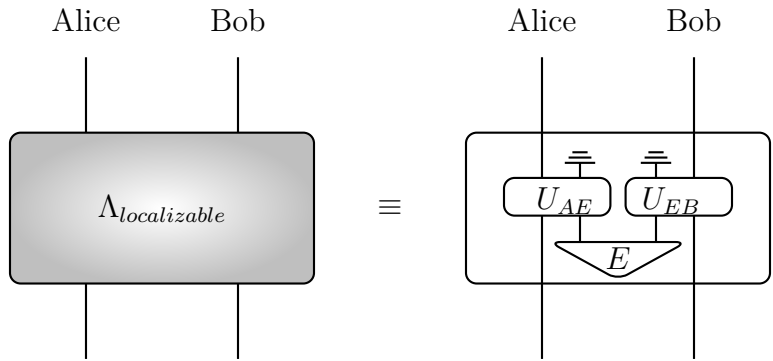

This result can be generalised to multipartite causal channels. {We now use notation where parties are labelled by numbers going from $1$ to $N$.} Given a multipartite causal map $\Lambda_{1 \ldots N}$, there exist unitary operators $\{U_{\pi}\}_{\pi}$ acting on local systems plus a global auxiliary system $E$ such that $\Lambda_{1 \ldots N} [\rho] = \Tr_E \left\{ U_{\pi} (\rho \otimes \ket{0} \bra{0}_E) U_{\pi}^\dagger \right\}$, where the unitary $U_{\pi}$ has the form of $U_{\pi}= \prod_{k=1}^N U_{{\pi}(k)E} $ and $\pi$ is a permutation of the parties $\{1,2,\ldots ,N\}$. The proof is presented in Appendix \ref{ap:multicausal}. 

A particular class of causal channels is the class of \textit{localizable channels} \cite{Beck}. These are channels implemented by local operations performed by each party on their input and a share of a quantum ancilla (see Fig.~\ref{f:localiz}). We formalise this definition below.

\begin{defn}\label{def:loc} \textbf{Localizable channels} \\
A causal channel $\Lambda_{1 \ldots N}$ is \textbf{localizable} if there exists an $N$-partite ancilla system $R$, with Hilbert spaces $\mathcal{H}_{R}:=\cH_{R_1}\otimes\cH_{R_2}\otimes...\otimes\cH_{R_N}$ with $R_{j}$  labelling the $j$th sub-system of $R$, and state $\sigma_{R}\in\mathcal{D}(\mathcal{H}_{R})$ such that, for all states $\rho\in\mathcal{D}(\bigotimes_{j=1}^{N}\mathcal{H}_{in}^{j})$,
\begin{equation*}
\Lambda_{1 \ldots N} [\rho] = \otimes_{k=1}^N \Lambda_{R_k} \, [\rho \otimes \sigma_{R}],
\end{equation*}
where $\Lambda_{R_k}:\mathcal{L}(\mathcal{H}_{in}^{k}\otimes\mathcal{H}_{R_{k}})\rightarrow\mathcal{L}(\mathcal{H}_{out}^{k})$. 
\end{defn}

{Notice that in this definition for localizable channels, the ancilla $\sigma_{R}$ is the same for all inputs to the channel $\Lambda_{1 \ldots N}$.}

It is known that, already for bipartite systems, there exist channels that are causal but not localizable \cite{Beck}. Furthermore, there are examples that are not entanglement-breaking \cite{Pao}, unlike the example given in Ref. \cite{Beck}. 

Just as we considered the causal channels in terms of unitaries, we can consider localizable channels in terms of unitary operators. Since there are only local maps in the localizable channels, it is straightforward to dilate each of these maps if we increase the Hilbert space dimension of local systems $R_{j}$ in the ancilla. This gives the following equivalent definition of localizable channels. 

\begin{defn} \textbf{Unitary representation of localizable channels.} \\
A causal channel $\Lambda_{1 \ldots N}$ is localizable if and only if there exists an $N$-partite ancilla system $E$ with input and output Hilbert spaces $\mathcal{H}_{in}^{E}:=\cH_{in}^{E_1}\otimes\cH_{in}^{E_2}\otimes...\otimes\cH_{in}^{E_N}$ and $\mathcal{H}_{out}^{E}:=\cH_{out}^{E_1}\otimes\cH_{out}^{E_2}\otimes...\otimes\cH_{out}^{E_N}$,  with $j$  labelling the $j$th sub-system of $E$, and state $|\psi\rangle_{E}\in\mathcal{H}_{in}^{E}$ such that, for all states $\rho\in\mathcal{D}(\bigotimes_{j=1}^{N}\mathcal{H}_{in}^{j})$,
\begin{equation*}
\Lambda_{1 \ldots N} [\rho] = \Tr_{E}\left\{V \, (\rho \otimes |\psi\rangle\langle\psi|_{E}) V^{\dagger}\right\},
\end{equation*}
where $V=\otimes_{k=1}^N U_{kE_{k}}$ for unitary operators $U_{kE_{k}}:\mathcal{H}_{in}^{k}\otimes\mathcal{H}_{in}^{E_{k}}\rightarrow\mathcal{H}_{out}^{k}\otimes\mathcal{H}_{out}^{E_{k}}$.
\end{defn} 

In addition to the above unitary representation, there exists another equivalent representation. This representation does not make reference to a tensor product structure in the ancilla, instead the unitaries in a unitary representation of a causal channel are independent of each other, in a particular sense. {Now we have a global ancilla living in Hilbert space $\mathcal{H}_{E}$ and local ancillae $E_{k}$ for each $k$th party, with input and output Hilbert spaces $\cH_{in}^{E_k}$ and $\cH_{out}^{E_k}$ respectively. The local ancillae are introduced so that everything can remain unitary. Therefore, the total input and output Hilbert spaces of all the ancillae are $\mathcal{H}_{E}\otimes\cH_{in}^{E_1}\otimes\cH_{in}^{E_2}\otimes...\otimes\cH_{in}^{E_N}$ and $\mathcal{H}_{E}\otimes\cH_{out}^{E_1}\otimes\cH_{out}^{E_2}\otimes...\otimes\cH_{out}^{E_N}$ respectively. }

\begin{defn}\label{def:locuni}  \textbf{Commuting unitary representation of localizable channels.} \\
A causal channel $\Lambda_{1 \ldots N}$ is localizable if and only if there exists {a global ancilla system $E$ with Hilbert space $\mathcal{H}_{E}$, a local ancilla system $E_{k}$ for each $k$th party, with input and output Hilbert spaces $\cH_{in}^{E_k}$ and $\cH_{out}^{E_k}$ respectively, and state $|\psi\rangle_{E}\in\mathcal{H}_{E}\otimes\cH_{in}^{E_1}\otimes\cH_{in}^{E_2}\otimes...\otimes\cH_{in}^{E_N}$ such that, for all states $\rho\in\mathcal{D}(\bigotimes_{j=1}^{N}\mathcal{H}_{in}^{j})$,}
\begin{equation*}
\Lambda_{1 \ldots N} [\rho] = \Tr_{EE_1 \ldots E_N}\left\{\prod_{j=1}^{N}U_{j E} \, (\rho \otimes |\psi\rangle\langle\psi|_{E}) \prod_{k=0}^{N-1}U_{N-k E}^{\dagger}\right\},
\end{equation*}
where $U_{kE}:\mathcal{H}_{in}^{k}\otimes\mathcal{H}_{E}\otimes\cH_{in}^{E_k}\rightarrow\mathcal{H}_{out}^{k}\otimes\mathcal{H}_{E}\otimes\cH_{out}^{E_k}$ is a unitary operator for all $k$, such that, for any permutation $\pi$ on the set $\{1,2,...,N\}$, $\prod_{k=1}^{N}U_{k E}=\prod_{k=1}^{N}U_{\pi({k}) E}$.
\end{defn} 

Since all the Hilbert spaces in this work are taken to be finite dimensional, these two unitary representations of localizable channels are equivalent. This can be shown by a straightforward extension of Lemma 4.1 in Ref. \cite{doherty}. It should be remarked upon that if we were to allow for infinite dimensional Hilbert spaces, then these two definitions will not be equivalent, as pointed out by Cleve, Liu and Paulsen \cite{cleve}. In full generality, since the first unitary representation implies the commuting unitary representation, one could then take the commuting unitary representation to be the most general definition of localizable channels when allowing for infinite dimension Hilbert spaces. 

Finally, from the point-of-view of non-locality and steering, the set of \textit{local} channels is of interest. 

\begin{defn}\textbf{Local channel.}\\ A channel is {local} if it is localizable, but with the additional constraint that the ancilla state $\sigma_{R}$ is a separable state, i.e.  
\begin{equation*}
\sigma_{R}=\sum_{\lambda}p_{\lambda}|\phi_{\lambda}^{1}\rangle\langle\phi_{\lambda}^{1}|_{R_{1}}\otimes|\phi_{\lambda}^{2}\rangle\langle\phi_{\lambda}^{2}|_{R_{2}}\otimes ... \otimes |\phi_{\lambda}^{N}\rangle\langle\phi_{\lambda}^{N}|_{R_{N}},
\end{equation*}
for $|\phi_{\lambda}^{j}\rangle_{R_{j}}\in\mathcal{H}_{R_{j}}$.
\end{defn} 

It can be readily seen that localizable channels are more general than local channels. For example, the latter set of channels cannot be used to generate entanglement between two parties, but a localizable channel can. We can now summarise all the information about causal channels in the following theorem.

\begin{thm}\label{channelhierarchy} Let $\mathsf{C}$, $\tilde{\mathsf{Q}}$, $\mathsf{Q}$, and $\mathsf{L}$ be the set of causal, almost localizable, localizable, and local channels respectively, then we have that $\mathsf{C}\supsetneq\tilde{\mathsf{Q}}\supsetneq\mathsf{Q}\supsetneq\mathsf{L}$.
\end{thm}

A natural question is given a channel $\Lambda$, can we decide if it belongs to $\mathsf{C}$, $\tilde{\mathsf{Q}}$, $\mathsf{Q}$, or $\mathsf{L}$. We first restrict to the bipartite setting, and we suppose that $\Lambda$ is given a convenient representation, such as the Choi-Jamio{\l}kowski representation \cite{choi}. That is, for a bipartite setting channel $\Lambda$, the \textit{Choi state} $\Omega\in\mathcal{D}(\mathcal{H}^{1}_{in}\otimes\mathcal{H}^{1}_{out}\otimes\mathcal{H}^{2}_{in}\otimes\mathcal{H}^{2}_{out})$, such that $\Omega:=(\Lambda\otimes\id_{\mathcal{H}^{1}_{in}}\otimes\id_{\mathcal{H}^{2}_{in}})\left(|\Phi^{+}\rangle\langle\Phi^{+}|\right)$, for 
\begin{equation*}
|\Phi^{+}\rangle=\frac{1}{\sqrt{ d_{\mathcal{H}^{1}_{in}}d_{\mathcal{H}^{2}_{in}}}}\sum_{j}|j\rangle_{\mathcal{H}^{1}_{in}}|j\rangle_{\mathcal{H}^{1}_{in}}\sum_{k}|k\rangle_{\mathcal{H}^{2}_{in}}|k\rangle_{\mathcal{H}^{2}_{in}},
\end{equation*}
with $d_{\mathcal{H}}$ being the dimension of the Hilbert space $\mathcal{H}$. The Choi state $\Omega$ is positive semi-definite if only if $\Lambda$ is a channel. In order to decide if the channel $\Lambda$ is causal, we can use the definition of causal channels along with this property of Choi states to state the following result.

\begin{prop} A channel $\Lambda$ is causal if and only if its Choi state $\Omega$ satisfies the following:
\begin{enumerate}
\item $\exists$ a density matrix $\Sigma_{1}\in\mathcal{D}(\mathcal{H}^{1}_{in}\otimes\mathcal{H}^{1}_{out})$ such that $\textrm{tr}_{\mathcal{H}^{2}_{out}}\Omega=\frac{1}{d_{\mathcal{H}^{2}_{in}}}\Sigma_{1}\otimes\id_{\mathcal{H}^{2}_{in}}$, and
\item $\exists$ a density matrix $\Sigma_{2}\in\mathcal{D}(\mathcal{H}^{2}_{in}\otimes\mathcal{H}^{2}_{out})$ such that $\textrm{tr}_{\mathcal{H}^{1}_{out}}\Omega=\frac{1}{d_{\mathcal{H}^{1}_{in}}}\id_{\mathcal{H}^{1}_{in}}\otimes\Sigma_{2}$.
\end{enumerate}
\end{prop}

A consequence of this result is that there exist positive semi-definite matrices $\Sigma_{1}$ and $\Sigma_{2}$ such that the conditions of the proposition are satisfied. In other words, one can decide whether a channel is in the set $\mathsf{C}$ using a semi-definite program (cf. Ref. \cite{gutoski}), and so one can efficiently decide this problem.

Deciding whether a channel belongs in the sets of $\mathsf{Q}$ and $\mathsf{L}$ is not as easily resolved as the case for $\mathsf{C}$. For example, while channels in $\mathsf{L}$ will have a Choi state $\Omega$ that is not entangled across the partition of party $1$ and party $2$'s respective Hilbert spaces, the converse is not true. That is, there are channels in $\mathsf{C}$ (but not in $\mathsf{Q}$) whose Choi state is also not entangled across this partition \cite{Beck}. Furthermore, deciding membership in $\mathsf{L}$ (even up to some error) is \textsc{NP}-hard, although it is possible to find conditions to test whether a channel is in the set $\mathsf{L}$ \cite{gutoski}.

For case of deciding membership in $\mathsf{Q}$ (even up to some error), this is problem is also \textsc{NP}-hard \cite{gutoski}. In addition to this, the set of localizable channels is not closed \cite{leung}. Gutoski has given a criterion for deciding if a channel is in $\mathsf{Q}$, which is somewhat analogous to the condition of complete positivity for channels. However, in general, there is no known way of deciding in finite time whether this criterion is satisfied. Indeed, as we will point out, this problem is deeply related to the problem of deciding whether certain correlations in a Bell test can be realised by local measurements on a quantum state; a problem with deep connections to open problems in mathematics \cite{tsirelsonproblem}. For the case of deciding if a channel belongs to the set $\tilde{\mathsf{Q}}$, we leave this to future work.

\section{Bell non-locality}\label{ap:nonloc}

A traditional Bell experiment (sometimes called a `Bell scenario') consists of $N$ distant parties, each of them having access to a share of a physical system. These  parties input (in a space-like separated manner) classical data into their device (labelled as $x_i\in\{1,...,m\}$ for party $i$), and obtain outputs (labelled as $a_i\in\{1,...,d\}$ for party $i$) from the device.  {For simplicity, in a bipartite setting (i.e. for $N=2$), we will use the notation of inputs being $x$ and $y$ instead of $x_{1}$ and $x_{2}$, with outputs being $a$ and $b$ instead of $a_{1}$ and $a_{2}$.}

The objects of interest in these Bell experiments are the correlations observed in the generated classical data, i.e.~ the conditional probability distribution $p(a_1, \ldots, a_N | x_1, \ldots, x_N)$. Depending on the type of device that the parties use (i.e. classical, quantum, possibly post-quantum), different correlations may be feasible in the experiment. The sets of correlations that have been of main interest in the literature are the following. 

\begin{defn} \textbf{Classical correlations}, also referred to as `locally causal' \cite{bell}, are those allowing for shared random variables $\lambda\in\Lambda$, and take the form
\begin{align}
p(a_1, \ldots, a_N | x_1, \ldots, x_N) = \sum_{\lambda\in\Lambda} D^{1}_\lambda(a_1 | x_1) \ldots D^{N}_\lambda(a_N | x_N) \, p(\lambda)\,,
\end{align}
where $D^{j}_\lambda(\cdot)\in\{0,1\}$ is a deterministic response function given $\lambda$ for the $j$th party, and $p(\lambda)$ is the distribution over the variables $\lambda$ such that $\sum_{\lambda}p(\lambda)=1$. 
\end{defn}

\begin{defn} \textbf{Quantum correlations} arise if there exists a Hilbert space $\mathcal{H}$, a state $|\psi\rangle\in\mathcal{H}$, and (complete) projective measurements $\{\Pi^{(i)}_{a_i|x_i}\}_{a_i,x_i}$ for each party  $i$, such that the conditional probability distribution is given by the Born rule:
\begin{align}
p(a_1, \ldots, a_N | x_1, \ldots, x_N) = \langle\psi|\Pi^{(1)}_{a_1|x_1} \ldots \Pi^{(N)}_{a_N|x_N} \, |\psi\rangle\,,
\end{align}
and such that $\prod_{i=1}^N \Pi^{(i)}_{a_i|x_i} = \prod_{j=1}^N \Pi^{(\pi(j))}_{a_{\pi(j)}|x_{\pi(j)}}$, for any permutation $\pi$ of the parties $\{1,2,...,N\}$. 
\end{defn}

\begin{defn}
Given a bipartition $S_1 \cup S_2 = \{1, ..., N\}$ of $N$ parties $S_{1}=\{i_1,... ,i_s\}$ and $S_{2}=\{i_{s+1},..., i_N\}$, a conditional probability distribution $p(a_1, ..., a_N | x_1, ..., x_N)$ does not permit signalling across this bipartition if

\begin{eqnarray}
p(a_{i_1}, ..., a_{i_s}|x_{i_1}, ..., x_{i_s})&=&\sum_{a_{i_{s+1}},... ,a_{i_N}}p(a_1, ..., a_N | x_1, ..., x_N)\\
p(a_{i_{s+1}}, ..., a_{i_N}|x_{i_{s+1}}, ..., x_{i_N})&=&\sum_{a_{i_1},... ,a_{i_s}}p(a_1, ..., a_N | x_1, ..., x_N),
\end{eqnarray}
for all inputs $(x_{1}, ..., x_{N})$. A distribution $p(a_1, ..., a_N | x_1, ..., x_N)$ belongs to the set of \textbf{non-signalling correlations} if and only if it does not permit signalling across all bipartitions of the $N$ parties.
\end{defn}

There exist non-signalling correlations that do not have a quantum realisation \cite{PR}. A relevant set of post-quantum yet non-signalling correlations is that of the \textit{almost-quantum correlations} \cite{aqp}. It is notable that, as mentioned, the almost quantum correlations happen to comply with the physical information theoretic principles that have been proposed so far to characterise the quantum set \cite{pples}. We now present the definition of the set of almost quantum correlations.

\begin{defn} \textbf{Almost quantum correlations} arise if there exists a Hilbert space $\mathcal{H}$, a state $|\psi\rangle\in\mathcal{H}$, and (complete) projective measurements $\{\Pi^{(i)}_{a_i|x_i}\}_{a_i,x_i}$ for each party  $i$, such that the conditional probability distribution is given by:
\begin{align}\label{eq:aqcomm}
p(a_1, \ldots, a_N | x_1, \ldots, x_N) = \langle\psi|\Pi^{(1)}_{a_1|x_1} \ldots \Pi^{(N)}_{a_N|x_N} \, |\psi\rangle\,,
\end{align}
such that $\prod_{i=1}^N \Pi^{(i)}_{a_i|x_i}|\psi\rangle = \prod_{j=1}^N \Pi^{(\pi(j))}_{a_{\pi(j)}|x_{\pi(j)}}|\psi\rangle$, for any permutation $\pi$ of the $N$ parties. 
\end{defn}

\section{Einstein-Podolsky-Rosen Steering}\label{ap:eprstee}

{In analogy with the study of non-locality, in steering scenarios there are four sets of assemblages of particular interest \cite{pqsp}. }

\begin{defn}
\textbf{Local hidden state} (LHS) assemblages  (a.k.a. {\it unsteerable} assemblages) are those that take the form:
\begin{equation*}
\sigma_{a_1...a_{N}|x_1...x_{N}} = \sum_\lambda \prod_{j=1}^{N}D_\lambda(a_j|x_j) \sigma_\lambda
\end{equation*}
where $\lambda$ is a shared random variable, $D_\lambda(a_j|x_j)\in\{0,1\}$ a deterministic response function for the $j$th party, and $\sigma_\lambda\in\mathcal{L}(\mathcal{H}_{B})$ a subnormalised quantum state prepared by Bob as a function of $\lambda$, such that~ $\sum_\lambda \tr{\sigma_\lambda}=1$.
\end{defn}

\begin{defn} \textbf{Quantum assemblages} arise when all $N$ untrusted parties perform local measurements on a shared (possibly entangled) quantum state that is also shared with Bob. That is, the elements of the assemblage are 
\begin{align}
\sigma_{a_{1}...a_{N}|x_{1}...x_{N}}   =\Tr_{\rmA_{1}...\rmA_{N}} [(\Pi_{a_{1}|x_{1}}\otimes\Pi_{a_{2}|x_{2}}\otimes...\otimes\Pi_{a_{N}|x_{N}}\otimes \id_B)\rho]
\end{align}
where for each $j$th party $\sum_{a_{j}} \Pi_{a_{j}|x_{j}} = \id$ forms a complete projective measurement for each $x_{j}$, and $\rho\in\mathcal{D}(\mathcal{H}_{1}\otimes...\otimes\mathcal{H}_{N}\otimes\mathcal{H}_{B})$ is the state of the shared system between $N$ parties (the $j$th party having Hilbert space $\mathcal{H}_{j}$, for each $j \leq N$) and Bob (with Hilbert space $\mathcal{H}_{B}$).
\end{defn}

\begin{defn}
Given a bipartition $S_1 \cup S_2 = \{1, ..., N\}$ of the $N$ untrusted parties where $S_{1}=\{i_1,... ,i_s\}$ and $S_{2}=\{i_{s+1},..., i_N\}$, an assemblage $\{\sigma_{a_{1}...a_{N}|x_{1}...x_{N}}\}$ does not permit signalling across this bipartition if its elements satisfy

\begin{eqnarray}
\sigma_{a_{i_1}, ..., a_{i_s}|x_{i_1}, ..., x_{i_s}}&=&\sum_{a_{i_{s+1}},... ,a_{i_N}}\sigma_{a_1, ..., a_N | x_1, ..., x_N}\\
\sigma_{a_{i_{s+1}}, ..., a_{i_N}|x_{i_{s+1}}, ..., x_{i_N}}&=&\sum_{a_{i_1},... ,a_{i_s}}\sigma_{a_1, ..., a_N | x_1, ..., x_N},
\end{eqnarray}
for all inputs $(x_{1}, ..., x_{N})$. An assemblage $\{\sigma_{a_1, ..., a_N | x_1, ..., x_N}\}$ belongs to the set of \textbf{non-signalling assemblages} if and only if it does not permit signalling across any bipartition of the $N$ untrusted parties, and 
\begin{equation*}
\rho_{B}=\sum_{a_{1},... ,a_{N}}\sigma_{a_1, ..., a_N | x_1, ..., x_N}
\end{equation*}
for all inputs $(x_{1}, ..., x_{N})$, where $\rho_{B}\in\mathcal{D}(\mathcal{H}_{B})$ is Bob's reduced quantum state.
\end{defn}

In complete analogy with the study of non-locality, we call \textit{post-quantum assemblages} those assemblages that are non-signalling yet are not quantum, and \textit{post-quantum steering} is the demonstration that an assemblage is post-quantum. Furthermore, one can now study more specific relaxations of the set of quantum assemblages; this not only allows us to generate post-quantum assemblages, but if an assemblage does not belong to a set that is a relaxation of the quantum set, it is definitely not quantum. A relevant set is that of {\it almost quantum assemblages} \cite{pqsp}, inspired by almost quantum correlations, and defined in Ref. \cite{pqsp} in terms of a semidefinite program (cf. \cite{NPApaper}). 

Before presenting the definition, we give a bit of simplifying notation. Given a subset $S\subseteq\{1...N\}$, we define the strings $(\vec{a}_{S}|\vec{x}_{S}):=(a_{j}...a_{k}...a_{l}|x_{j}...x_{k}...x_{l})$ such that $j$, $k$, $l\in S$, i.e. strings of outputs given inputs for parties in the subset $S$. The string of all $N$ parties is denoted  $(\vec{a}|\vec{x})$. To refer to particular inputs and outputs in the string $(\vec{a}_{S}|\vec{x}_{S})$, if $j\in S$, then $[(a_{S}|x_{S})]_j:=(a_{j}|x_{j})$, $[a_{S}]_{j}:=a_{j}$, and $[x_{S}]_{j}:=x_{j}$. We then take the set of such strings, or \textit{words}, to be $W:=\{(\vec{a}_{S}|\vec{x}_{S})\}_{S}$ (which also includes the empty string $\emptyset$ for when $S$ is the empty set). In addition to this notation we also define two words $(\vec{a}_{S}|\vec{x}_{S})$ and $(\vec{a}^\prime_{S'}|\vec{x}^\prime_{S'})$,  to be \textit{orthogonal}, denoted as $(\vec{a}_{S}|\vec{x}_{S}) \perp (\vec{a}^\prime_{S'}|\vec{x}^\prime_{S'})$ if there is a $j\in S\cap S'$ such that $[x_{S}]_{j}=[x'_{S'}]_{j}$ but $[a_{S}]_{j}\neq[a'_{S'}]_{j}$. We can now present the definition of almost quantum assemblages \cite{pqsp}.

\begin{defn}\label{def:sdpaq}
An assemblage $\{\sigma_{a_{1}...a_{N}|x_{1}...x_{N}}\}$ is an \textbf{almost quantum assemblage} if for the set of words $W:=\{(\vec{a}_{S}|\vec{x}_{S}) \,:\, S\subseteq\{1...N\}\}$ with $|W|$ elements, there exists a matrix $\Gamma$ of dimension $|W| \times |W|$, whose elements are $d_{B} \times d_{B}$ matrices $\Gamma_{(\vec{a}_S|\vec{x}_S),(\vec{a}^\prime_{S'}|\vec{x}^\prime_{S'})}$ indexed by words $(\vec{a}_S|\vec{x}_S)$ and $(\vec{a}^\prime_{S'}|\vec{x}^\prime_{S'})$, such that
\begin{enumerate}
\item[(i)]\label{i} $\Gamma \geq 0$,
\item[(ii)]\label{ii} $\Gamma_{(\vec{a}_S|\vec{x}_S),(\vec{a}^\prime_{S'}|\vec{x}^\prime_{S'})} = 0\quad \text{if} \quad(\vec{a}_{S}|\vec{x}_{S}) \perp (\vec{a}^\prime_{S'}|\vec{x}^\prime_{S'})$,
\item[(iii)]\label{iii} $\Gamma_{\emptyset,\emptyset} = \sigma_R$, 
\item[(iv)]\label{iv} $\Gamma_{\emptyset,(\vec{a}_{S}|\vec{x}_{S})} = \sigma_{(\vec{a}_{S}|\vec{x}_{S})}\,, \quad \forall \quad S$,
\item[(v)]\label{v}$\Gamma_{(\vec{a}_{T}\vec{a}^\prime_{S}|\vec{x}_{T}\vec{x}^\prime_{S}), (\vec{a}_{T}\vec{a}''_{S'}|\vec{x}_{T}\vec{x}''_{S'})} = \Gamma_{(\vec{a}^\prime_{S}|\vec{x}^\prime_{S}), (\vec{a}_{T}\vec{a}''_{S'}|\vec{x}_{T}\vec{x}''_{S'})} = \Gamma_{(\vec{a}_{T}\vec{a}^\prime_{S}|\vec{x}_{T}\vec{x}^\prime_{S}), (\vec{a}''_{S'}|\vec{x}''_{S'})}\,$ for all subsets $S$, $S'$, and $T$ such that $T\cap S$ and $T\cap S'$ are both the empty set. 
\end{enumerate}
\end{defn}

While this last definition is somewhat technical, Section \ref{se:substee} presents an equivalent physical definition of the almost quantum assemblages.

\section{Proofs of the statements in Section \ref{se:channels} {and Appendix \ref{ap:channels}}}

In this section we provide the proofs of the theorems and propositions of Section \ref{se:channels} {and Appendix \ref{ap:channels}}. 

\subsection{Unitary representation of multipartite causal maps} \label{ap:multicausal}

In this section we characterise multipartite causal channels (see fig.~\ref{f:multicausalmap}). Consider an $N$-partite map that acts on $\cH_{in} = \otimes_{k=1}^N \, \cH^{k}_{in}$, and denote by $A_k$ the $N$ input systems (a.k.a. parties). The main theorem that we will prove is the following.
\begin{thm}\label{theoremunirep} Unitary representation for multipartite semicausal maps.\\
Let $\Lambda_{1... N}$ be a CP map, with $A_N \not\rightarrow ... \not\rightarrow A_2 \not\rightarrow A_1$. Then, there exist unitary operators 
$U_{A_kE}\,:\, \cH_{in}^{k}\otimes\cH_{int_{k}}^{E} \rightarrow \cH_{out}^{k}\otimes\cH_{int_{k+1}}^{E}$ for $2\leq k\leq N-1$, $U_{A_1E}\,:\, \cH_{in}^{1}\otimes\cH_{in}^{E} \rightarrow \cH_{out}^{1}\otimes\cH_{int_{1}}^{E}$, and $U_{A_NE}\,:\, \cH_{in}^{N}\otimes\cH_{int_{N}}^{E} \rightarrow \cH_{out}^{N}\otimes\cH_{out}^{E}$, with $\mathcal{H}_{int_{k}}^{E}$ being the $k$th Hilbert space of the system $E$ between unitaries $U_{A_{(k-1)}E}$ and $U_{A_kE}$
acting on the system plus an ancilla $E$ such that
\begin{equation*}
\Lambda_{1... N} [\rho] = \Tr_E \left\{ U (\rho \otimes \ket{0} \bra{0}_E) U^\dagger \right\}\,,
\end{equation*}
where $U = U_{A_NE} ... U_{A_1E}$. 
\end{thm}

For bipartite maps this reduces to the result by \cite{SW}. To prove the multipartite statement, we need the following lemma: 
\begin{lem}\label{lem:uniuni} Unitary representation for unitary semicausal maps. \cite{SW} \\
Let $\Lambda_{ABC}$ be a unitary tripartite CPTP map. If $C \not\rightarrow A$ the map can be decomposed as 
$$
\Lambda_{ABC} [\rho_{ABC}] = U \rho_{ABC} U^\dagger\,, 
$$
with $U=U_{BC}\,U_{AB}$, where $U_{BC}$ and $U_{AB}$ are unitaries.
\end{lem}

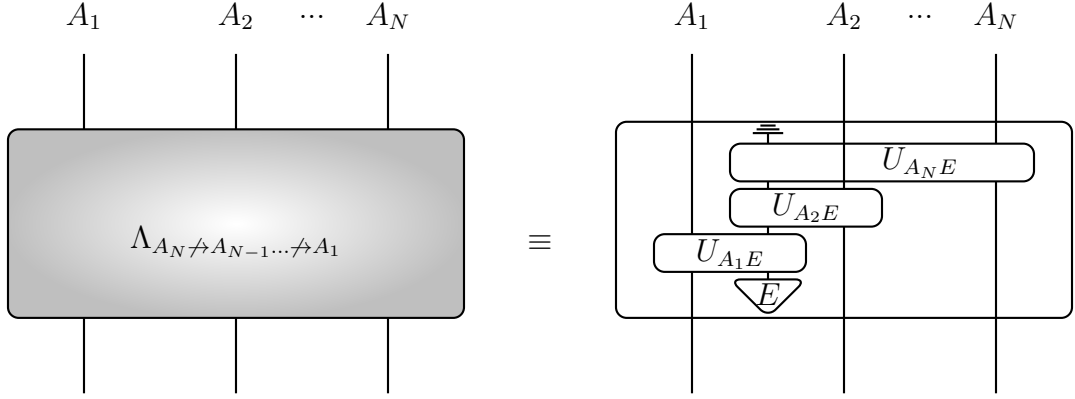
\begin{figure}
\begin{center}
\begin{tikzpicture}
\shade[draw, thick, ,rounded corners, inner color=white,outer color=gray!50!white] (-7,-1) rectangle (-1,1.5) ;
\node at (-4,0) {$\Lambda_{A_N \not\rightarrow A_{N-1} ... \not\rightarrow A_1}$};
\draw[thick] (-4,-2) -- (-4,-1);
\draw[thick] (-6,-2) -- (-6,-1);
\draw[thick] (-2,-2) -- (-2,-1);
\draw[thick] (-4,1.5) -- (-4,2.5);
\draw[thick] (-6,1.5) -- (-6,2.5);
\draw[thick] (-2,1.5) -- (-2,2.5);
\node at (-6,3) {$A_1$};
\node at (-4,3) {$A_2$};
\node at (-3,3) {$...$};
\node at (-2,3) {$A_N$};

\node at (0,0) {$\equiv$};
\draw[thick, rounded corners] (1,-1) rectangle (7,1.6);
\draw[thick, rounded corners] (3,-0.99) -- (3.5,-0.49) -- (2.5, -0.49) -- cycle;
\node at (3,-0.69) {$E$};
\draw[thick, rounded corners] (1.5,-0.39) rectangle (3.5,0.11);
\draw[thick, rounded corners] (2.5,0.21) rectangle (4.5,0.71);
\node at (2.5,-0.15) {$U_{A_1E}$};
\node at (3.5,0.46) {$U_{A_2E}$};

\draw[thick] (2,-2) -- (2,-0.39);
\draw[thick] (2,0.11) -- (2,2.5);
\draw[thick] (4,-2) -- (4,0.21);
\draw[thick] (4,0.71) -- (4,2.5);

\draw[thick] (3,-0.49) -- (3,-0.39);
\draw[thick] (3,0.11) -- (3,0.21);
\draw[thick] (3,0.71) -- (3,0.81);
\node at (5,1.06) {$U_{A_NE}$};
\draw[thick, rounded corners] (2.5,0.81) rectangle (6.5,1.31);

\draw[thick] (6,-2) -- (6,0.81);
\draw[thick] (6,1.31) -- (6,2.5);

\draw[thick] (3,1.31) -- (3,1.45);
\draw[thick] (2.8,1.45) -- (3.2,1.45);
\draw[thick] (2.85,1.5) -- (3.15,1.5);
\draw[thick] (2.9,1.55) -- (3.1,1.55);
\node at (2,3) {$A_1$};
\node at (4,3) {$A_2$};
\node at (5,3) {$...$};
\node at (6,3) {$A_N$};

\end{tikzpicture}
\end{center}
\caption{Multipartite semicausal channel, with $A_N \not\rightarrow ... \not\rightarrow A_2 \not\rightarrow A_1$.}
\label{f:multicausalmap}
\end{figure}

Now we can present the proof of the main theorem of the section. 
\begin{proof}
consider an $N$-partite map $\Lambda_{1... N}$ where $A_N \not\rightarrow A_{N-1} \not\rightarrow ... \not\rightarrow A_1$. By Lemma \ref{lem:bip}, this channel can be decomposed as: 
$$
\Lambda_{1... N} [\rho] = \Tr_E \left\{ U (\rho \otimes \ket{0} \bra{0}_E) U^\dagger \right\}\,,
$$
with $U=U_{A_NE} U_{A_{N-1} ... A_1 E}$. 

Now, the situation is the following: we have $N-1$ input systems plus an ancilla $E$ interacting via the unitary $U_{A_{N-1} ... A_1 E}$, and afterwards a unitary transformation $U_{A_NE}$ is applied to the ancilla $E$ and the last system $A_N$. The first part of this protocol, by Lemma \ref{lem:uniuni}, can be decomposed as a unitary between $A_1 ... A_{N-2}$ and $E$, followed by a unitary between $A_{N-1}$ and $E$, i.e. $U_{A_{N-1} ... A_1 E} = U_{A_{N-1}E}U_{A_{N-2} ... A_1 E}$. By applying this way Lemma \ref{lem:uniuni} recursively, one gets $U_{A_{N-1} ... A_1 E} = U_{A_{N-1}E} U_{A_{N-2} E} ... U_{A_1 E}$. It follows that the channel $\Lambda_{1 ... N}$ has a unitary decomposition with $U = U_{A_NE} U_{A_{N-1}E} ... U_{A_1 E}$. 
\end{proof}

A multipartite map is then causal if it is semicausal for all possible orderings of the parties. 

\section{Assemblages, correlations, distributed measurements, and teleportages, and the channels that define them}

\subsection{Proofs for correlations and assemblages}\label{appsec21}

In this subsection we gives proofs of Propositions \ref{assLHSprop}, \ref{prop11}, \ref{assQprop}. Essentially the same proofs apply for the Propositions \ref{probLHSprop}, \ref{probNSprop}, \ref{probQprop}, since one can take Bob's Hilbert space to be empty in a steering scenario, and for $N\geq 2$, we recover a Bell scenario.
\newline

\noindent
\textbf{Proposition \ref{assLHSprop}.}
\textit{An assemblage is a local hidden state assemblage if and only if there exists a local channel $\Lambda_{1...N,B}^{\mathsf{L}}:\mathcal{L}(\mathcal{H}_{m}^{\otimes N}\otimes\mathcal{H}_{B_{in}})\rightarrow\mathcal{L}(\mathcal{H}_{d}^{\otimes N}\otimes\mathcal{H}_{B_{out}})$ such that the assemblage is channel-defined by $\Lambda_{1...N,B}^{\mathsf{L}}$.}

\begin{proof}
First we prove that if an assemblage has a local hidden state model, then it is channel-defined by a local channel. The first thing to note is that any local hidden state assemblage can be reproduced by the $N$ untrusted parties making local measurements $\{M_{a_{j}|x_{j}}\in\mathcal{L}(\mathcal{H}_{j})\}$ (for the $j$th party) on a separable state $\rho\in\bigotimes_{j=1}^{N}\mathcal{H}_{j}\otimes\mathcal{H}_{B}$, for $\mathcal{H}_{j}$ being the $j$th party's local Hilbert space. {While some entangled states will only produce a local hidden state model assemblage, this means that there exists another separable state that can produce the same local hidden state model assemblage.} Without loss of generality we can model this measurement as a projective measurement. This choice of local measurement $\{M_{a_{j}|x_{j}}\}$ can then be simulated by preparing the input choice as the state $|x_{j}\rangle\in\mathcal{H}_{m}$ where $x_{j}\in\{1,...,m\}$. The outcome of the measurement will be translated into a register with Hilbert space $\mathcal{H}_{d}$, with outcomes $|a_{j}\rangle\in\mathcal{H}_{m}$ described by elements of an orthonormal basis; the register is initially prepared in the state $|0\rangle\in\mathcal{H}_{d}$. Therefore, each $j$th party is associated with the Hilbert space $\mathcal{H}_{j}\otimes\mathcal{H}_{d}\otimes\mathcal{H}_{m}$, and to this system we apply the unitary
\begin{equation}
U_{j}=\sum_{x_{j},a_{j}}M_{a_{j}|x_{j}}\otimes |a_{j}\rangle\langle 0|\otimes|x_{j}\rangle\langle x_{j}|.
\end{equation}
After applying this unitary, the systems in the register associated with $\mathcal{H}_{j}\otimes\mathcal{H}_{m}$ are traced out for each $j$th system, leaving the system in $\mathcal{H}_{d}$, which is then measured in the orthonormal basis $|a_{j}\rangle\in\mathcal{H}_{m}$. Therefore, the assemblage can be channel-defined by this whole local channel including the unitaries as the separable state $\rho$.

In the other direction, if given a local channel then it channel-defines a local hidden state assemblage. This should be clear since the local unitaries followed by a measurement acting on a share of a separable state, will only produce local measurements on a separable state, and therefore the assemblage has a local hidden state model.
\end{proof}

\noindent
\textbf{Proposition \ref{prop11}.}
\textit{An assemblage $\{\sigma_{a_{1}...a_{N}|x_{1}...x_{N}}\}$ is non-signalling if and only if there exists a causal channel $\Lambda_{1...N,B}^{\mathsf{C}}:\mathcal{L}(\mathcal{H}_{m}^{\otimes N}\otimes\mathcal{H}_{B_{in}})\rightarrow\mathcal{L}(\mathcal{H}_{d}^{\otimes N}\otimes\mathcal{H}_{B_{out}})$ such that the assemblage is channel-defined by $\Lambda_{1...N,B}^{\mathsf{C}}$.}

\begin{proof}
First we prove that an assemblage that is channel-defined by a causal channel is a non-signalling assemblage. This follows immediately from the definition of causal channel. Given this channel-defined assemblage $\{\sigma_{a_{1}...a_{N}|x_{1}...x_{N}}\}$, when we take a sum over outcomes $a_{j}$, then this is equivalent to tracing out the output system of a causal channel. This thus results in a new assemblage that is channel-defined by a causal channel (with fewer output systems), and thus the initial assemblage is a non-signalling assemblage. It is also straightforward to see, given the definition of a causal channel, that when tracing out the $N$ untrusted parties, we obtain a reduced quantum state for Bob that is independent of the inputs $(x_{1},...,x_{N})$.

We now proceed to the converse statement that given a non-signalling assemblage, then there exists a causal channel such that the assemblage is channel-defined by it. First, if we fix orthonormal bases for the input and output Hilbert spaces as $\{|x_{j}\rangle\}$ and $\{|a_{j}\rangle\}$ respectively, then we construct the channel with Kraus decomposition
\begin{equation}
\Gamma(\cdot)=\sum_{\textbf{a}:=(a_{1},...,a_{N}),\textbf{x}:=(x_{1},...,x_{N})}K_{\textbf{a},\textbf{x}}(\Tr_{B}\{\cdot\}\otimes\sigma_{\textbf{a}|\textbf{x}})K_{\textbf{a},\textbf{x}}^{\dagger},
\end{equation}
such that $\sigma_{\textbf{a}|\textbf{x}}\in\mathcal{D}(\mathcal{H}_{B})$, and
\begin{equation}
K_{\textbf{a},\textbf{x}}=|a_{1}\rangle\langle x_{1}|\otimes|a_{2}\rangle\langle x_{2}|\otimes...\otimes|a_{N}\rangle\langle x_{N}|\otimes\id_{B}.
\end{equation}
It can be readily verified that this is a channel, and when we prepare $|x_{1}...x_{N}\rangle\langle x_{1}...x_{N}|\otimes|0\rangle\langle 0|_{B}$ as input, act on the input with $\Gamma$, and then measure in the basis $\{|a_{1}...a_{N}\rangle\}$ on the $N$ untrusted parties, we obtain an assemblage. It remains to show that channel $\Gamma$ is itself a causal channel given a non-signalling assemblage. This can be shown inductively first tracing out the output system of party $1$ as follows:
\begin{eqnarray}
\Tr_{1}\{\Gamma(\cdot)\}&=&\sum_{a_{1}}\langle a_{1}|\left(\sum_{\textbf{a}:=(a_{1},...,a_{N}),\textbf{x}:=(x_{1},...,x_{N})}K_{\textbf{a},\textbf{x}}(\Tr_{B}\{\cdot\}\otimes\sigma_{\textbf{a}|\textbf{x}})K_{\textbf{a},\textbf{x}}^{\dagger}\right)|a_{1}\rangle\\
&=&\sum_{x_{1}}\langle x_{1}|\left(\sum_{\textbf{a}:=(a_{2},...,a_{N}),\textbf{x}:=(x_{2},...,x_{N})}K_{\textbf{a},\textbf{x}}(\Tr_{B}\{\cdot\}\otimes\sigma_{\textbf{a}|\textbf{x}})K_{\textbf{a},\textbf{x}}^{\dagger}\right)|x_{1}\rangle\\
&=&\Gamma'(\Tr_{1}\{\cdot\}),
\end{eqnarray}
where $\Gamma':\mathcal{L}(\mathcal{H}_{m}^{N-1}\otimes\mathcal{H}_{B})\rightarrow\mathcal{L}(\mathcal{H}_{d}^{N-1}\otimes\mathcal{H}_{B})$ is another channel corresponding to parties $2$ to $N$. The second line above results from the fact that the assemblages are non-signalling, and new channel $\Gamma'$ is written as
\begin{equation}
\Gamma'(\cdot)=\sum_{\textbf{a}:=(a_{2},...,a_{N}),\textbf{x}:=(x_{2},...,x_{N})}L_{\textbf{a},\textbf{x}}(\Tr_{B}\{\cdot\}\otimes\sigma_{\textbf{a}|\textbf{x}})L_{\textbf{a},\textbf{x}}^{\dagger},
\end{equation}
for
\begin{equation}
L_{\textbf{a},\textbf{x}}=|a_{2}\rangle\langle x_{2}|\otimes|a_{3}\rangle\langle x_{3}|\otimes...\otimes|a_{N}\rangle\langle x_{N}|\otimes\id_{B}.
\end{equation}
The same argument works for any party $j$, and then given the new channel $\Gamma'$, one can trace out one or more of the remaining parties' outputs to get another channel, and so on. In this way, the channel $\Gamma$ channel-defines the non-signalling assemblage, and is causal, thus concluding the proof.
\end{proof}

\noindent
\textbf{Proposition \ref{assQprop}.}
\textit{An assemblage $\{\sigma_{a_{1}...a_{N}|x_{1}...x_{N}}\}$ is quantum if and only if there exists a localizable channel $\Lambda_{1...N,B}^{\mathsf{Q}}:\mathcal{L}(\mathcal{H}_{m}^{\otimes N}\otimes\mathcal{H}_{B_{in}})\rightarrow\mathcal{L}(\mathcal{H}_{d}^{\otimes N}\otimes\mathcal{H}_{B_{out}})$ such that the assemblage is channel-defined by $\Lambda_{1...N,B}^{\mathsf{Q}}$.}

\begin{proof}
The proof of this is exactly the same as the proof for Proposition 23, except the separable state in the proof is replaced with an entangled state.
\end{proof}

As mentioned above, the proofs above easily generalise to the study of correlations. Indeed, one can run through the above arguments and just have Bob's system be the empty system, thus recovering the Bell scenario for $N\geq 2$.

\subsection{Proofs for distributed measurements and teleportages}

In this subsection we gives proofs of Propositions \ref{tellocprop}, \ref{telqprop}, \ref{telNSprop}. Essentially the same proofs apply for the Propositions \ref{buslocprop}, \ref{busqprop}, \ref{busNSprop} since, as with the connection between steering and Bell scenarios, one can take Bob's Hilbert space to be empty in a non-classical teleportation scenario, and for $N\geq 2$ we recover a Buscemi non-locality scenario.
\newline

\noindent
\textbf{Proposition \ref{tellocprop}.}
\textit{A teleportage is local if and only if there exists a local channel $\Lambda_{1...N}^{\mathsf{L}}:\mathcal{L}(\bigotimes_{j}\cH^{j}_{in}=\bigotimes_{j}\mathcal{K}_{j}\otimes\mathcal{K}_{B})\rightarrow\mathcal{L}(\bigotimes_{j=1}^{n}\mathcal{H}_{d}^{\otimes N}\otimes\mathcal{K}_{B})$ such that the teleportage is channel-defined by $\Lambda_{1...N}^{\mathsf{L}}$.}

\begin{proof}
The proof that a teleportage is channel-defined by a local channel is a local teleportage is immediate from the definitions, i.e. a local channel sequentially combined with a local measurement is again a local measurement. For the converse statement that given a local teleportage, there exists a local channel that channel-defines the teleportage, the channel is constructed by having local unitaries that `copy' the outcome of a local measurement to a local output register (with Hilbert space $\mathcal{H}_{d}$) into an orthonormal basis of this register, which is then measured in this basis.
\end{proof}

\noindent
\textbf{Proposition \ref{telqprop}.}
\textit{A teleportage is quantum if and only if there exists a localizable channel $\Lambda_{1...N}^{\mathsf{Q}}:\mathcal{L}(\bigotimes_{j}\cH^{j}_{in}=\bigotimes_{j}\mathcal{K}_{j}\otimes\mathcal{K}_{B})\rightarrow\mathcal{L}(\bigotimes_{j=1}^{n}\mathcal{H}_{d}^{\otimes N}\otimes\mathcal{K}_{B})$ such that the teleportage is channel-defined by $\Lambda_{1...N}^{\mathsf{Q}}$.}

\begin{proof}
The proof of this follows the proof of Proposition 44, except now with an entangled state instead of a separable state.
\end{proof}

\noindent
\textbf{Proposition \ref{telNSprop}.}
\textit{A teleportage is non-signalling if and only if there exists a causal channel $\Lambda_{1...N}^{\mathsf{C}}:\mathcal{L}(\bigotimes_{j}\cH^{j}_{in}=\bigotimes_{j}\mathcal{K}_{j}\otimes\mathcal{K}_{B})\rightarrow\mathcal{L}(\bigotimes_{j=1}^{n}\mathcal{H}_{d}^{\otimes N}\otimes\mathcal{K}_{B})$ such that the teleportage is channel-defined by $\Lambda_{1...N}^{\mathsf{C}}$.}

\begin{proof}
Given a teleportage that is channel-defined by a causal channel, the teleportage is non-signalling essentially by definition: taking a sum over outcomes $a_{j}$ is equivalent to tracing out the $j$th output system of the channel, resulting in a new teleportage for all systems not including $j$. For the other direction, of given a non-signalling teleportage, we can construct a causal channel that channel-defines the teleportage. First, given the elements $T_{a_{1},...,a_{N}}$ of the teleportage, since it is forms an instrument in general, we can straightforwardly construct a channel from an instrument: we introduce output registers $\mathcal{H}_{d}$ for each $j$th party and thus define a channel $\Gamma:\mathcal{L}(\bigotimes_{j}\mathcal{K}_{j}\otimes\mathcal{H}_{B})\rightarrow\mathcal{L}(\mathcal{H}_{d}^{\otimes N}\otimes\mathcal{H}_{B})$ in the following way:
\begin{equation}
\Gamma(\cdot)=\sum_{\textbf{a}:=(a_{1},...,a_{N})}|\textbf{a}\rangle\langle\textbf{a}|\otimes T_{\textbf{a}}(\Tr_{B}\{\cdot\}),
\end{equation}
with $|\textbf{a}\rangle\in\mathcal{H}_{d}^{\otimes{N}}$. Now with measurements on the register of the outputs $\mathcal{H}_{d}^{\otimes{N}}$, the teleportage is then channel-defined by $\Gamma$. It remains to show that $\Gamma$ is causal. This follows from the definition of non-signalling teleportages. 
\end{proof}

\subsection{Constructing channels from correlations and assemblages}\label{appseccon}

{Let us begin with Proposition \ref{canon}'s statement and proof.}
\newline

\noindent
\textbf{Proposition \ref{canon}.}
\textit{Given $\Lambda^{c}_{1...N}(\cdot)$ from $p(a_{1}...a_{N}|x_{1}...x_{N})$, for all measurements $M_{a'_{j}}$, and all states $\rho_{x'_{j}}$, the correlations}
\begin{equation}
p(a'_{1},a'_{2},...,a'_{N}|x'_{1},x'_2,...,x'_{N})=\Tr\left\{\bigotimes_{j=1}^{N}M_{a'_{j}}\left(\Lambda^{c}_{1...N}\otimes\id_{aux}(\bigotimes_{j=1}^{N}\rho_{x'_{j}})\right)\right\}
\end{equation}
\textit{are local if the correlations $p(a_{1}...a_{N}|x_{1}...x_{N})$ are local.}
\noindent
\begin{proof} Given a distribution $p(a'_{1}...a'_{N}|x'_{1}...x'_{N})$ of the form in eq. \eqref{gencorr} with $\Lambda_{1...N}(\cdot)$ being the channel in \eqref{canonical}, we first observe that this probability is invariant if for each state $\rho_{x'_{j}}\in\mathcal{D}(\mathcal{H}_{in}^{j}\otimes\mathcal{H}_{aux}^{j})$, we replace it with
$\rho'_{x'_{j}}=\sum_{x_{j}}(\id^{j}_{aux}\otimes|x_{j}\rangle\langle x_{j}|)\rho_{x'_{j}}(\id^{j}_{aux}\otimes|x_{j}\rangle\langle x_{j}|)$, where $\id^{j}_{aux}$ is the identity acting on $\mathcal{H}_{aux}^{j}$. Therefore, for whatever state $\rho_{x'_{j}}$, the probabilities $p(a'_{1}...a'_{N}|x'_{1}...x'_{N})$ are unchanged if we replace it with a state $\rho'_{x'_{j}}=\sum_{x_{j}}\sigma_{x_{j},x'_{j}}\otimes|x_{j}\rangle\langle x_{j}|$, where $\sigma_{x_{j},x'_{j}}\in\mathcal{D}(\mathcal{H}_{aux}^{j})$ and is equal to $\Tr_{in}(\rho_{x'_{j}}\id^{j}_{aux}\otimes|x_{j}\rangle\langle x_{j}|)$. Furthermore, we apply a similar argument to the general measurement $M_{a'_{j}}$, such that the distribution is conserved by replacing it with
\begin{equation}
M'_{a'_{j}}=\sum_{a_{j}}Q_{a_{j},a'_{j}}\otimes |a_{j}\rangle\langle a_{j}|,
\end{equation}
such that $Q_{a_{j},a'_{j}}=\Tr_{in}\left((\id^{j}_{aux}\otimes\langle a_{j}|)M_{a'_{j}}\right)$, which lives in $\mathcal{D}(\mathcal{H}_{aux}^{j})$. Taking $\rho'_{a'_{j}}$ and $\{M'_{a'_{j}}\}_{j}$ as our state and measurements, we obtain the correlations 
\begin{eqnarray}
p(a'_{1},a'_{2},...,a'_{N}|x'_{1},x'_2,...,x'_{N})&=&\sum_{a_{1}...a_{N}}\sum_{x_{1}...x_{N}}p(a_{1}...a_{N}|x_{1}...x_{N})\Tr\left\{\bigotimes_{j=1}^{N}Q_{a_{j},a'_{j}}\sigma_{x_{j},x'_{j}}\right\}\\
&=&\sum_{a_{1}...a_{N}}\sum_{x_{1}...x_{N}}p(a_{1}...a_{N}|x_{1}...x_{N})\prod_{j=1}^{N}q(a_{j'}|a_{j},x_{j},x'_{j}).
\end{eqnarray}
{Hence, the correlations $p(a'_{1},a'_{2},...,a'_{N}|x'_{1},x'_2,...,x'_{N})$ can be seen as the correlations \linebreak $p(a_{1}...a_{N}|x_{1}...x_{N})$ acted on by some local stochastic processing in the form of a conditional probability $q(a_{j'}|a_{j},x_{j},x'_{j})$. Since local stochastic processing cannot turn local correlations into non-local ones, the distribution $\{p(a'_{1},a'_{2},...,a'_{N}|x'_{1},x'_2,...,x'_{N})\}$ will be local provided that $\{p(a_{1}...a_{N}|x_{1}...x_{N})\}$ is.}
\end{proof}

We now {present} Proposition \ref{canonass}, {and its proof}.
\newline

\noindent
\textbf{Proposition \ref{canonass}.}
\textit{Given $\Sigma^{c}_{1...N,B}(\cdot)$ from assemblage elements $\sigma_{a_{1}...a_{N}|x_{1}...x_{N}}$, this channel is local-limited if for all measurements $P_{a_{B}|x_{B}}\in\mathcal{L}(\mathcal{H}_{B})$ indexed by the choice $x_{B}$ and outcomes $a_{B}$, the correlations {$p(a_{1},...,a_{N},a_{B}|x_{1},...,x_{N},x_{B}):=\Tr\{P_{a_{B}|x_{B}}\sigma_{a_{1}...a_{N}|x_{1}...x_{N}}\}$} are local.}
\noindent

\begin{proof}
Firstly, since the channel $\Sigma^{c}_{1...N,B}$ traces out the part of $\rho_{x'_{B}}$ that is input into the channel, without loss of generality, we can replace $\rho_{x'_{B}}$ with $\sigma_{0}\otimes\sigma_{x'_{B}}$, where $\sigma_{0}\in\mathcal{D}(\mathcal{H}_{B_{in}})$ is some fixed state, and $\sigma_{x'_{B}}\in\mathcal{D}(\mathcal{H}_{B_{aux}})$. Therefore, the preparation of $\sigma_{x'_{B}}$ following by the measurement $M_{a'_{B}}$ can be incorporated into a single measurement $\{P_{a'_{B}|x'_{B}}\in\mathcal{L}(\mathcal{H}_{B_{out}})\}$, such that $\sum_{a'_{B}}P_{a'_{B}|x'_{B}}=\id_{B_{out}}$. This then simplifies the correlations to be
\begin{equation}
p(a'_{1},...,a'_{N},a'_{B}|x'_{1},...,x'_{N},x'_{B})=\Tr\left\{\bigotimes_{j=1}^{N}M_{a'_{j}}\otimes P_{a'_{B}|x'_{B}}\left(\Sigma^{c}_{1...N,B}\otimes\id_{aux}(\bigotimes_{j=1}^{N}\rho_{x'_{j}}\otimes\rho_{0})\right)\right\}.
\end{equation}
By expanding out this expression and using identical reasoning to the proof of Proposition \ref{canon}, we arrive at
\begin{equation}
p(a'_{1},...,a'_{N},a'_{B}|x'_{1},...,x'_{N},x'_{B})=\sum_{a_{1}...a_{N}}\sum_{x_{1}...x_{N}}\Tr\{P_{a'_{B}|x'_{B}}\sigma_{a_{1}...a_{N}|x_{1}...x_{N}}\}\prod_{j=1}^{N}q(a_{j'}|a_{j},x_{j},x'_{j}),
\end{equation}
where $q(a_{j'}|a_{j},x_{j},x'_{j})$ is a local conditional probability. Therefore, if $\Tr\{P_{a'_{B}|x'_{B}}\sigma_{a_{1}...a_{N}|x_{1}...x_{N}}\}$ is local for all measurements $P_{a'_{B}|x'_{B}}$, the correlations are local, and the channel is local-limited.
\end{proof}

\section{Unitary representation of non-signalling assemblages and teleportages}\label{appc}

In this section we discuss the unitary representations of non-signalling assemblages and teleportages as outlined in Theorems \ref{nonsigass} and \ref{nonsigtel}. As outlined in the main text, the Gisin-Hughston-Jozsa-Wootters theorem \cite{Gisin,HJW} can also be seen as a corollary of Theorem \ref{nonsigass}, and our generalisation of the Gisin-Hughston-Jozsa-Wootters theorem is a corollary of Theorem \ref{nonsigtel}. We will only present the proof of Theorem \ref{nonsigtel} since Theorem \ref{nonsigass} is a special case. \\

\noindent
\textbf{Theorem} \ref{nonsigtel} \textbf{Unitary representation of non-signalling teleportages}\\
\textit{Let $\{T_{a_{1}...a_{N}}\}$ be a non-signalling teleportage. Then, the teleportage is channel-defined by a channel $\Lambda_{1...N,B}^{\mathsf{C}}:\mathcal{L}(\bigotimes_{j}\cH^{j}_{in}\otimes\mathcal{K}_{B})\rightarrow\mathcal{L}(\mathcal{H}_{d}^{\otimes N}\otimes\mathcal{K}_{B})$ if and only if there exist}

\begin{itemize}
\item \textit{auxiliary systems $E$ and $E'$ with input and output Hilbert spaces, $\mathcal{H}^{E}_{in}$  and $\mathcal{H}^{E}_{out}$ for $E$, with $\mathcal{H}^{E'}_{in}=\mathcal{K}_{B}$ and  $\mathcal{H}^{E'}_{out}=\mathcal{K}_{B}$ for $E'$, that is the output Hilbert space of $E'$ and $B$ coincide;}
\item \textit{quantum state $|R\rangle\in\mathcal{H}^{E}_{in}\otimes\mathcal{H}^{E'}_{in}$;}
\item \textit{unitary operator $V:\bigotimes_{j}\cH^{j}_{in}\otimes\mathcal{H}^{E}_{in}\rightarrow\mathcal{H}_{d}^{\otimes N}\otimes\mathcal{H}^{E}_{out}$,}
\end{itemize}
\textit{which produce a unitary representation of the channel $\Lambda_{1...N,B}^{\mathsf{C}}:\mathcal{L}(\bigotimes_{j}\cH^{j}_{in}\otimes\mathcal{K}_{B})\rightarrow\mathcal{L}(\mathcal{H}_{d}^{\otimes N}\otimes\mathcal{K}_{B})$ via}
\begin{equation*}
\Lambda_{1...N,B}^{\mathsf{C}}(\cdot)=\Tr_{E_{out}}\{V(\Tr_{B_{in}}(\cdot)\otimes|R\rangle\langle  R|_{E,E'})V^{\dagger}\}.
\end{equation*}
\textit{Futhermore the unitary $V$ can be decomposed into a sequence of unitaries $U_{k,E}:\mathcal{H}_m\otimes\mathcal{H}^{E}_{1}\rightarrow\mathcal{H}_d\otimes\mathcal{H}^{E}_{2}$ for appropriately chosen Hilbert spaces $\mathcal{H}^{E}_{1}$ and $\mathcal{H}^{E}_{2}$, where for any given permutation $\pi$ of the set $\{1,...,N\}$, we have that}
\begin{equation*}
V=U^{\pi}_{\pi(1),E}U^{\pi}_{\pi(2),E}...U^{\pi}_{\pi(N),E}
\end{equation*}
\textit{where $U^{\pi}_{k,E}$ is not necessarily the same as $U^{\pi'}_{k,E}$ for two different permutations $\pi$ and $\pi'$.}

\begin{proof}
First of all, we can treat all $N$ untrusted parties as a single party called 'Alice' (or $A$ for short) that produces an $N$-length string $\textbf{a}:=(a_{1},...,a_{N})$. This allows us to then consider a bipartite non-classical teleportation scenario and a causal map $\Lambda_{AB}$. First, write the map in its unitary representation with $A \not\rightarrow B$, 
\begin{center}
\begin{tikzpicture}
\shade[draw, thick, ,rounded corners, inner color=white,outer color=gray!50!white] (-5,-1) rectangle (-1,1) ;
\node at (-3,0) {$\Lambda_{A \not\leftrightarrow B}$};
\draw[thick] (-4,-1.5) -- (-4,-1);
\draw[thick] (-2,-1.5) -- (-2,-1);
\draw[thick] (-4,1) -- (-4,1.5);
\draw[thick] (-2,1) -- (-2,1.5);
\node at (-4,2.5) {Alice};
\node at (-2,2.5) {Bob};
\draw[thick, rounded corners] (-2,-2) -- (-1.5,-1.5) -- (-2.5, -1.5) -- cycle;
\node at (-2,-1.7) {$0$};
\draw[thick, rounded corners] (-4,2) -- (-3.5,1.5) -- (-4.5, 1.5) -- cycle;
\node at (-4,1.7) {$\textbf{a}$};

\node at (0,0) {$\equiv$};
\draw[thick, rounded corners] (1,-1) rectangle (5,1);
\draw[thick, rounded corners] (3,-0.99) -- (3.5,-0.49) -- (2.5, -0.49) -- cycle;
\node at (3,-0.69) {$E$};
\draw[thick, rounded corners] (1.5,0.21) rectangle (3.5,0.71);
\draw[thick, rounded corners] (2.5,-0.39) rectangle (4.5,0.11);
\node at (3.5,-0.15) {$V_{EB}$};
\node at (2.5,0.46) {$V_{AE}$};

\draw[thick] (4,-1.5) -- (4,-0.39);
\draw[thick] (4,0.11) -- (4,1.5);
\draw[thick] (2,-1.5) -- (2,0.21);
\draw[thick] (2,0.71) -- (2,1.5);

\draw[thick] (3,-0.49) -- (3,-0.39);
\draw[thick] (3,0.11) -- (3,0.21);
\draw[thick] (3,0.71) -- (3,0.85);
\draw[thick] (2.8,0.85) -- (3.2,0.85);
\draw[thick] (2.85,0.9) -- (3.15,0.9);
\draw[thick] (2.9,0.95) -- (3.1,0.95);
\node at (2,2.5) {Alice};
\node at (4,2.5) {Bob};

\draw[thick, rounded corners] (4,-2) -- (3.5,-1.5) -- (4.5, -1.5) -- cycle;
\node at (4,-1.7) {$0$};
\draw[thick, rounded corners] (2,2) -- (1.5,1.5) -- (2.5, 1.5) -- cycle;
\node at (2,1.7) {$\textbf{a}$};
\node at (5.5,0) {.};
\end{tikzpicture}
\end{center}
By defining $U_{AR} := V_{AE}$, and the new ancilla $R$ and map $U_{RB}$ as follows
\begin{center}
\begin{tikzpicture}
\draw[thick, rounded corners] (3,-0.99) -- (3.5,-0.49) -- (2.5, -0.49) -- cycle;
\node at (3,-0.69) {$R$};
\draw[thick] (3.25,-0.49)--(3.25,0);
\draw[thick] (2.75,-0.49)--(2.75,0);

\node at (4,-0.5) {$:=$};
\draw[thick, rounded corners, xshift=2cm] (3,-0.99) -- (3.5,-0.49) -- (2.5, -0.49) -- cycle;
\node[xshift=2cm] at (3,-0.69) {$E$};
\draw[thick, rounded corners, xshift=3cm] (3,-0.99) -- (3.5,-0.49) -- (2.5, -0.49) -- cycle;
\node[xshift=3cm] at (3,-0.69) {$0$};
\draw[thick, rounded corners,xshift=2cm] (2.5,-0.39) rectangle (4.5,0.11);
\node[xshift=2cm] at (3.5,-0.15) {$V_{EB}$};

\draw[thick, xshift=1.75cm] (3.25,-0.49)--(3.25,-0.39);
\draw[thick, xshift=1.75cm] (3.25,0.11)--(3.25,0.4);
\draw[thick, xshift=2.75cm] (3.25,-0.49)--(3.25,-0.39);
\draw[thick, xshift=2.75cm] (3.25,0.11)--(3.25,0.4);

\node at (7,-0.5) {;};


\draw[thick, rounded corners] (7.5,-0.75) rectangle (8.5,-0.25);
\node at (8,-0.5) {$U_{RB}$};
\draw[thick] (7.75,-1) -- (7.75,-0.75);
\draw[thick] (8.25,-1) -- (8.25,-0.75);
\draw[thick] (7.75,-0.25) -- (7.75,0);
\draw[thick] (8.25,-0.25) -- (8.25,0);

\node at (10,-0.5) {$:= \quad$ SWAP \,;};

\end{tikzpicture}
\end{center}
we obtain
\begin{center}
\begin{tikzpicture}
\draw[thick, rounded corners] (1,-1) rectangle (5,1);
\draw[thick, rounded corners] (3,-0.99) -- (3.5,-0.49) -- (2.5, -0.49) -- cycle;
\node at (3,-0.69) {$E$};
\draw[thick, rounded corners] (1.5,0.21) rectangle (3.5,0.71);
\draw[thick, rounded corners] (2.5,-0.39) rectangle (4.5,0.11);
\node at (3.5,-0.15) {$V_{EB}$};
\node at (2.5,0.46) {$V_{AE}$};

\draw[thick] (4,-1.5) -- (4,-0.39);
\draw[thick] (4,0.11) -- (4,1.5);
\draw[thick] (2,-1.5) -- (2,0.21);
\draw[thick] (2,0.71) -- (2,1.5);

\draw[thick] (3,-0.49) -- (3,-0.39);
\draw[thick] (3,0.11) -- (3,0.21);
\draw[thick] (3,0.71) -- (3,0.85);
\draw[thick] (2.8,0.85) -- (3.2,0.85);
\draw[thick] (2.85,0.9) -- (3.15,0.9);
\draw[thick] (2.9,0.95) -- (3.1,0.95);
\node at (2,2.5) {Alice};
\node at (4,2.5) {Bob};

\draw[thick, rounded corners] (4,-2) -- (3.5,-1.5) -- (4.5, -1.5) -- cycle;
\node at (4,-1.7) {$0$};
\draw[thick, rounded corners] (2,2) -- (1.5,1.5) -- (2.5, 1.5) -- cycle;
\node at (2,1.7) {$\textbf{a}$};

\node[xshift=5.5cm] at (0,0) {$\equiv$};
\draw[xshift=5cm, thick, rounded corners] (1,-1) rectangle (5,1);
\draw[thick, rounded corners, xshift=5cm] (3,-0.99) -- (3.5,-0.49) -- (2.5, -0.49) -- cycle;
\node[xshift=5cm] at (3,-0.69) {$R$};
\draw[thick, rounded corners,xshift=5cm] (1.5,0.21) rectangle (3.5,0.71);
\node[xshift=5cm] at (2.5,0.46) {$V_{AE}$};

\draw[thick,xshift=5cm] (4,-1.5) -- (4,-1);
\draw[thick,xshift=5cm] (4,1) -- (4,1.5);
\draw[thick,xshift=5cm] (2,-1.5) -- (2,0.21);
\draw[thick,xshift=5cm] (2,0.71) -- (2,1.5);

\draw[thick,xshift=5cm] (3,-0.49) -- (3,0.21);
\draw[thick,xshift=5cm] (3,0.71) -- (3,0.85);
\draw[thick,xshift=5cm] (2.8,0.85) -- (3.2,0.85);
\draw[thick,xshift=5cm] (2.85,0.9) -- (3.15,0.9);
\draw[thick,xshift=5cm] (2.9,0.95) -- (3.1,0.95);
\node[xshift=5cm] at (2,2.5) {Alice};
\node[xshift=5cm] at (4,2.5) {Bob};

\draw[thick, rounded corners,xshift=5cm] (4,-2) -- (3.5,-1.5) -- (4.5, -1.5) -- cycle;
\node[xshift=5cm] at (4,-1.7) {$0$};
\draw[thick, rounded corners,xshift=5cm] (2,2) -- (1.5,1.5) -- (2.5, 1.5) -- cycle;
\node[xshift=5cm] at (2,1.7) {$\textbf{a}$};

\draw[thick] (9,-1) -- (9,-0.8);
\draw[thick, xshift=6cm, yshift=-1.65cm] (2.9,0.99) -- (3.1,0.99);
\draw[thick, xshift=6cm, yshift=-1.65cm] (2.85,0.92) -- (3.15,0.92);
\draw[thick, xshift=6cm, yshift=-1.65cm] (2.8,0.85) -- (3.2,0.85);
\draw[thick] (8.3,-0.49) to [out=90,in=-90] (9,1);

\node at (10.5,0) {$\equiv$};

\draw[thick, rounded corners, xshift=10cm] (1,-1) rectangle (5,1);
\draw[thick, rounded corners,xshift=10cm] (3,-0.95) -- (3.9,-0.45) -- (2.1, -0.45) -- cycle;
\node[xshift=10cm] at (3,-0.66) {$R$};

\draw[thick, rounded corners,xshift=10cm] (1.5,-0.1) rectangle (2.9,0.4);
\draw[thick, rounded corners,xshift=10cm] (3.1,-0.1) rectangle (4.5,0.4);
\node[xshift=10cm] at (2.3,0.1) {$U_{AR}$};
\node[xshift=10cm] at (3.7,0.1) {$U_{RB}$};

\draw[thick,xshift=10cm] (2,-1.5) -- (2,-0.1);
\draw[thick,xshift=10cm] (2,0.4) -- (2,1.5);
\draw[thick,xshift=10cm] (4,-1.5) -- (4,-0.1);
\draw[thick,xshift=10cm] (4,0.4) -- (4,1.5);

\draw[thick,xshift=10cm] (2.5,-0.45) -- (2.5,-0.1);
\draw[thick,xshift=10cm] (3.5,-0.45) -- (3.5,-0.1);

\draw[thick,xshift=10cm] (2.5,0.4) -- (2.5,0.6);
\draw[thick, xshift=9.5cm, yshift=-0.25cm] (2.9,0.99) -- (3.1,0.99);
\draw[thick, xshift=9.5cm, yshift=-0.25cm] (2.85,0.92) -- (3.15,0.92);
\draw[thick, xshift=9.5cm, yshift=-0.25cm] (2.8,0.85) -- (3.2,0.85);

\draw[thick,xshift=10cm] (3.5,0.4) -- (3.5,0.6);
\draw[thick, xshift=10.5cm, yshift=-0.25cm] (2.9,0.99) -- (3.1,0.99);
\draw[thick, xshift=10.5cm, yshift=-0.25cm] (2.85,0.92) -- (3.15,0.92);
\draw[thick, xshift=10.5cm, yshift=-0.25cm] (2.8,0.85) -- (3.2,0.85);

\node[xshift=10cm] at (2,2.5) {Alice};
\node[xshift=10cm] at (4,2.5) {Bob};
\draw[thick, rounded corners,xshift=10cm] (4,-2) -- (3.5,-1.5) -- (4.5, -1.5) -- cycle;
\node[xshift=10cm] at (4,-1.7) {$0$};
\draw[thick, rounded corners,xshift=10cm] (2,2) -- (1.5,1.5) -- (2.5, 1.5) -- cycle;
\node[xshift=10cm] at (2,1.7) {$\textbf{a}$};
\node at (15.5,0) {.};

\end{tikzpicture}
\end{center}
This first equivalence then gives part of the statement of the theorem where $V$ is the unitary $V_{AE}$, and the state on $\mathcal{K}_{B}$ is the state of the right-hand system of $|R\rangle$. The second equivalence shows that the whole channel can be seen as a localizable channel across the bipartition between Alice and Bob. 

In order to prove the remainder of the theorem, it remains to decompose $V$ into a sequence of unitaries. To do this, we use the fact that among the $N$ parties (that make up Alice) there cannot be any signalling. Therefore, using Theorem \ref{theoremunirep} we arrive at the full statement of the theorem.
\end{proof}

\section{Almost quantum assemblages}\label{ap:AQCSDP}

Here we provide the proofs of Lemma \ref{lemmaAQ} and Theorem \ref{stalmost}, respectively. 

\begin{lemma*}
An assemblage $\{\sigma_{a_{1}...a_{N}|x_{1}...x_{N}}\}$ is almost quantum if and only if there exists a Hilbert space $\mathcal{H}\cong\mathcal{K}\otimes\mathcal{H}_{B}$, quantum state $|\psi\rangle\in\mathcal{H}$, and projective measurements $\{\Pi_{a_{j}|x_{j}}\in\mathcal{L}(\mathcal{K})\}$ for each $j$th party where $\sum_{a_{j}}\Pi_{a_{j}|x_{j}}=\id$ and for all permutations $\pi$ of $\{1,...,N\}$, $\prod_{j=1}^{N}\Pi_{a_{\pi(j)}|x_{\pi(j)}}|\psi\rangle=\prod_{j=1}^{N}\Pi_{a_{j}|x_{j}}|\psi\rangle$, such that
\begin{equation}
\sigma_{a_{1}...a_{N}|x_{1}...x_{N}}=\Tr_{\mathcal{K}}\{\prod_{j=1}^{N}\Pi_{a_{j}|x_{j}}\otimes\id_{{B}}|\psi\rangle\langle\psi|\}.
\end{equation}
\end{lemma*}

\begin{proof}
Consider a steering scenario where $N$ parties steer one, by performing $m$ measurements of $d$ outcomes each. Let $D$ be the dimension of the Hilbert space of the characterised party, denoted by $\mathcal{H}_B$. 

We will first prove the `if' direction.
 
Take an assemblage $\{\sigma_{\vec{a}|\vec{x}}\}$, generated by the uncharacterised parties performing measurements $\Pi_{a_k|x_k}^{(k)}$ on the state $\rho$, which without loss of generality we can consider to be a pure state $\rho = \ket{\psi}\bra{\psi}$. Let $W$ be the set of words as in Def.~\ref{def:sdpaq}. Now define the matrix $\Gamma$, with size $|W|\times |W|$ and whose elements are $D\times D $ matrices, as follows:  
\begin{align*}
&\Gamma_{(\vec{a}_{S}|\vec{x}_{S}),(\vec{a}^\prime_{S'}|\vec{x}^\prime_{S'})} := \Tr_{1, ..., N} \left\{ \prod_{i=1:|S|} \Pi_{a_i|x_i}^{\dagger} \prod_{j=1:|S'|} \Pi_{a_j^\prime|x_j^\prime} \, \rho \right\}\,,\\
&\Gamma_{\emptyset,(\vec{a}^\prime_{S'}|\vec{x}^\prime_{S'})} := \Tr_{1, ..., N} \left\{ \prod_{j=1:size(x^\prime)} \Pi_{a_j^\prime|x_j^\prime} \, \rho \right\}\,,\\
&\Gamma_{(\vec{a}_{S}|\vec{x}_{S}),\emptyset} := \Tr_{1, ..., N} \left\{ \prod_{i=1:size(x)} \Pi_{a_i|x_i}^{\dagger} \, \rho \right\}\,,\\
&\Gamma_{\emptyset,\emptyset} := \sigma_R\,.
\end{align*}

By definition, this $\Gamma$ satisfies (iii) and (iv) of Def.~\ref{def:sdpaq}. Since the measurements $\Pi_{a_k|x_k}$ are projective, $\Gamma$ also satisifies condition (ii). The commutation relation of the projective measurements $\Pi_{a_k|x_k}$ on the state $\rho$ further imply (v). 

We still need to show that $\Gamma\geq 0$. We hence need to show that $\Gamma$ as an element of $M_{|W|} (M_D)$ (i.e. of the set of $|W| \times |W|$ matrices whose entries are $D\times D$ matrices) is positive semidefinite, which is equivalent to showing that $\Gamma$ as an element of $M_{|W|D}$ (i.e. of the set of $|W|D \times |W|D$ matrices whose entries are complex numbers) is positive semidefinite \cite{paulsen}. 
We will hence show that $\Gamma \in M_{|W|D}$ is a Gramian matrix {(i.e.~it can be written as $\Gamma=V^{\dagger}V$ for some matrix $V$),} and since all Gramian matrices are positive semidefinite the claim follows. 

First observe that entries of $\Gamma$ are of the form 
$$
\Gamma_{[(\vec{a}_{S}|\vec{x}_{S})]_j,[(\vec{a}^\prime_{S'}|\vec{x}^\prime_{S'})]_l}=\bra{j}{\Gamma}_{(\vec{a}_{S}|\vec{x}_{S}),(\vec{a}^\prime_{S'}|\vec{x}^\prime_{S'})}\ket{l}
$$ 
where $[(\vec{a}_{S}|\vec{x}_{S})]_j,[(\vec{a}^\prime_{S'}|\vec{x}^\prime_{S'})]_l$ denotes the $(j,l)$ component of the matrix ${\Gamma}_{(\vec{a}_{S}|\vec{x}_{S}),(\vec{a}^\prime_{S'}|\vec{x}^\prime_{S'})}$, with 
 ${\Gamma}_{(\vec{a}_{S}|\vec{x}_{S}),(\vec{a}^\prime_{S'}|\vec{x}^\prime_{S'})}\in\{\Gamma_{u,v} \,:\, u,v \in W\}$, and $\ket{j}$, $\ket{l}$ elements of an orthonormal basis of $\mathcal{H}_{B}$. 

By cyclicity of the partial trace we can also write 
$$
{\Gamma}_{(\vec{a}_{S}|\vec{x}_{S}),(\vec{a}^\prime_{S'}|\vec{x}^\prime_{S'})}=\textrm{tr}_{1 ... N}(F_{(\vec{a}^\prime_{S'}|\vec{x}^\prime_{S'})}\ket{\psi}\bra{\psi}G^{\dagger}_{(\vec{a}_{S}|\vec{x}_{S})})
$$ 
for $F_{(\vec{a}^\prime_{S'}|\vec{x}^\prime_{S'})} , G_{(\vec{a}_{S}|\vec{x}_{S})}\in\{\id_{A}\} \cup \{\Pi_{v} \,:\,v \in S\}$, where we have identified $\Pi_{(\vec{a}_{S}|\vec{x}_{S})} := \prod_{i=1:|S|} \Pi_{a_i|x_i}$. 

By defining $\mathcal{H}_{A}$ as the Hilbert space of the uncharacterised parties, note that
\begin{align*}
\bra{j}{\Gamma}_{(\vec{a}_{S}|\vec{x}_{S}),(\vec{a}^\prime_{S'}|\vec{x}^\prime_{S'})}\ket{l}
&=\sum_{\ket{y}\in\mathcal{H}_{A}}\bra{j}\bra{y}F_{(\vec{a}^\prime_{S'}|\vec{x}^\prime_{S'})}\ket{\psi}\bra{\psi}G^{\dagger}_{(\vec{a}_{S}|\vec{x}_{S})}\ket{y}\ket{l}\\
&=\sum_{\ket{y}\in\mathcal{H}_{A}}\bra{\psi}G^{\dagger}_{(\vec{a}_{S}|\vec{x}_{S})}\ket{y}\ket{l}\bra{j}\bra{y}F_{(\vec{a}^\prime_{S'}|\vec{x}^\prime_{S'})}\ket{\psi}\\
&=\left(\sum_{\ket{y'}\in\mathcal{H}_{A}}\bra{\psi}G^{\dagger}_{(\vec{a}_{S}|\vec{x}_{S})}\ket{y'}\ket{l}\bra{y'}\right)\left(\sum_{\ket{y}\in\mathcal{H}_{A}}\bra{j}\bra{y}F_{(\vec{a}^\prime_{S'}|\vec{x}^\prime_{S'})}\ket{\psi}\ket{y}\right)\\
&=\sum_{y'}\alpha_{y',[(\vec{a}_{S}|\vec{x}_{S})]_l }^{*}\bra{y'}\sum_{y}\alpha_{y,[(\vec{a}^\prime_{S'}|\vec{x}^\prime_{S'})]_j} \ket{y}\\
&=\bra{{u}}_{[(\vec{a}_{S}|\vec{x}_{S})]_l}\ket{v}_{[(\vec{a}^\prime_{S'}|\vec{x}^\prime_{S'})]_j}
\end{align*}
where $\{\ket{y}\}$ is an orthonormal basis in $\mathcal{H}_{A}$ such that $\langle y'\ket{y}=\delta^{y}_{y'}$ and $\alpha_{y,[(\vec{a}^\prime_{S'}|\vec{x}^\prime_{S'})]_j}=\bra{j}\bra{y}F{(\vec{a}^\prime_{S'}|\vec{x}^\prime_{S'})}\ket{\psi}$ is some scalar.
{Now we can further define the set of vectors
\begin{align*}
\left\{ \ket{v}_{[(\vec{a}^\prime_{S'}|\vec{x}^\prime_{S'})]_j} \,:\, (\vec{a}^\prime_{S'}|\vec{x}^\prime_{S'}) \in W\,,\, j=1\ldots D \,,\, \ket{v}_{[(\vec{a}^\prime_{S'}|\vec{x}^\prime_{S'})]_j} =  \sum_{y}\alpha_{y,[(\vec{a}^\prime_{S'}|\vec{x}^\prime_{S'})]_j} \ket{y}  \right\},
\end{align*}
And define the matrix $V$ as that whose columns are each one of these vectors. We see then that $\Gamma=V^{\dagger}V$, i.e.~the elements of $\Gamma$ are all the inner product of vectors associated with a row and column. Therefore, $\Gamma$ is Gramian. }

This proves the first part of the claim: an assemblage that arises from performing those types of measurements on a quantum state are almost quantum assemblages.

For the converse, take an almost-quantum assemblage $\{\sigma_{\vec{a}|\vec{x}}\}$. Let $\Gamma \in M_{|W|} (M_D)$ be its moment matrix from Def.~\ref{def:sdpaq}. Let {$\{ \ket{i} \,:\, i \in W \}$} be the Gram vectors of that matrix, i.e. $\Gamma_{i,j} = \langle i | j \rangle$. Note that these vectors have as entries elements of $M_D$ (i.e. $D\times D$ matrices on complex numbers). For each uncharacterised party $k$, define the subspace $V^{(k)}_{a_{k}|x_{k}} := \text{span}\{\ket{(a_{k}\vec{a}'|x_{k}\vec{x}')} \,:\, (\vec{a}'|\vec{x}') \in S_{\bar{k}} \}$, where $S_{\bar{k}} \subset W $ is the subset of words that do not involve party $k$. Note for clarity we have dropped the index indicating the subset of the parties. Now define:
\begin{align}
&E^{(k)}_{a_{k}|x_{k}} := \text{proj} (V^{(k)}_{a_{k}|x_{k}})\,, \quad \forall \, x=1:m\,,\quad \forall \, a=1:d-1\,,\\
&E^{(k)}_{d|x_{k}} := \id - \sum_{a=1}^{d-1} E^{(k)}_{a_{k}|x_{k}} \,.
\end{align}
By definition, $E^{(k)}_{a_{k}|x_{k}}$ are projection operators. Condition (ii) of Def.~\ref{def:sdpaq} implies that $E^{(k)}_{a_{k}|x_{k}}E^{(k)}_{a'_{k}|x_{k}}=0$ if $a \neq a'$. Hence, $\{E^{(k)}_{a_{k}|x_{k}}\}_a$ defines a complete projective measurement for each $x_{k}$ for each party. 

Now we will see the action of any sequence $\prod_{k=1}^N E^{(k)}_{a_k|x_k}$ on $\ket{\emptyset}$. Let us start with just one projector: 
\begin{align*}
E^{(k)}_{a_k|x_k} \ket{\emptyset} &=  E^{(k)}_{a_k|x_k} \ket{(a_k|x_k)} + E^{(k)}_{a_k|x_k} \left( \ket{\emptyset} - \ket{(a_k|x_k)} \right) \\
&= \ket{(a_k|x_k)} \,
\end{align*}
since by definition $E^{(k)}_{a_k|x_k} \ket{(a_k|x_k)} = \ket{(a_k|x_k)}$ and condition (v) implies that 
$$
\langle (a_k\vec{a}'|x_k\vec{x}'), \emptyset \rangle = \langle (a_k\vec{a}'|x_k\vec{x}'), (a_k|x_k) \rangle
$$ 
for all $(\vec{a}'|\vec{x}') \in S_{\bar{k}}$.

The same reasoning can be applied to $E^{(k)}_{a|x}$ acting on an arbitrary $\ket{v}$ with $v \in S_{\bar{k}}$:
\begin{align}\label{eq:applytogen}
E^{(k)}_{a_k|x_k} \ket{v} &=  E^{(k)}_{a_k|x_k} \ket{(a_k|x_k) \, v} + E^{(k)}_{a_k|x_k} \left( \ket{v} - \ket{(a_k|x_k) \, v} \right) \\ \nonumber
&= \ket{(a_k|x_k) \, v} \,
\end{align}
since $\ket{(a_k|x_k) \, v} \in V^{(k)}_{a_k|x_k}$ and condition (v) implies $ \langle (a_k|x_k) \, v , v \rangle =\langle (a_k|x_k) \, v , (a_k|x_k) \, v \rangle$.

Therefore,
\begin{align}\label{eq:theaux}
\sigma_{\vec{a}_{S}|\vec{x}_{S}} = \bra{\emptyset} \prod_{k=1:n} E^{(k)}_{a_k|x_k} \ket{\emptyset} = \tr{\prod_{k=1:n} E^{(k)}_{a_k|x_k} \, \ket{\emptyset}\bra{\emptyset} }.
\end{align}

This motivates the following definition: $\rho := \ket{\emptyset}\bra{\emptyset}$. 

The next ingredient is to check that commutation relations of the projective measurements $\{ E^{(k)}_{a_k|x_k}\}$ on the state $\rho$ satisfy the conditions stated in the Lemma. This follows immediately from \eqref{eq:applytogen}.

To conclude the proof, notice that the trace in eq.~\eqref{eq:theaux} can actually be interpreted as a partial trace on the Hilbert space of the uncharacterised parties $\mathcal{H}_A$. This follows from understanding the set $M_{|S|}(M_D)$  as  the tensor product algebra $M_{|S|} \otimes M_D$ \cite{paulsen}. 
First, invoke the isomorphism between projection operators $E \in M_{|S|}( M_D )$ and projection operators $E \otimes \mathds{1}_D$ with $E \in M_{|S|}$ \footnote{ For the reader who is not familiar with operator algebras, the main idea of the identification goes as follows. Let $E$ be a projection operator onto a linear subspace of a vector space of dimension $|S|$ on the space of matrices $M_D$. $E$ is hence an element of $M_{|S|}( M_D )$. Now perform a change of basis of $M_{|S|}( M_D )$ so that $E$ is diagonal. That is, in this new basis $\{\ket{\phi}\}$, either $E\ket{\phi} = \ket{\phi}$ or $E\ket{\phi}=0$. Remember that these $\ket{\phi}$ are vectors with entries in $M_D$.
Hence, $\bra{\phi}E\ket{\phi} = \mathds{1}_D$ or $\bra{\phi}E\ket{\phi} = \mathds{O}_D$. It follows that we can think of $E$ as $E= E_{\mathcal{H}_A} \otimes \mathds{1}_D$, where $ E_{\mathcal{H}_A}$ is a projector in $M_{|S|}$. } . 
Then, notice that an orthonormal basis $\{\ket{\phi}\}$ for $M_{|S|}( M_D )$ can be seen as $\{\ket{\phi} \bra{\phi} = \ket{\varphi}\bra{\varphi} \otimes \mathds{1}_D\}$, where $\{ \ket{\varphi}\}$ in an orthonormal basis for $\mathcal{H}_A$. 

\end{proof}

\begin{thm*} 
An assemblage $\{\sigma_{a_{1}...a_{N}|x_{1}...x_{N}}\}$ is almost quantum if and only if there exists an almost localizable channel $\Lambda_{1...N,B}^{\mathsf{\tilde{Q}}}:\mathcal{L}(\mathcal{H}_{m}^{\otimes N}\otimes\mathcal{H}_{B_{in}})\rightarrow\mathcal{L}(\mathcal{H}_{d}^{\otimes N}\otimes\mathcal{H}_{B_{out}})$ such that the assemblage is channel-defined by $\Lambda_{1...N,B}^{\mathsf{\tilde{Q}}}$.
\end{thm*}

\begin{proof}
Given an assemblage channel-defined by an almost localizable channel, it is immediate that is an almost quantum assemblage due to Lemma 29. To show that an assemblage as defined in Lemma 29 can be channel-defined by an almost localizable channel, we can use exactly the same constructive argument as in Proposition 17. That is, given projectors as in Lemma 29, we can construct local unitaries that act on a register in the state $|\psi\rangle$ as in the proof of Proposition 17.
\end{proof}

\end{document}